\documentclass[11pt, onecolumn]{IEEEtran} 

\usepackage{amsmath}
\usepackage{amssymb}
\usepackage{amsfonts}
\usepackage{epsfig}
\usepackage{subfigure}
\usepackage{calc}
\usepackage{color}
\usepackage[all]{xy}
\usepackage{bm}
\usepackage{algorithmicx}
\usepackage{algpseudocode}
\usepackage{algorithm}
\usepackage{enumerate}
\usepackage{hyperref}

\newtheorem{defn}{Definition}
\newtheorem{thm}{Theorem}[section]
\newtheorem{cor}[thm]{Corollary}
\newtheorem{prop}{Proposition}

\newtheorem{lem}[thm]{Lemma}
\newtheorem{conj}[thm]{Conjecture}
\newtheorem{constr}[thm]{Construction}
\newtheorem{note}{Remark}
\newtheorem{example}{Example}
\newcommand{\bit}{\begin{itemize}}
\newcommand{\eit}{\end{itemize}}
\newcommand{\bcor}{\begin{cor}}
\newcommand{\ecor}{\end{cor}}
\newcommand{\beq}{\begin{equation}}
\newcommand{\eeq}{\end{equation}}
\newcommand{\beqn}{\begin{equation*}}
\newcommand{\eeqn}{\end{equation*}}
\newcommand{\bea}{\begin{eqnarray}}
\newcommand{\eea}{\end{eqnarray}}
\newcommand{\bean}{\begin{eqnarray*}}
\newcommand{\eean}{\end{eqnarray*}}
\newcommand{\ben}{\begin{enumerate}}
\newcommand{\een}{\end{enumerate}}
\newcommand{\bdefn}{\begin{defn}}
\newcommand{\edefn}{\end{defn}}
\newcommand{\bnote}{\begin{note}}
\newcommand{\enote}{\end{note}}
\newcommand{\bprop}{\begin{prop}}
\newcommand{\eprop}{\end{prop}}
\newcommand{\blem}{\begin{lem}}
\newcommand{\elem}{\end{lem}}
\newcommand{\bthm}{\begin{thm}}
\newcommand{\ethm}{\end{thm}}
\newcommand{\bconj}{\begin{conj}}
\newcommand{\econj}{\end{conj}}
\newcommand{\bconstr}{\begin{constr}}
\newcommand{\econstr}{\end{constr}}
\newcommand{\bpf}{\begin{proof}}
\newcommand{\epf}{\end{proof}}

\begin{document}

\title{Codes with Local Regeneration}
 \author{Govinda M. Kamath, N. Prakash, V. Lalitha and P. Vijay Kumar
\thanks{Govinda M. Kamath,  N. Prakash, V. Lalitha and P. Vijay Kumar are with the Department of ECE, Indian Institute of Science, Bangalore,
560 012 India (email: \{govinda,  prakashn, lalitha, vijay\}@ece.iisc.ernet.in).}
\thanks{The results in this paper were presented in part at the 2012 IEEE International Symposium on Information Theory \cite{PraKamLalKum}, NSF Workshop on Frontiers in Stochastic Systems, Networks and Control, Texas A\&M University, College Station, TX, October 27, 2012 as well as the Workshop on Trends in Coding Theory, Ascona, Switzerland, October 28-November 2, 2012.}
\thanks{This research is supported in part by the National Science Foundation under Grant 0964507 and in part by the NetApp
Faculty Fellowship program. The work of V. Lalitha is supported by a TCS Research Scholarship. }
\thanks{A part of work in Section \ref{sec:scalar_local} of this paper has appeared in an
earlier arXiv submission, see \cite{PraKamLalKum_arxiv}.}
}
\date{\today}
\maketitle

\begin{abstract}
Regenerating codes and codes with locality are two schemes that have
recently been proposed to ensure data collection and reliability in a distributed storage network.  In a situation where one
is attempting to repair a failed node, regenerating codes seek to minimize the amount of data downloaded for node repair,
while codes with locality attempt to minimize the number of helper nodes accessed.  In this paper, we provide several
constructions for a class of vector codes with locality in which the local codes are regenerating codes, that enjoy both
advantages.   We derive an upper bound on the minimum distance of this class of codes and show that the proposed
constructions achieve this bound. The constructions include both the cases where the local regenerating codes correspond to
the MSR as well as the MBR point on the storage-repair-bandwidth tradeoff curve of regenerating codes.   Also included is a performance comparison of various code constructions for fixed block length and minimum distance.
\end{abstract}

\section{Introduction}\label{sec:intro}

Apart from ensuring reliability, the principal goals in a distributed storage network relate to data collection and node
repair.  We will seek architectures which store the data across $n$ nodes in such a way that a data collector can recover the
data by connecting to a small number $k$ of nodes in the network.     Node repair will be accomplished by connecting to a
subset of $d$ nodes and downloading a uniform amount of data from each node for a total download of $W$.  Here $W$ is termed
the repair bandwidth and it is of interest to minimize both $W$ as well as the repair degree, defined as the number $d$ of
nodes accessed during repair. It is also desirable to have multiple options for both data collection and node repair in terms
of the set of $k$ or $d$ nodes that one connects to.

Distributed storage systems found in practice, include Windows Azure Storage \cite{HuaSimXu_etal_azure} and the Hadoop-based
systems~\cite{hadoop} used in Facebook and Yahoo. In Facebook data centers, a $[14,10]$ maximum-distance separable (MDS)
code is used in a coding scheme referred to as HDFS RAID~\cite{hdfs_raid}. Here data can be downloaded by connecting to any
$10$ nodes. The coding scheme is however, inefficient in terms of node repair, as the repair degree as well as repair
bandwidth both equal $10$. Regenerating codes~\cite{DimGodWuWaiRam} and codes with locality~\cite{GopHuaSimYek} are two
alternative approaches proposed to address the situation.

Two alternative approaches to coding have recently been advocated to enable more efficient node repair, namely, regenerating codes~\cite{DimGodWuWaiRam} and codes with locality~\cite{GopHuaSimYek}.

\subsection{Regenerating Codes} \label{sec:intro_regen}

In the regenerating-code framework, there are $n$ nodes in the network, with each node storing $\alpha$ code symbols drawn from a finite field $\mathbb{F}_q$.  A data collector should be able to download the data by connecting to any $k$ nodes (see Fig. \ref{fig:regen_framework}.).   Node repair is required to be accomplished by connecting to any $d$ nodes and downloading $\beta \leq \alpha$ symbols from each node.  Thus the repair bandwidth is given by $d\beta$.  A regenerating code may be regarded as a vector code, i.e., a code of block length $n$ over the vector alphabet $\mathbb{F}_q^{\alpha}$.  The parameter set of a regenerating code will be listed in one of two forms: $((n,k,d), (\alpha,\beta),B)$ if the file size or number of message symbols $B$ is known and relevant and $((n,k,d), (\alpha,\beta))$ otherwise.

\begin{figure}[h]
  \centering
  \subfigure[Data Collection]{\label{fig:regen_data_collection}\includegraphics[height=2.5in]{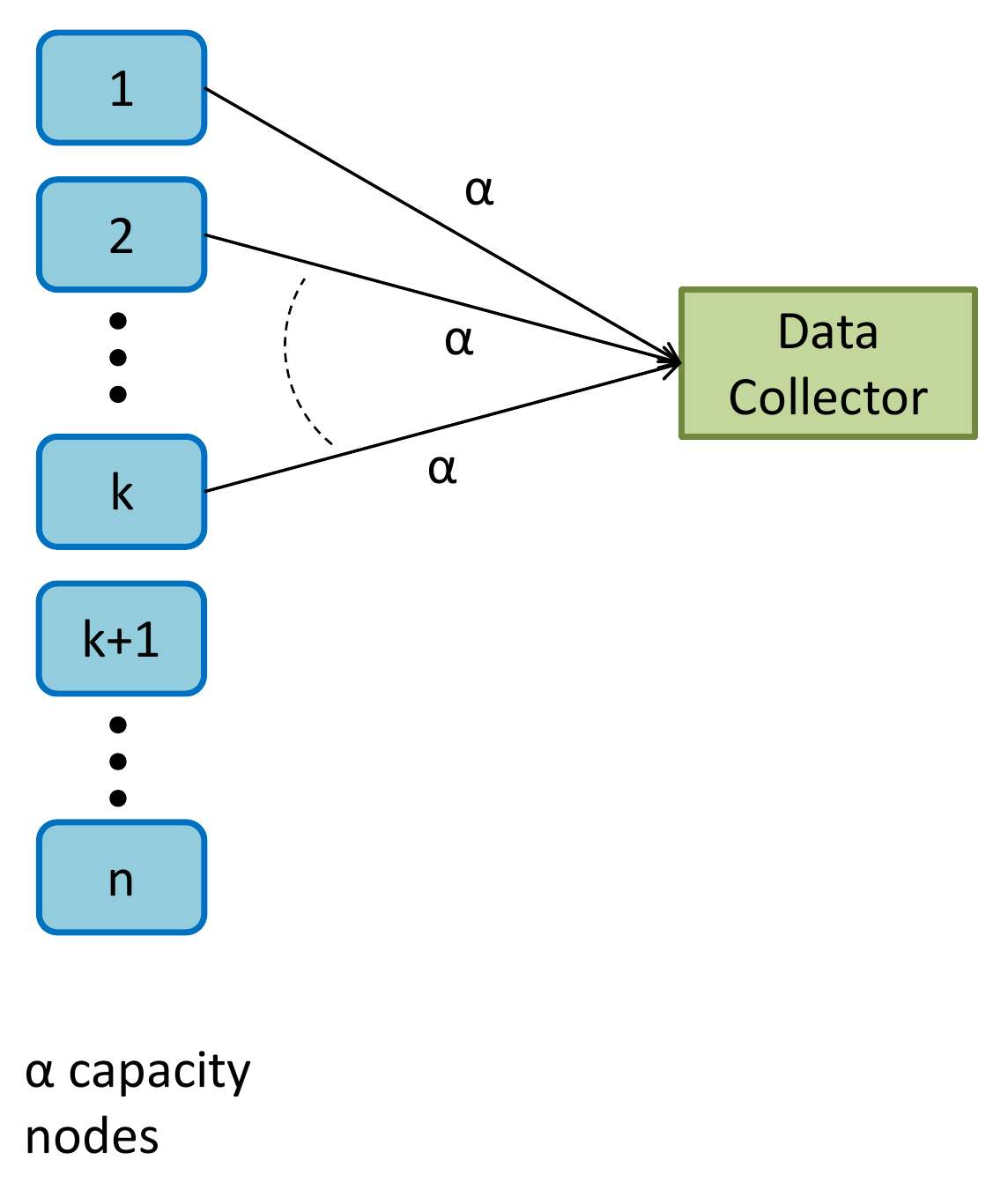}}
  \hspace{0.5in}
  \subfigure[Node Repair]{\label{fig:regen_node_repair}\includegraphics[height=2.5in]{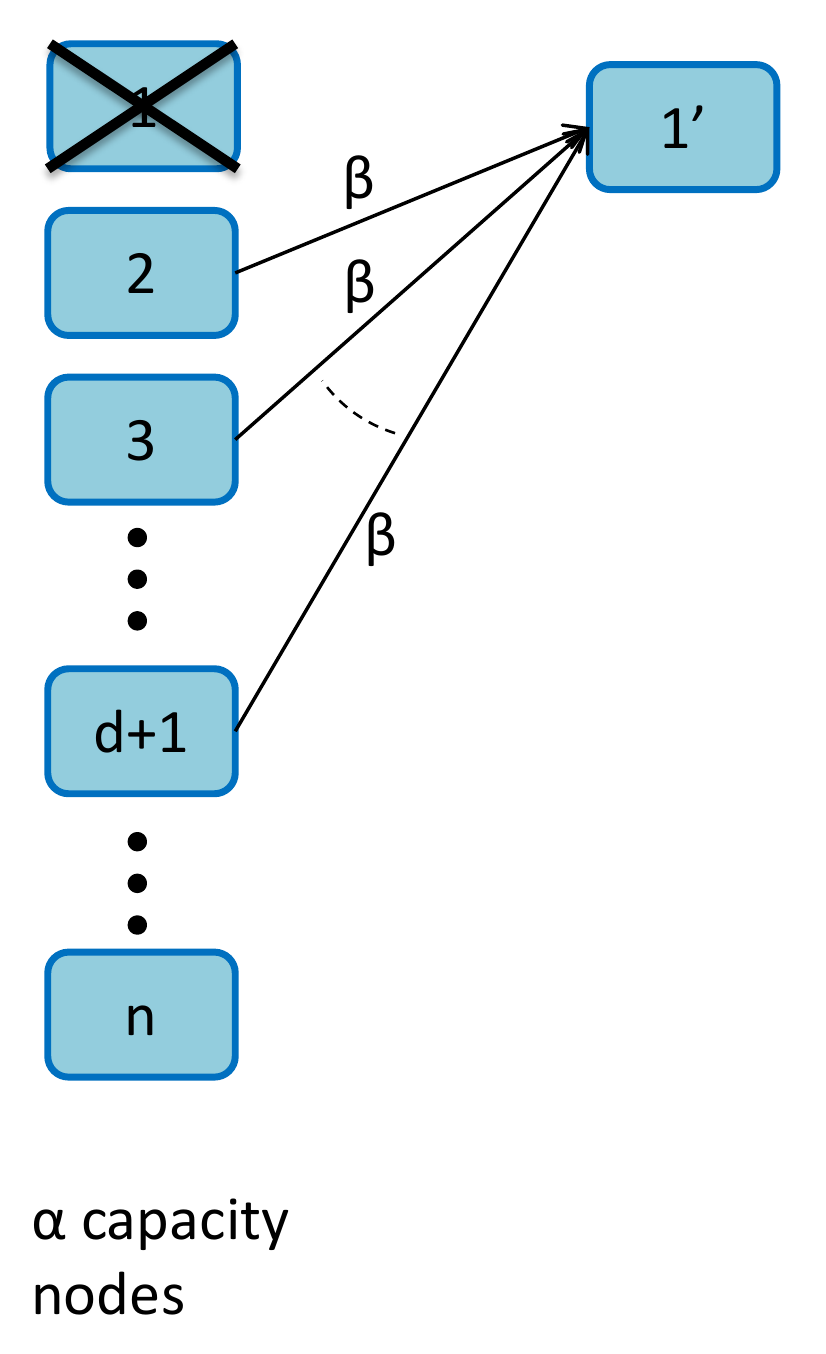}}
  \caption{The Regenerating Code Framework.}
  \label{fig:regen_framework}
\end{figure}

A cut-set bound based on network-coding concepts, tells us that given code parameters $((n,k,d), (\alpha,\beta),B)$ the size $B$ of the data file is upper bounded~\cite{DimGodWuWaiRam} by
\bea \label{eq:cut_set_bd}
B & \leq & \sum_{i=0}^{k-1} \min\{\alpha,(d-i)\beta\} .
\eea
A regenerating code is considered as being optimal if
\ben
\item the file size $B$ satisfies \eqref{eq:cut_set_bd}
\item the bound is violated if either $\alpha$ or $\beta$ is reduced.
\een

Given the file size $B$ as well as regenerating-code parameters $(k,d)$, there are multiple pairs $(\alpha,\beta)$ that satisfy \eqref{eq:cut_set_bd}.  This leads to the storage-repair-bandwidth trade-off shown in Fig.~\ref{fig:trade-off}.   The two extremal points in the trade-off are the Minimum Storage Regeneration (MSR) and Minimum Bandwidth Regeneration (MBR) points. At the MSR point, we have  $\alpha = \frac{B}{k} = (d - k + 1)\beta$ and at the MBR point, $\alpha = d\beta$. The remaining points on the trade-off curve will be referred to as interior points.  A regenerating code is said to be exact if the replacement of a failed node stores the same data as did the failed node, and functional otherwise.

\begin{figure}[h]
  \centering
  \subfigure[Storage-repair-bandwidth trade-off for fixed values of $B=7500, k=10, d=12$.]{\label{fig:trade-off}\includegraphics[height=2.8in]{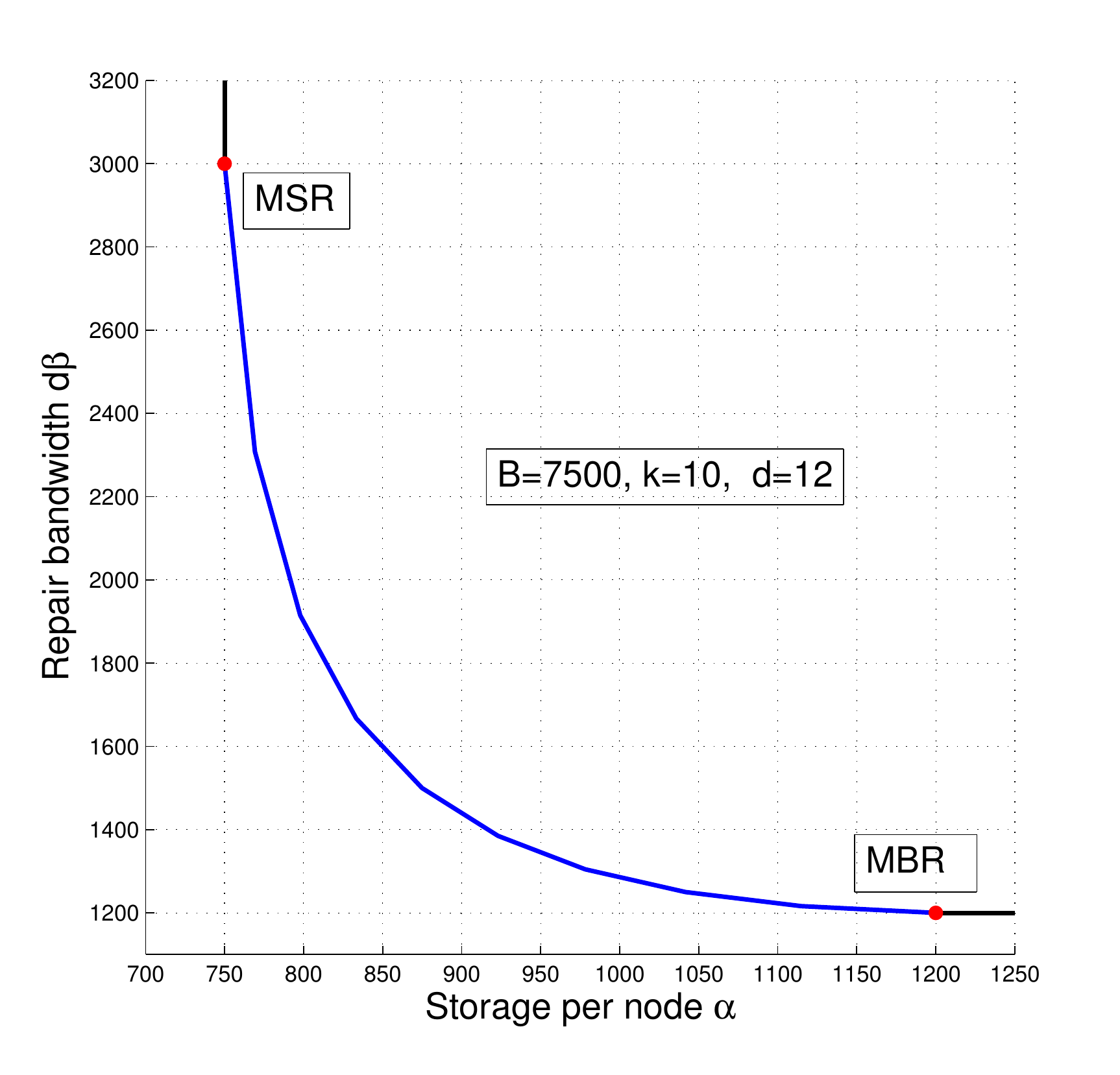}}
  \hspace{0.5in}
  \subfigure[Pictorial depiction of the repair-by-transfer MBR code in Eg~\ref{eg:pentagon}.]{\label{fig:pentagon}\includegraphics[height=2in]{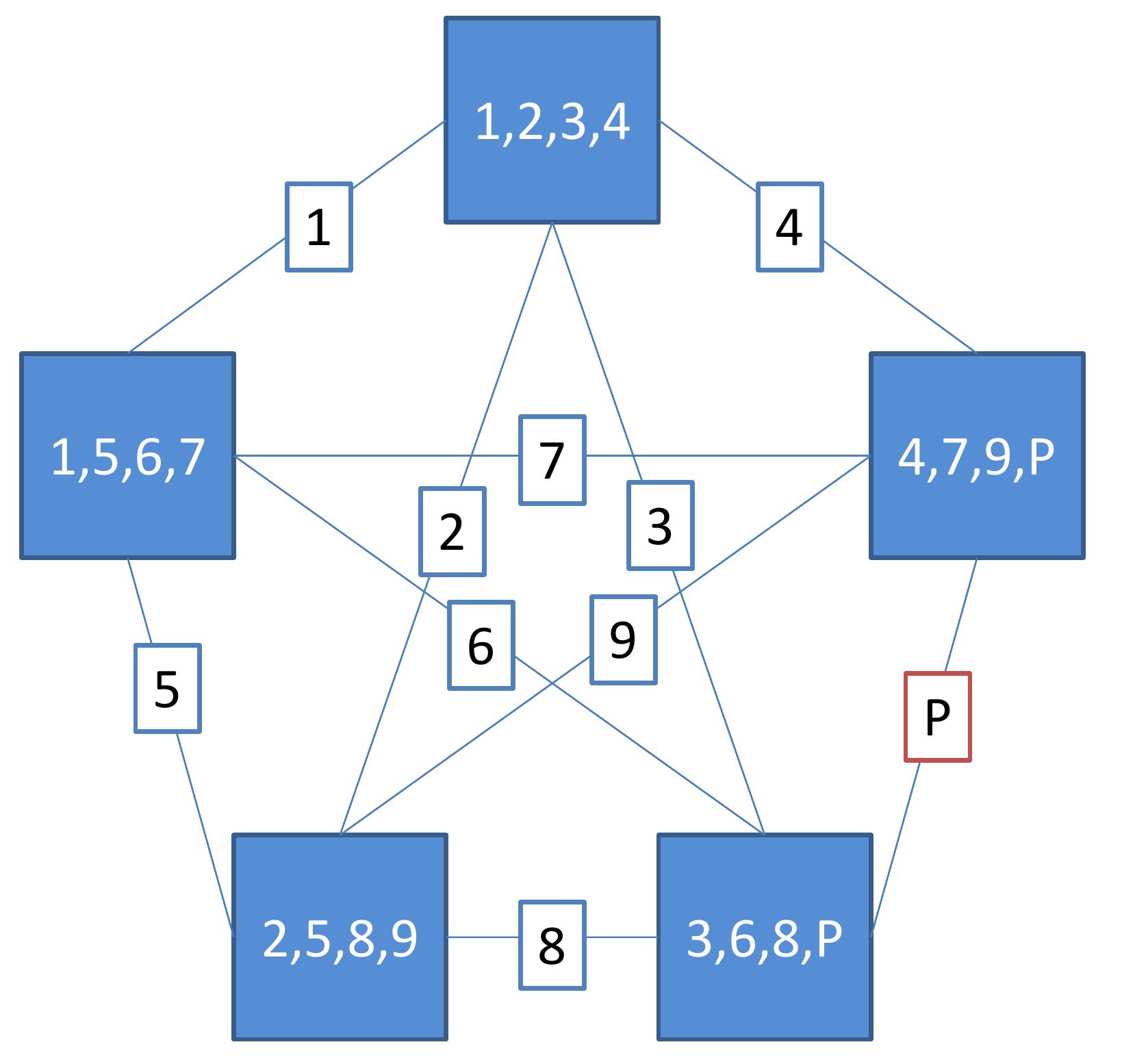}}
  \caption{Storage-Bandwidth Tradeoff and an example construction.}
 \end{figure}

\vspace{0.1in}

\subsubsection{Regenerating Code Constructions} \label{sec:regen_code_constns}

It has been shown in \cite{ShaRasKumRam_rbt} that the interior points on the trade-off are not achievable using exact-repair
regenerating codes. We summarize below the constructions known in literature for the MSR and MBR points. Except where
otherwise noted, all results described below, pertain to exact-repair regenerating codes.
\ben[a)]
\item MBR Point: There are two principal families of MBR codes:
\ben [(i)]
\item The repair-by-transfer family discussed in Example~\ref{eg:pentagon}, \item MBR codes constructed using the product
matrix construction, see~\cite{RasShaKum_pm}.  This construction can be used to generate MBR codes for any value of code
parameters
\bean
\left((n,k,d),(\alpha=d\beta,\beta=1),B= dk - {k \choose 2}\right).
\eean
\een

 \item MSR Point: At the MSR point, we have $\alpha=(d-k+1)\beta$. There are several families of MSR codes:
\ben[(i)]
\item MSR codes constructed using the product-matrix construction, see~\cite{RasShaKum_pm}.  This construction can be used to
generate MSR codes for any value of code parameters
\bean
((n,k,d \geq 2k-2),(\alpha,\beta=1),B= k\alpha ).
\eean
\item MSR codes with parameters
\bean
((n,k,d = n-1\geq 2k-1),(\alpha,\beta=1),B= k\alpha ),
\eean
described in \cite{ShaRasKumRam_ia} and \cite{SuhRam}.
 \item The Hadamard-design-based construction \cite{PapDimCad} of high-rate MSR codes with parameters
 \bean
((n,k=n-2,d = n-1),(\alpha,\beta=2^k),B= k\alpha ).
\eean
\item The Zigzag code construction \cite{TamWanBru} of high-rate MSR codes with parameters
 \bean
((n,k=n-m,d = n-1),(\alpha,\beta=m^{k-1}),B= k\alpha ),
\eean
that are guaranteed to only repair systematic nodes.
\item An explicit, functional-repair MSR code with parameters
\bean
((n,k,d = k+1),(\alpha=2,\beta=1),B= 2\alpha ),
\eean
can be found in \cite{ShaRasKumRam_ia}.
\item Apart from these explicit constructions, the existence of MSR codes for all $(n,k,d), \ n > d \geq k,$ is shown in
\cite{CadJafMalRamSuh}.
\een
\een

\vspace{0.1in}

\subsubsection{Other Work Related to Regenerating Codes}

\ben[a)]

\item Fractional repetition codes, a framework studied in \cite{ElrRam}, is related to the repair-by-transfer MBR code
discussed above.   Under this framework, node repair is required to be carried out without any computations, i.e, by mere
transfer of data.  The requirement on node repair is relaxed in the sense that, one needs to be able to recover from failure
of a node by connecting to any one of several subsets of $d$ nodes rather than by connecting to any $d$ nodes.
\item The framework of cooperative regenerating codes where multiple node repairs are carried out simultaneously and in a
cooperative manner has been studied in \cite{ShuHu}. A cut-set based bound is derived and two explicit class of constructions
are presented there.
\een

Studies on implementation and performance evaluation of regenerating codes in distributed storage settings can be found in
\cite{HuYuLiLeeLui, HuCheLeeTan, DumBie}.

\vspace{0.1in}

An example construction of a regenerating code taken from \cite{ShaRasKumRam_rbt}, is given below.
\vspace{0.1in}

\begin{example} \label{eg:pentagon}
In the example (see Fig.~\ref{fig:pentagon}), the regenerating code has parameters $((n=5, k=3,d=4), (\alpha=4,\beta=1),B=9)$. The collection of $B=9$ message symbols are first encoded using a $[10,9,2]$ MDS code of block length $10$.  Each code symbol is then placed on a distinct edge of a fully-connected graph with $5$ nodes.  The code symbols stored in a node are the symbols associated to edges incident on the particular node.  It follows that every pair of nodes share exactly one code symbol.  A data collector connects to $k=3$ nodes and thus has access to $\alpha k - {k \choose 2}=12-3=9$ distinct code symbols of the MDS code and can hence decode the message symbols.

Node repair is easily accomplished by the simple means of symbol transfer. Thus the replacement of a failed node simply receives from each of the neighbors of the failed node, the symbol the two nodes share in common.  The code can be verified to achieve the upper bound in  \eqref{eq:cut_set_bd} corresponding to the MBR point, i.e., corresponding to $\alpha=d\beta$ and for this reason, these codes are referred to as repair-by-transfer MBR codes (RBT-MBR).
The construction generalizes to any parameter set of the form $((n,k,d=n-1),(\alpha=n-1,\beta=1))$ and the file size $B$ is then given by
\bean
B & = & dk-{k \choose 2},
\eean
and can be shown to achieve the cut-set bound at the MBR point.
\end{example}

\vspace{0.1in}

\subsection{Codes with Locality}

In ~\cite{GopHuaSimYek}, Gopalan et al introduced the interesting notion of locality of information. This was also in part, motivated by applications to distributed storage, where the aim was to design codes in such a way that the number of remaining nodes accessed to repair a failed node is much smaller than the block length of the code.  The $i$th code-symbol $c_i$, $1 \leq i \leq n$, of an $[n,k,d]$ linear code $\mathcal{C}$ over the field $\mathbb{F}_q$ is said to have locality $r$ if this symbol can be recovered by accessing at most $r$ other code symbols of code $\mathcal{C}$. Equivalently, for any coordinate $i$, there exists a row in the parity-check matrix of the code of Hamming weight at most $r + 1$, whose support includes $i$. An $(r, d)$ code was defined as a systematic linear code $\mathcal{C}$ having minimum distance $d$, where all $k$ message symbols have locality $r$.  It was shown that the minimum distance of an $(r, d)$ code is upper bounded by
\bean \label{eq:gopalan_bound}
d & \leq &  n- k -\left\lceil \frac{k}{r} \right\rceil + 2.
\eean
A class of codes constructed earlier and known as pyramid codes \cite{HuaCheLi} are shown to be $(r, d)$ codes that are optimal with respect to this bound.  The structure of an optimal code is deduced for the case when $r|k$ and $d<r+3$ and it is shown that the local codes must necessarily be MDS and support disjoint.  The paper also introduces the notion of all-symbol locality in which all the code symbols, not just the message symbols have locality $r$.  The existence of all-symbol locality was established for the case when $(r+1) |n$.

\vspace{0.1in}

\subsubsection{Other Work on Codes with Locality}

A class of codes with locality known as {\em Homomorphic Self-Repairing Codes}, that
makes use of linearized polynomials, was introduced in an earlier work by the authors of \cite{OggDat}. These codes have
all-symbol locality and an example provided in \cite{OggDat} turns out to be optimal with respect to the bound in
\eqref{eq:gopalan_bound}.  A general construction of explicit and optimal codes with all-symbol locality is provided in
\cite{SilRawVis} that is based on Gabidullin maximum rank-distance codes.

Locality in vector codes is considered in \cite{PapDim}. The authors derive an upper bound on the minimum
distance of a vector code under the assumption of all-symbol locality and also provide an explicit construction of a class of
codes which achieve the bound for certain code-parameter sets. This construction is related to an earlier
construction of codes with locality (see \cite{PapLuoDimHuaLi}), involving the same authors.

The notion of locality in scalar codes was subsequently extended by the authors of the present paper in \cite{PraKamLalKum},
to the case when the local codes have minimum distance greater than $2$.  An analogous bound on minimum distance and code
constructions are provided and these results are described in detail in Section~\ref{sec:scalar_local}.  An earlier,
parity-splitting construction appearing in \cite{HanLas} turns out to provide an example of such an extension of the notion
of locality. A similar construction was subsequently presented in \cite{BlaHafHet} in the context of solid-state storage
drives.  These results will be revisited in Section~\ref{sec:scalar_local}.    The results in \cite{RawKoySilVis} are
described in Section~\ref{sec:overview_results}.

Studies on implementation and performance evaluation of codes with locality in distributed storage settings can be found in
\cite{HuaSimXu_etal_azure, sathiamoorthy}. In \cite{HuaSimXu_etal_azure}, a class of code termed as local reconstruction code
and related to the pyramid code has been employed in a distributed storage code solution known as Windows Azure Storage, see
Fig.~\ref{fig:azure}.  This code has block length $16$ and by puncturing the code in two coordinates $P_1$, $P_2$, one will
obtain a code that is the concatenation of two support disjoint single-parity-check $[7,6,2]$ MDS codes (corresponding to
code symbols labeled using $X$ and $Y$ respectively), which provide locality. The two global parity symbols $P_1,P_2$ ensure
that the minimum distance of the overall code equals $4$.
\begin{center}
\hspace*{2.0in}
\begin{figure} [h!]
\begin{center}
\includegraphics[width=4.0in]{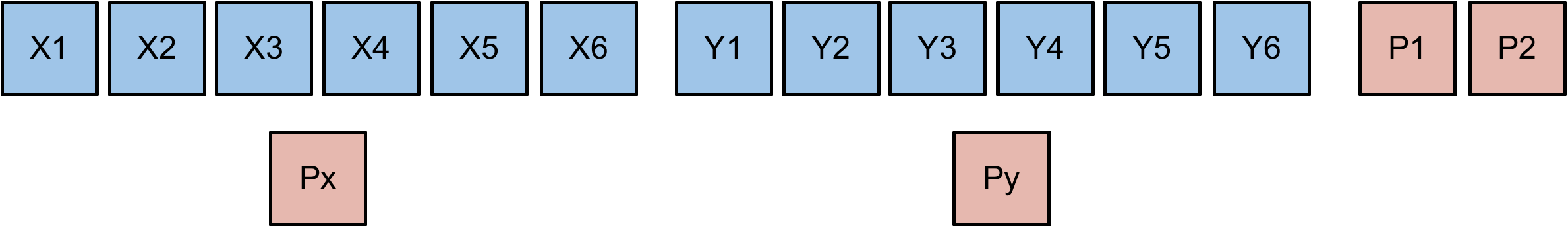}
\end{center}
\caption{The pyramid code employed in Windows Azure Storage.}
\label{fig:azure}
\end{figure}
\end{center}
In \cite{sathiamoorthy}, the authors discuss implementation of a class of codes with locality (called locally repairable
codes) in Hadoop Distributed File System and compare the performance with Reed Solomon codes.

\subsection{Array Codes}

Regenerating codes are examples of vector codes, by which we mean codes over a vector alphabet, $\mathbb{F}_q^{m}$ for some
integer $m$.   In the case of regenerating codes, $m=\alpha$.    Any vector code may also be regarded as an array code in
which each codeword corresponds to an array of size $(m \times n)$.  A survey of array codes can be found in
\cite{BlaFarTil}.   Array codes have found extensive application in storage systems and examples include the EVENODD code
constructed in \cite{BlaBraBruMen} and later extended in \cite{BlaBraBruMenVar} as well as the Row-Diagonal Parity code
presented in \cite{CorEngGoeGrcKleLeoSan}.
Section~\ref{sec:vec_code_prelims}.

Section~\ref{sec:overview_results} of the paper provides a brief overview of the results of the present paper.  A comparison
of some coding options for distributed storage also appears here.
The extended notion of scalar locality is discussed in Section~\ref{sec:scalar_local}.  Section~\ref{sec:vec_code_prelims}
introduces vector codes and Section~\ref{sec:locality_vec_codes} discusses locality in the context of vector codes and
provides bounds on minimum distance and code size assuming the local codes to be identical.   Locality in the context of
vector codes permits one to consider codes in which the local codes are regenerating codes. Optimal constructions of vector
codes with locality, where the local codes are MSR and MBR codes are presented in Section~\ref{sec:msr_local_codes}
and~\ref{sec:mbr_local_codes} respectively. In Section~\ref{sec:kappa_bound}, additional bounds on minimum distance are
derived that take into account the particular structure of the code and which do not require the local codes to be identical.
Most proofs are relegated to the Appendix.

\section{Overview of Results} \label{sec:overview_results}

\subsection{Results in Summary}

In terms of coding options for distributed storage, regenerating codes aim to minimize the download bandwidth during node
repair, whereas, codes with locality seek to reduce the number of helper nodes contacted.  This raises the question as to
whether it is possible to design codes that combine the desirable features of both classes of codes, i.e.,
construct codes with locality, in which the local codes are regenerating codes.  The present paper
answers this in the affirmative.  We term such codes as codes with local regeneration or equivalently, local
regenerating codes.  We develop bounds on the minimum distance of local regenerating codes as well as several constructions
of codes that achieve these bounds with equality and are hence, optimal.

In an independent and parallel work\footnote{Both papers were presented at the Workshop on {\em Trends in Coding
Theory}, Ascona, Oct. 29-Nov. 2, 2012.}, the authors of \cite{RawKoySilVis} also consider codes with all-symbol locality
where the local codes are regenerating codes. Bounds on minimum distance are provided and a construction for optimal codes
with MSR all-symbol locality based on rank-distance codes are presented.

We now briefly state the various results contained in this paper.

\bit
\item {\em Extension of Notion of Scalar Locality} The paper begins by extending the
notion of locality in scalar codes, where we allow the local codes to be more general codes, instead of just single parity
check codes (and hence can have local minimum distance $\delta > 2$). An upper bound to the minimum distance is derived and
the structure of a code that achieves this bound is derived for the case when the dimension of the local code divides the
dimension of the overall or global code.  It is shown that pyramid codes achieve the upper bound on minimum distance.   The
existence of optimal codes with all-symbol locality when the local code length divides the global code length is also shown.
The explicit construction of codes with all-symbol locality contained in \cite{HanLas} called the parity-splitting
construction is also presented and shown to be optimal for certain parameter sets. It is noted that concatenated codes are
examples of codes with all-symbol locality and this is used to obtain a new upper bound on the minimum distance of a
concatenated code. We note that most of the results on the extension of scalar locality have appeared in \cite{PraKamLalKum}.

\item {\em Vector Codes} The discussion of codes with local regeneration necessitates a discussion of vector codes of which
they are an example. As such, some basic observations about codes possessing a vector alphabet are made here and it is shown
that exact-repair MBR and MSR regenerating codes naturally fall into a particular class of vector codes, which we term as
uniform rank-accumulation (URA) codes.

\item {\em Vector Codes with Locality} This is followed by an extension of the notion of locality to vector codes and a bound on the minimum distance is derived for the case when the local codes have identical parameters and belong to the class of URA codes.  The structure of the code is determined under additional assumptions.  These assumptions hold for the case when (a) the local codes are MBR codes and (b) the local codes are MSR codes and the scalar dimension of the local code divides the scalar dimension of the global code.  The scalar dimension of a code over the vector alphabet $\mathbb{F}_q^{\alpha}$ is its dimension as a vector space over $\mathbb{F}_q$.

\item {\em Codes with Local Regeneration} We then provide several constructions for the class of codes with local
regeneration, which are optimal with respect to the upper bound on the minimum distance. The constructions include both the
cases the local codes belong to the MSR and the MBR family of regenerating codes.

\item {\em Bounds on Minimum Distance of a General Vector Code with Locality} Finally, we also provide additional bounds on
minimum distance that take into account the particular structure of the vector code and which do not require the
local codes to have identical parameters.
\eit

\vspace{0.1in}
A summary of the bounds on minimum distance derived in this paper is given in Table~\ref{tab:dmin_bounds}. An overview of
various constructions (appearing in this paper) of codes with local regeneration is presented next.

\begin{table}
 \caption{Bounds on Minimum Distance Appearing in the Paper} \label{tab:dmin_bounds}
\centering
\begin{tabular}{||c|c|c||} \hline
 \hline
  & & \\
  Theorem & Bound & Comments \\
  &  &  \\
  \hline
  \hline
  && \\
 Theorem \ref{thm:scalar_info_locality}   & $d_{\min} \ \leq \ n - k + 1 - \left(\left\lceil{\frac{k}{r}}\right\rceil - 1\right)(\delta - 1)$  & Bound for scalar codes with $(r,\delta)$ information locality  \\
  && \\
    \hline
    && \\
  Theorem \ref{thm:URA_bound} & $d_{\min}  \ \leq \  n-P^{(\text{inv})}(K)+1$ &    Bound for vector codes  with exact
$(r,\delta)$ information locality  \\
  URA Bound &&  with URA local codes\\
    && \\
    \hline
      && \\
     Theorem \ref{thm:kappa_bound} & $d_{\text{min}} \  \leq \ n - |\mathcal{I}_0| + 1 - \left(\left \lceil \frac{|\mathcal{I}_0|}{r}\right \rceil - 1\right)(\delta - 1)$  & Bound for vector codes  with $(r,\delta)$ information locality \\
     $\mathcal{I}_0$ bound && \\
      && \\
      \hline
      \hline
\end{tabular}

\end{table}

\subsection{Overview of Constructions of Codes with Local Regeneration} \label{sec:constructions}

The constructions presented in this section are optimal with respect to the bound on minimum distance of a code with local URA codes as well as the bound on scalar dimension presented in Theorem \ref{thm:URA_bound} of Section~\ref{sec:locality_vec_codes}.
\ben[(a)]
\item {\em Sum-Parity MSR-Local Code:}   The construction is illustrated in Fig.~\ref{fig:sum_parity}.   The construction begins with a parent MSR code whose generator matrix is of the form $[I \mid P_1 \mid P_2]$ and which is moreover, such that the punctured code having generator matrix $[I \mid P_1]$ is also an MSR code. The codewords in the constructed local regenerating code are then of the form
\bean
[\mathbf{m}_a^t \mid  \mathbf{m}_a^tP_1 \mid \mathbf{m}_b^t \mid  \mathbf{m}_b^tP_1 \mid  (\mathbf{m}_a+\mathbf{m}_b)^t P_2],
\eean
where $\mathbf{m}_a, \mathbf{m}_b$ are the message vectors associated with the two constituent, local regenerating codes.
This construction turns out to yield optimal codes regardless of the number of constituent local codes, provided that the global minimum distance $d_{\min}$ does not exceed twice the local minimum distance $\delta$.

\begin{center}
\begin{figure}[h!]
\begin{center}
\includegraphics[width=3.3in]{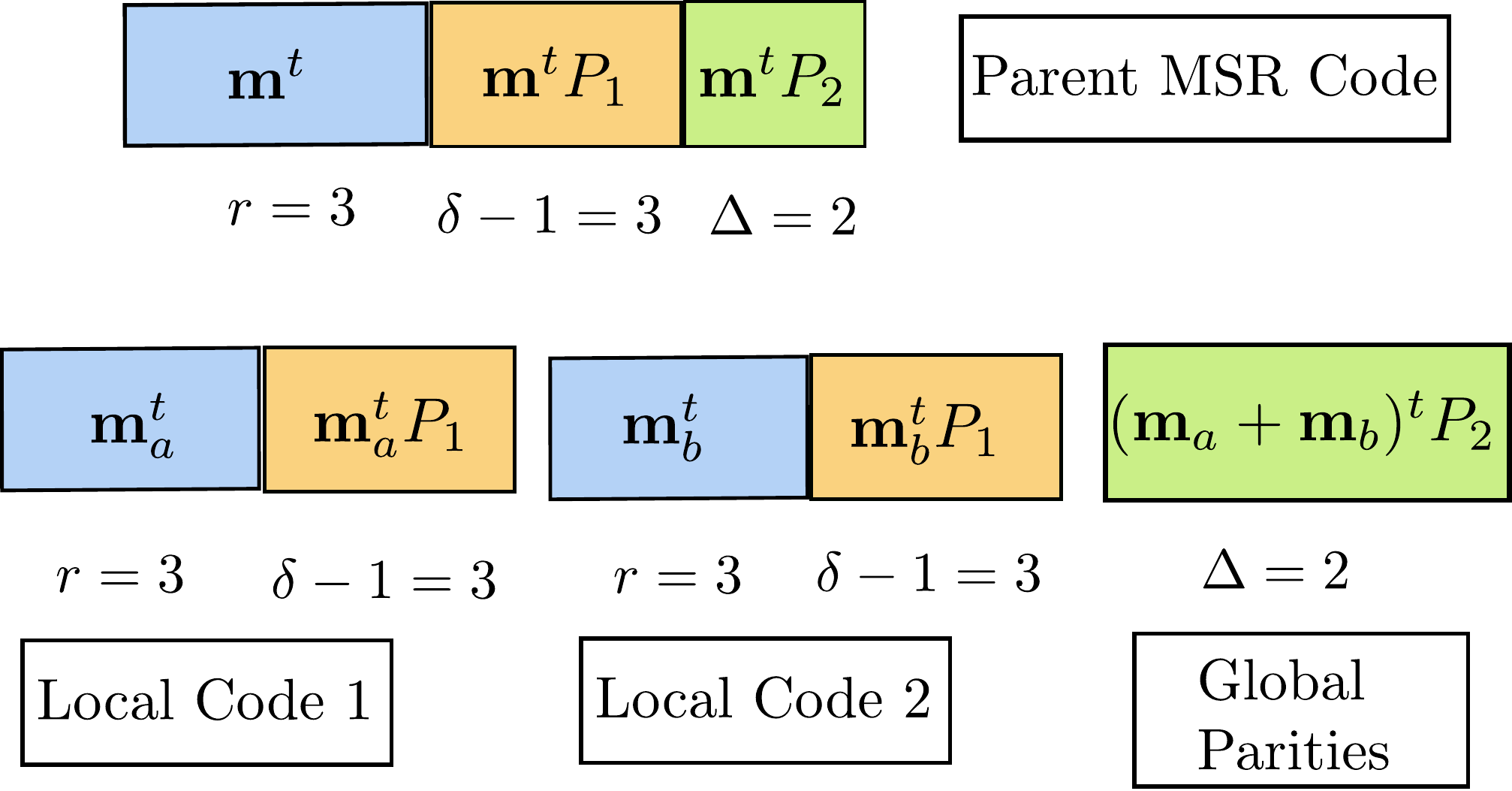}
\end{center}
\caption{The Sum-Parity MSR-Local Code Construction.}\label{fig:sum_parity}
\end{figure}
\end{center}

\item {\em Pyramid-like MSR-Local Code:}   This construction mimics the construction of pyramid codes, with the difference that we are now dealing with vector symbols in place of scalars, and local MSR codes in place of local MDS codes. If we puncture $\Delta$ thick columns, and  the repair degree of the MSR  code that we start out with is less than $n-\Delta $, then the construction will result in an optimal MSR-Local code.

\item {\em Repair-by-Transfer MBR-Local Codes:}  In a repair-by-transfer MBR code, the vector MBR code may be regarded as being built on top of a scalar MDS code.  A scalar pyramid code has constituent local codes which are scalar MDS codes.  The scalar pyramid code also possess a certain number $p$ of global parity symbols.   The present construction begins with a scalar pyramid code in which there are $\ell$ local MDS codes and where the number of global parity symbols $p$ is a multiple of $\alpha$, say $p=\Delta\alpha $.  The next step is the building of a separate repair-by-transfer MBR code on top of each of the $\ell$ constituent local MDS codes.  In the final step, $\Delta$ global-parity nodes are added, each containing a disjoint set of $\alpha$ scalar global parities of the scalar pyramid code.  The construction is illustrated in Fig.~\ref{fig:mbr_local}.

\begin{center}
\begin{figure} [h!]
\begin{center}
\includegraphics[width=5in]{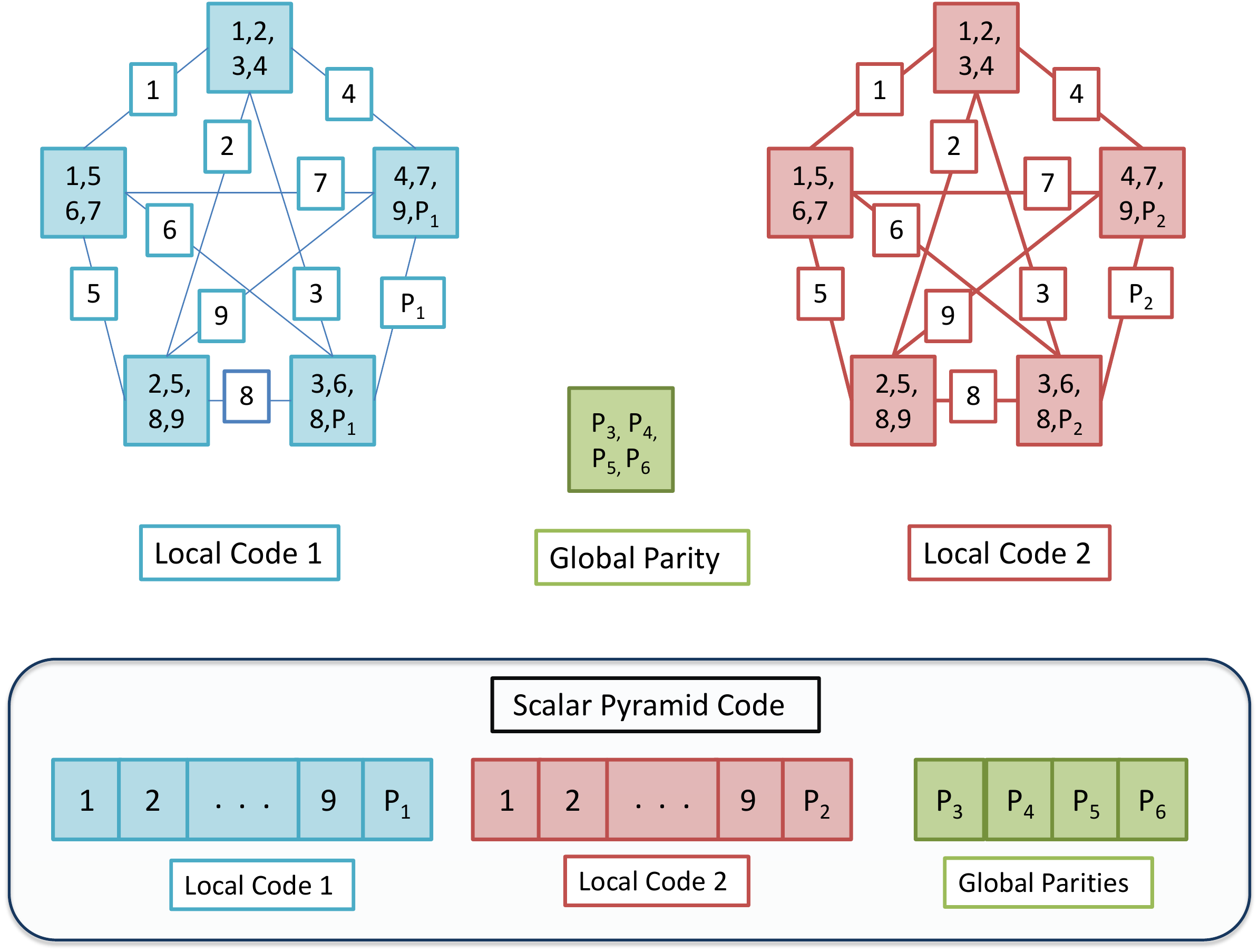}
\end{center}
\caption{The Repair-by-Transfer MBR-Local code is shown on top.  The code below is the underlying scalar pyramid code used to construct the MBR-Local code.}
\label{fig:mbr_local}
\end{figure}
\end{center}

\item {\em Repair-by-Transfer MBR-Local Codes with All-Symbol Locality}.  The difference between this and the immediately previous Repair-by-Transfer Local-MBR code construction is that the scalar pyramid code employed in that construction is replaced here by a scalar all-symbol locality code. Thus the construction begins with a scalar all-symbol locality code in which there are $\ell$ local MDS codes.  The next step is the building of a separate repair-by-transfer MBR code on top of each of the $\ell$ constituent local MDS codes.  The construction is illustrated in Fig.~\ref{fig:mbr_local_all_symbol}.

\begin{center}
\begin{figure} [h!]
\begin{center}
\includegraphics[width=5in]{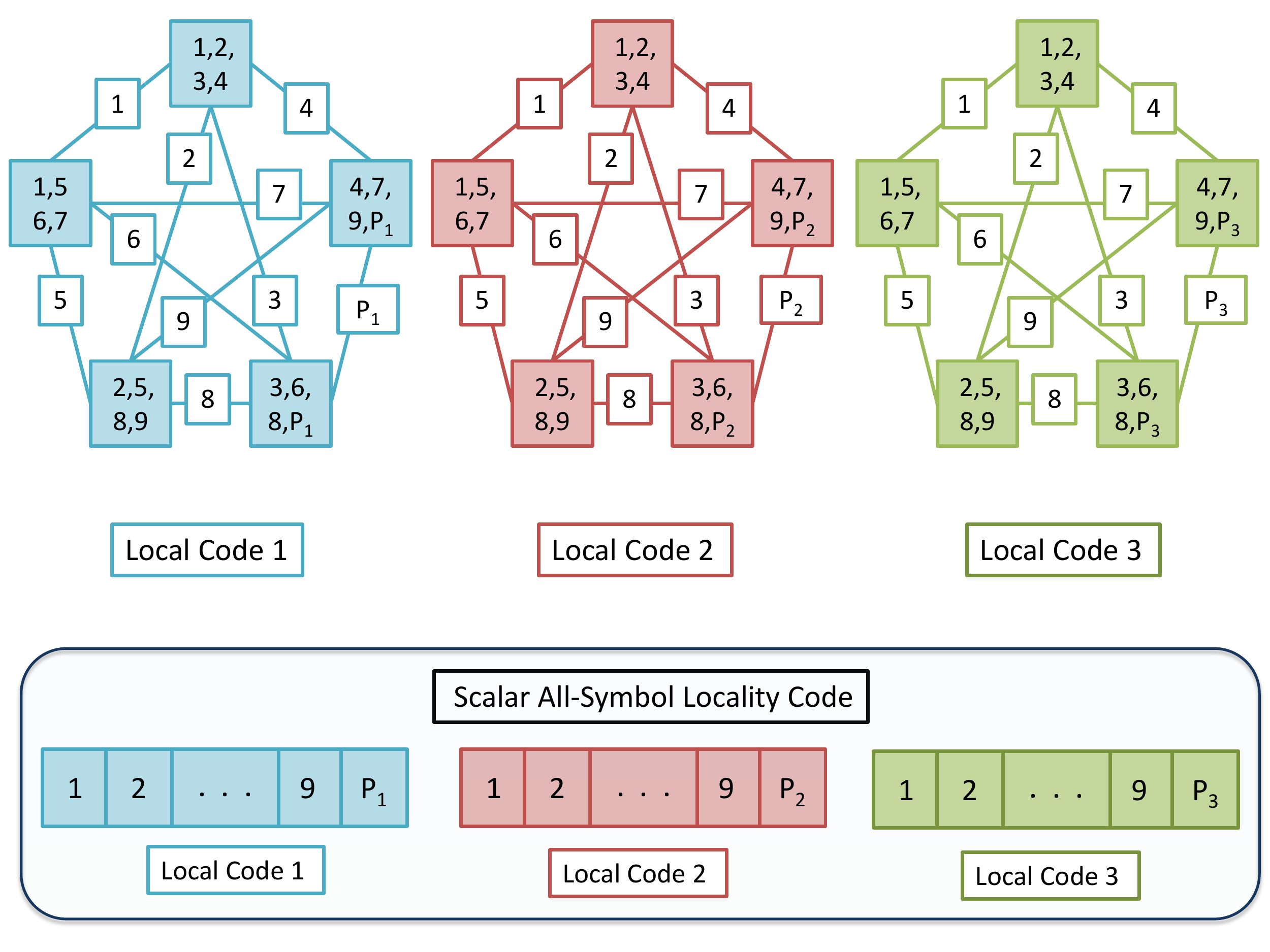}
\end{center}
\caption{The Repair-by-Transfer MBR-Local Code with All-Symbol Locality.}
\label{fig:mbr_local_all_symbol}
\end{figure}
\end{center}

\item We also show the existence, using counting arguments, of
\bit \item {\em MSR-Local codes with information locality} and of
\item {\em MSR-Local codes with all-symbol locality}, \eit
whenever the field size $q$ is sufficiently large. The existence results hold for a larger set of parameters than what we
present using explicit constructions.
\een

A tabular summary of the various constructions contained in this paper is presented in Table \ref{tab:compare1}.  This table summarizes the constructions of vector codes with locality whose local codes are regenerating codes. A performance comparison of the various classes of codes discussed so far is given in the next subsection.

\begin{table}[h]
\caption{Summary of Constructions of Codes with Local Regeneration} \label{tab:compare1}

 \begin{minipage}{7 in}
 \centering
\begin{tabular}{||c|c|c|c|c|c||} \hline
\hline
 & & & &  &  \\
  Construction & Construction Type & Locality Type & Rate & Field Size   & Restrictions on \\
 & & & Optimality  & & parameters \\
  & & & &  &  \\
\hline
\hline
 & & & &  &  \\
 Sum-Parity & Explicit & MSR &  Optimal  & Field size of & $d_{\min} \leq 2\delta$ \\
 Constr. \ref{constr:sum_msr} & & Information& &  Underlying MSR Code & \\
  & &  & &  &  \\
\hline
  & & & &  &  \\
Pyramid-Like & Explicit & MSR &  Optimal  & Field size of &  \\
Constr. \ref{constr:pyramid_msr} & & Information &  & Underlying MSR Code & \\
 & & & &  &  \\
\hline
 & & & &  &  \\
 & Existence & MSR &  Optimal   & ${n \choose mr}$& $K = mr\alpha$ \\
Thm. \ref{thm:msr_info_locality_existence} & & Information & & &  \\
  & & & &  &  \\
\hline
  & & & &  &  \\
 & Existence & MSR &  Optimal  & ${n  \choose \ell }$& $n = m(r+\delta-1), K = \ell \alpha$ \\
Thm. \ref{thm:msr_all_symbol} & & All-Symbol & & & \\
   & & & &  &  \\
\hline
   & & & &  &  \\
RBT-based & Explicit & MBR &  Optimal   & $n \alpha$&  $K_L \mid K$ \footnote{where $K_L$ is the size of the local MBR code and $K$
is the total file size} \\
Constr. \ref{constr:mbr_info_locality} & & Information & & &   \\
    & & & &  &  \\
\hline
 & & & &  &  \\
 RBT-based & Existence & MBR & Optimal & ${n \choose \kappa}$& $K_L\mid K$ and \\
Constr. \ref{constr:mbr_allsymbol_existence} & & All-symbol &  &   & $(r+\delta-1)\mid n$ \\
 & & & &  &  \\
\hline
\hline
\end{tabular}\par
 \vspace{-0.75\skip\footins}
 \renewcommand{\footnoterule}{}
 \end{minipage}
\end{table}

\subsection{Performance Comparison}

We now provide a method for approximately comparing the performance of codes with local regeneration with those of
regenerating codes and scalar codes with locality. The parameters against which
comparison is made are as follows:
\ben
\item the storage overhead $\Omega$ which is the inverse of the code rate
\item the normalized average bandwidth, $\xi$, needed to carry out node repair; the normalization is carried out both with
respect to the amount of data stored as well as the code-length $n$ since the number of node failures will typically be
proportional to $n$, as is the case for example, under a Poisson model of node failures
\item the repair degree, i.e., the number $h$ of helper nodes that a failed node needs to access.
\een
We assume that all codes are designed to offer roughly the same level of reliability which we will translate to mean that
codes having the same block length $n$ must have the same value of minimum distance $d_{\min}$.   Note that the repair degree
$h$ is given by \bit
\item $h=d$ in the case of a regenerating code
\item $h \leq r$ in the case of a scalar local code
\item $h=d$ in the case of a local regenerating code where $d$ is in this case, the repair degree of the constituent local
regenerating codes.
\eit
In general, codes with locality offer a smaller value of repair degree for a given block length of the code.  The challenge
therefore, is to construct codes with locality, which compare favorably with regenerating codes in terms of the two other
performance metrics, namely, storage overhead and repair bandwidth.

To compare the storage overhead and repair bandwidth of the various code constructions, we proceed as follows.   We assume
that a user desires to store a file of size $K$ across $n$ nodes for a time period $T$ with each node storing $\alpha$
symbols. A cost is associated with both node storage as well as for bandwidth consumed during node
repair. We also assume a Poisson-process model of node failures for the whole system. Under this model, the number of
failures in time $T$ is proportional to the product of $T$ and the number of nodes $n$ (for large $n$). For simplicity, we
only consider the case of single-node repairs in the plots, although a similar analysis can be carried out under the
assumption of multiple node failures.  The average cost of a single repair for a coding scheme is taken as the average amount
of data download to repair a node which we denote by $\bar{\omega} $. The cost of storage is assumed to be proportional to
the amount of data stored, i.e., to $n\alpha$.

With this, it follows that if $\gamma(K,T)$ denotes the average cost incurred to store a file of size $K$ for a time period
$T$ using a particular coding scheme, then
\bea
 \gamma(K,T) &=& \left( \gamma_K n\bar{\omega} + \gamma_S n\alpha \right)T
\eea
for some proportionality constants $\gamma_K, \gamma_S$. Hence the average cost incurred in storing one symbol for one unit
of time is given by
\bea
 \frac{\gamma(K,T)}{KT} &=& \gamma_K \frac{n\bar{\omega}}{K} + \gamma_S \frac{n\alpha}{K}.
\eea
We will refer to the  quantity $ \frac{n\bar{\omega}}{B} $ as the normalized repair bandwidth $\xi$ of the code.  Thus the
average cost is a linear combination of the normalized repair bandwidth $\xi=\frac{n\bar{\omega}}{K} $ as well as the storage
overhead $\Omega= \frac{n\alpha}{K}$.

In Fig. \ref{fig:60plot}, the performance of a representative set of codes with local regeneration (obtained via both
explicit constructions and existential arguments) having common length $n = 60$ and common minimum distance $d_{\min} = 8$
are plotted, for the case of a single node failure. Also included, are plot of the family of regenerating codes with
parameters $(n = 60, k = 53, d = 59)$. The repair degree is chosen as $d = 59$, since this results in the
best possible normalized repair-bandwidth vs storage-over tradeoff for this class of codes. In the plots, the $X$-axis
denotes the storage overhead $\Omega $. In the first plot, the $Y$-axis denotes the normalized repair bandwidth $\xi$, while
in the second plot, the $Y$-axis denotes the average number of nodes accessed during repair. We see
that codes with local regeneration, not only have better access, but also are comparable to regenerating codes in terms of
storage overhead and repair bandwidth. Such plots could be drawn for the case of multiple node failures as well.
\begin{figure}[h!]
\begin{center}
\includegraphics[width=13cm]{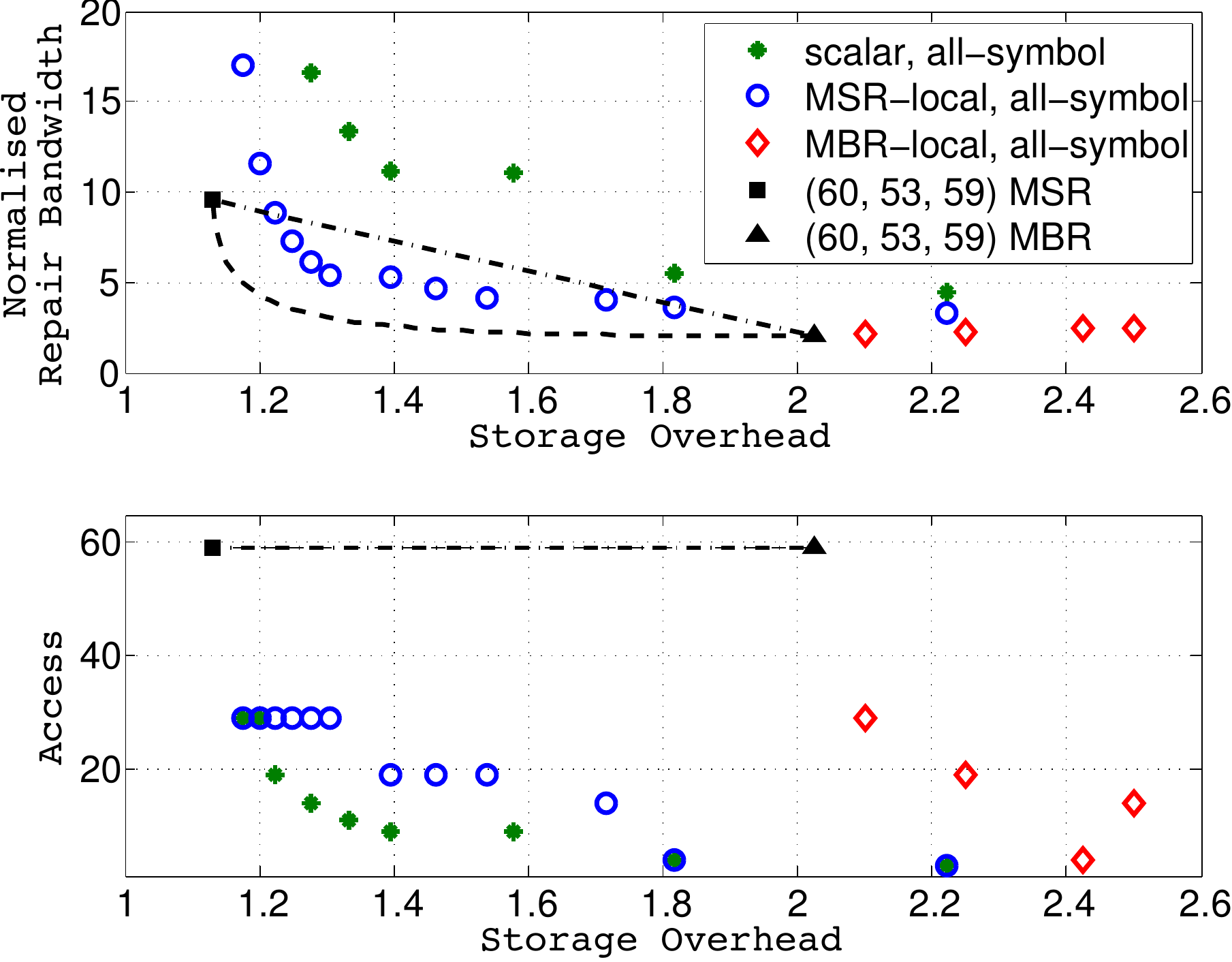}
\caption{The performance of various code constructions presented in this paper as well as that of regenerating codes, all
having common length $60$ and minimum distance $8$ are plotted.  Taken together, the two plots permit a comparison of the
various codes in terms of normalized repair bandwidth, storage overhead and access (i.e., repair degree). }
\label{fig:60plot}
\end{center}
\end{figure}

\section{Scalar Codes with Locality} \label{sec:scalar_local}

In this section, we extend the notion of locality in \cite{GopHuaSimYek} to the case when the local codes are allowed to be
more general codes than just single parity check codes. We derive an upper bound on the minimum distance of such
codes and also deduce the structure of the global code when
\ben[(a)]
\item the bound on minimum distance is achieved with equality and
\item the maximum possible dimension of a local code divides the dimension of the global code.
\een
We then discuss three code constructions, all of which are optimum with respect to the upper bound on minimum distance.
Finally, we end by making a comparison with concatenated codes as concatenated codes may be regarded as special cases of
scalar codes with locality. This viewpoint leads us to an upper bound on the minimum distance of concatenated codes that is
often tighter than what is currently known.

\vspace*{0.1in}

Let $\mathcal{C}$ denote an $[n,k,d_{\min}]$ linear code over $\mathbb{F}_q$ and let $G$ denote a generator matrix
of $\mathcal{C}$. Also, let $\mathbf{c} = (c_1, \ldots, c_n)$ denote a codeword of $\mathcal{C}$. The code $\mathcal{C}$
will also be referred to as a scalar code (considering elements of $\mathbb{F}_q$ as scalars).

\vspace*{0.1in}

\begin{defn}[$(r, \delta)$ code symbol locality]
The $i^{\text{th}}$ code symbol, $c_i, \ i \in [n]$, of $\mathcal{C}$ is said to have $(r, \delta)$ locality, $\delta \geq
2$, if there exists a punctured code of ${\mathcal C}$ with support containing $i$, whose length is at most $r + \delta - 1$,
and whose minimum distance is at least $\delta$, i.e., there exists a subset $S_i \subseteq [n]$ such that
\begin{itemize}
\item  $i \in S_i,  \ |S_i| \leq r + \delta - 1$ and
\item  $d_{\text{min}}\left(\mathcal{C}|_{S_i}\right) \geq \delta$, where $\mathcal{C}|_{S_i}$ denotes the code
obtained when $\mathcal{C}$ is punctured to the set of co-ordinates corresponding to $S_i$.
\end{itemize}
\end{defn}

\vspace*{0.1in}

It follows from the Singleton bound that $\text{dim}(\mathcal{C}|_{S_i}) \leq r$.

\vspace*{0.1in}

\begin{defn}[$(r, \delta)$ information locality] \label{defn:scalar_locality_info}
The code $\mathcal{C}$ is said to have  $(r,\delta)$ information locality if $\mathcal{C}$ has a set of
punctured codes  $\{ \mathcal{C}_i \}_{i \in \mathcal{L}}$ with supports $\{S_i\}_{i \in \mathcal{L}}$,  respectively, such
that, for all $i \in \mathcal{L}$, we have
\begin{itemize}
\item  $|S_i| \leq r + \delta - 1$,
\item  $d_{\text{min}}\left(\mathcal{C}_i\right) \geq \delta$, and
\item  $\displaystyle\text{Rank}(G|_{\cup_{i \in \mathcal{L}}S_i})=k$.
\end{itemize}
Here $\mathcal{L}$ denotes the index set for the local codes and by $G|_S$ we denote the restriction of $G$ to the set of
columns indexed by the set $S$.
\end{defn}

\vspace*{0.1in}
If further $\cup_{i \in \mathcal{L}} S_i=[n]$, then the code is said to have $(r,\delta)$ all-symbol locality.

The above definition for information locality is equivalent to saying that there exists a set of $k$ independent
columns of $G$, indexed by the $\mathcal{I} \subseteq [n], |\mathcal{I}| = k$,  such that all the $k$ code symbols $c_i,
i \in \mathcal{I}$ have $(r, \delta)$ locality.  The $(r, d)$ codes introduced by Gopalan \textit{et al} correspond to $(r,
\delta = 2)$ in the present notation. We also note that if $\mathcal{C}$ has $(r, \delta)$ information
locality, then it must be true that $d_{\min} \geq \delta$.

\vspace{0.1in}

\subsection{Upper Bound on Minimum Distance and Structure of Optimal Codes} \label{sec:scalar_dmin_bound}

An upper bound on the minimum distance of codes with $(r,\delta)$ information locality, was established in \cite{GopHuaSimYek} for the case $\delta=2$ and subsequently extended in \cite{PraKamLalKum} to the general case.  The general result is presented in Theorem~\ref{thm:scalar_info_locality} below.

\vspace{0.2in}

\begin{thm} \label{thm:scalar_info_locality}
Let $\mathcal{C}$ be an $[n,k,d_{\min}]$ scalar code with $(r,\delta)$ information locality. Then the minimum distance
$d_{\min}$ of code ${\cal C}$  is upper bounded by
\begin{eqnarray} \label{eq:bound_info_locality}
d_{\min} \ \leq \ n - k + 1 - \left(\left\lceil{\frac{k}{r}}\right\rceil - 1\right)(\delta - 1). \label{eq:bound_scalar_locality}
\end{eqnarray}
\end{thm}

\vspace{0.1in}

\begin{proof}
See Appendix \ref{app:scalar_info_locality_proof}.
\end{proof}

Any code achieving the bound in Theorem \ref{thm:scalar_info_locality} with equality will be referred to as an optimal code
having $(r, \delta)$ information locality (or all-symbol locality if $\mathcal{C}$ has all-symbol locality).

\vspace{0.1in}

We next deduce the structure of the code $\mathcal{C}$ for the case when $r\mid k$ and $\mathcal{C}$ is an optimal
code having $(r, \delta)$ information locality.

\vspace{0.1in}

\begin{thm} \label{thm:scalar_info_locality_equality}
Let $r \mid k$, set $\frac{k}{r}=t$ and let $\mathcal{C}$ be an optimal code having $(r, \delta)$ information locality.
Also, as in Definition \ref{defn:scalar_locality_info}, let $\mathcal{L}$ denote the index set for all the
local codes of $\mathcal{C}$. Then
\ben[(a)]
\item the local code $\mathcal{C}_i$ must be an $[r+\delta-1, r, \delta]$ MDS code, $\forall \ i \in \mathcal{L}$.
\item the local codes must all have disjoint supports, i.e., $S_i \cap S_j = \phi, \ \forall i,j \in \mathcal{L}, \ i \neq
j$, and
\item for any set of distinct indices $i_1, i_2, \cdots, i_{t} \in \mathcal{L}$ it must be that
 \bea
\text{dim}\left(V_{i_{t}} \cap \left( \displaystyle\sum_{j=1 }^{t-1}V_{i_j} \right)\right) = 0.
\eea
From this it follows that, up to permutation of columns, the $(k\times n)$ generator matrix $G$ of $\mathcal{C}$ can be
expressed in the form
\bea \label{eq:scalar_structure_of_G}
G= \left[ \begin{array}{cccc|c}
           G_1 &&&& \\
           && \ddots&& A \\
           &&&G_{t}&
          \end{array}
 \right],
\eea
where $G_i$ is the $(r  \times  r+\delta-1)$ generator matrix of an $[r+\delta-1,r,\delta]$ MDS code $\forall \ 1\leq i \leq t$, and $A$ is some $((n-t(r+\delta-1))\times n)$ matrix.
\een
\end{thm}

\vspace{0.2 in}

\begin{proof}
See Appendix \ref{app:scalar_info_locality_equality_proof}.
\end{proof}

\subsection{Constructions of Optimal Codes with Locality}

Three constructions, all optimal with respect to the bound on $d_{\min}$ in \eqref{eq:bound_info_locality} are discussed
here.  We begin by showing that the Pyramid codes of \cite{HuaCheLi} are optimal with respect to $(r,\delta)$ information
locality.    We then study codes with $(r,\delta)$ all-symbol locality for the case when $(r+\delta-1)\mid n$.
For the case when the block length is of the form $n=\lceil \frac{k}{r} \rceil  (r+\delta-1)$, we provide an explicit
construction of a code with all-symbol locality by splitting the rows of the parity check matrix of an appropriate MDS code.
We will refer to this as the parity-splitting construction. Finally, the existence of optimal codes with all-symbol locality
is shown for the case when $(r+\delta-1) \mid n$.

\subsubsection{Optimality of the Pyramid Code Construction} \hspace*{\fill}  \label{sec:pyramid_codes}
  \vspace{0.1 in}

We will now show that  under a suitable choice of parameters, the Pyramid code construction appearing in \cite{HuaCheLi}, achieves the bound in Theorem \ref{thm:scalar_info_locality} with equality.  For the sake of completeness, the construction is reproduced below.

Consider an $[n',k,d_{\min}]$ systematic MDS code over $\mathbb{F}_q$, where $n'=k+d_{\min}-1$, having generator matrix of the form
\begin{equation}
G = \left [ \begin{array}{c|c}\underbrace{I}_{(k \times k)} & \underbrace{Q}_{(k \times (d-1))} \end{array} \right].
\end{equation}
The pyramid-code construction will then proceed to modify $G$ to obtain the generator matrix of the desired optimal code. Let $k = \alpha r + \beta$, with $0 \leq \beta \leq (r-1)$.  First the matrix $Q$ is partitioned into submatrices as shown below:
\begin{equation}
Q = \left [ \begin{array}{c|c} Q_1 & \\ \vdots & Q' \\ Q_{\alpha} & \\ Q_{\alpha+1} & \end{array} \right],
\end{equation}
where
$Q_i, 1 \leq i \leq \alpha$  are matrices of size $ (r \times (\delta - 1))$, $Q_{\alpha+1}$ is of size $ (\beta \times (\delta-1))$ and $Q'$ is a $(k \times (d_{\min}-\delta))$ matrix. Next, consider a second generator matrix $G'$ obtained by splitting the first $(\delta -1)$ columns of $Q$ as shown below:
\begin{equation}
G' = \left [ \begin{array}{cccc|cccc|c} I_r &&&& Q_1 &&&& \\ & \ddots &&& & \ddots &&&Q' \\ & & I_r & & & & Q_{\alpha} & &\\&&&I_{\beta} &&&&Q_{\alpha+1} & \end{array} \right],
\end{equation}
Note that $G'$ is a $(k \times n)$ full rank matrix, where
\begin{equation} \label{eq:meets_req}
n = k+d_{\min}-1 + \left(\left \lceil \frac{k}{r} \right \rceil - 1\right)(\delta - 1).
\end{equation}
Clearly, by comparing the matrices $G$ and $G^{'}$, it follows that the code ${\cal C}$ generated by $G'$, has minimum
distance no smaller than $d_{\min}$.  Furthermore, $\mathcal{C}$ is a code with $(r,\delta)$ information locality.  Hence, it
follows from \eqref{eq:meets_req} that ${\cal C}$ is an optimal code having $(r,\delta)$ information locality.

\vspace*{0.2in}

\begin{example} \label{eg:pyramid}
Let $G$ be the generator matrix of a $[7,4,4]$, systematic MDS code:
\bean
G & = &  \left[ \begin{array}{ccccccc}
1 &  &  &  & g_{11} & g_{12}  &g_{13}  \\
 & 1 &  &  & g_{21} & g_{22} &  g_{23}  \\
 &  & 1  & &   g_{31} & g_{32}& g_{33}  \\
 &  &  & 1 &   g_{41} & g_{42} & g_{43}\\
  \end{array} \right].
  \eean
We construct a revised generator matrix $G_{\text{pyr}}$ by splitting the first two parity columns and then rearranging columns:
\bean
G_{\text{pyr}} & = & \left[ \begin{array}{cccc|cccc|c}
1 &  &  g_{11} & g_{12} & & &  & & g_{13}  \\
 & 1 &  g_{21} & g_{22} & & & & &  g_{23}  \\
 \hline
 & &  & & 1   & &  g_{31} & g_{32}& g_{33}  \\
 & &  &  & & 1  &  g_{41} & g_{42} & g_{43} \\
  \end{array} \right].
  \eean
The pyramid code ${\cal C}$ is then the code with generator matrix $G_{\text{pyr}}$.  It can be verified that for $L_1 =\{ 1,2,3,4 \}$ and $L_2=\{ 5, 6, 7, 8 \}$, $G_{\text{pyr}}|_{L_1}$ and  $G_{\text{pyr}}|_{L_2}$ are  both generator matrices of  $[4,2,3]$ MDS codes. Further it is easy to see that $\text{Rank}(G_{\text{pyr}}|_{L_1\cup L_2})=4$.
 It is also straightforward to show that the minimum distance of the pyramid code is no smaller than that of the parent $[7,4,4]$ MDS code.  It turns out that the pyramid code is optimal with respect to the bound in \eqref{eq:bound_info_locality} and hence has code parameters $[9,4,4]$.
\end{example}

\vspace*{0.2in}

 \subsubsection{Optimality of a Parity-Splitting Code Construction}
  \vspace{0.1 in}
\begin{thm}
 Let $n=\lceil \frac{k}{r} \rceil  (r+\delta-1)$. Then, for $q > n$, there exists an explicit and optimal $[n,k,d_{\min}]$
linear code over $\mathbb{F}_q$, having $(r,\delta)$ all-symbol locality .
\end{thm}
\begin{proof}
 Let $H'$ be the parity check matrix of an $[n,k',d]$ Reed-Solomon code over $\mathbb{F}_q$, where
$k'= k+ (\lceil \frac{k}{r} \rceil  - 1)(\delta-1)$ and minimum distance, $d= n-k'+1  = n-k+1- (\lceil \frac{k}{r} \rceil - 1)(\delta-1)$.  Such codes exist if $q \geq n$.
We choose $H'_{(n-k') \times n}$ to be a Vandermonde  matrix.  Let
\begin{equation}
H'= \left[ \begin{array}{c}
     Q_{(\delta-1) \times n} \\
     A_{(d-\delta)\times n}
    \end{array} \right].
\end{equation}
 We next partition the matrix $Q$ into submatrices as shown below:
\begin{equation}
Q = \left [ Q_1 \mid Q_2 \mid \ldots \mid Q_{\lceil \frac{k}{r} \rceil} \right],
\end{equation}
in which the matrices $\{Q_i, i = 1, \ldots, \lceil \frac{k}{r} \rceil \}$ are of uniform size $ ((\delta-1) \times (r+\delta
- 1))$. Now, consider the code $\mathcal{C}$ whose parity check matrix, $H$, is obtained by splitting the first $\delta -1$
rows of $H'$ as follows:
\begin{equation}
 H=\left[ \begin{array}{ccc}
            Q_1 &&\\
	    & \ddots &\\
	    &&Q_{\lceil \frac{k}{r} \rceil}\\
	    \hline \\
	    & A & \end{array}
\right].
\end{equation}
It is clear from the construction that code $\mathcal{C}$ has $(r,\delta)$ all-symbol locality. Let $K$ denote the dimension
of the code $\mathcal{C}$. We will now show that $\mathcal{C}$ is an optimal  $[n,k,d_{\min}]$
code, having $(r,\delta)$ all-symbol locality, by showing that
\begin{itemize}
\item $K = k$ and
\item the minimum distance $d_{\min}$of $\mathcal{C}$ is given by the equality condition in
\eqref{eq:bound_scalar_locality}.
\end{itemize}

To see this, first of all note that the dimension of $\mathcal{C}^{\perp}$ is upper bounded as
\begin{eqnarray}
 \text{dim}\left(\mathcal{C}^{\perp}\right) & \stackrel{(a)}{\leq} & \left \lceil \frac{k}{r} \right \rceil(\delta - 1) + d -
\delta \nonumber \\
& \stackrel{(b)}{=} & \left \lceil \frac{k}{r} \right \rceil (\delta - 1) + \left( n - k - \left(\left \lceil \frac{k}{r}
\right \rceil - 1 \right)(\delta - 1) + 1 \right) - \delta \nonumber \\
& = & n - k \label{eq:proof_parity_splitting_temp1},
\end{eqnarray}
where $(a)$ follows by counting the number of rows of $H$ and $(b)$ follows since $d = n - k' + 1$. Thus, from
\eqref{eq:proof_parity_splitting_temp1}, we get that
\begin{eqnarray}
K & \geq & k \label{eq:proof_parity_splitting_temp2}.
\end{eqnarray}

Next, we note by inspection of the matrices $H, H'$ that any vector which is in the null-space of $H$
is also in the null-space of $H'$.  It follows that the minimum distance $d_{\min}$ of the parity-splitting construction is
at least that of the parent Reed-Solomon code having parity-check matrix
$H'$.  Thus
\bea
d_{\min} & \geq & d \ = \  n-k+1- \left(\left \lceil \frac{k}{r} \right \rceil - 1\right)(\delta-1).
\label{eq:proof_parity_splitting_temp3}
\eea
But, since $\mathcal{C}$ has $(r, \delta)$ locality, from \eqref{eq:bound_info_locality}, we must
also have that
\bea
d_{\min} & \leq & n - K + 1 - \left(\left \lceil \frac{K}{r} \right \rceil - 1\right)(\delta-1).
\label{eq:proof_parity_splitting_temp4}
\eea
From \eqref{eq:proof_parity_splitting_temp3} and \eqref{eq:proof_parity_splitting_temp4}, we get that
\bea
k + \left \lceil \frac{k}{r} \right \rceil(\delta-1) & \geq & K + \left \lceil \frac{K}{r} \right \rceil(\delta-1),
\label{eq:proof_parity_splitting_temp5}
\eea
which together with \eqref{eq:proof_parity_splitting_temp2} implies that $K =k$ and also that $\mathcal{C}$ is optimal with
respect to \eqref{eq:bound_scalar_locality}.
\end{proof}

\vspace{0.1in}

\subsubsection{Existence of Optimal $(r,\delta)$ codes with All-Symbol Locality}\hspace*{\fill}
 \vspace{0.1in}

\begin{thm} \label{thm:scalar_allsymbol_existence}
 Let $q > kn^k$ and $(r+\delta-1)\mid n$ . Then there exists an optimal $[n,k,d_{\min}]$ code with  $(r,\delta)$ all-symbol locality  code over $\mathbb{F}_q$.
\end{thm}
\begin{proof}
See Appendix \ref{app:scalar_allsymbol_existence_proof}.
\end{proof}

%
%
%
%

\subsection{An Upper Bound to the Minimum Distance of Concatenated Codes} \label{sec:concatenated_codes}

Consider a (serially) concatenated code (see \cite{For}, \cite{Dum}) having an $[n_1, k_1, d_1]$ code $\mathcal{A}$ as the
inner code and an $[n_2, k_2, d_2]$ code $\mathcal{B}$ as the outer code.  Clearly, a concatenated code falls into the
category of an code with $(r,\delta)$ all-symbol locality with $\delta = d_1$, $r = n_1 - d_1 + 1$.   Hence, the bound in
\eqref{eq:bound_scalar_locality} applies to concatenated codes as well. Using the fact that a concatenated code has length $n
= n_1n_2$, dimension $k = k_1k_2$, we obtain from Theorem \ref{thm:scalar_info_locality} the following upper bound
on minimum distance $d_{\min}$:
\begin{equation} \label{eq:bound_concatenated}
d_{\min} \ \leq \ n_1n_2 - k_1k_2  +1 -\left(\left\lceil{\frac{k_1k_2}{n_1-d_1+1}}\right\rceil - 1\right)(d_1 - 1 ).
\end{equation}
Well known bounds on the minimum distance of a concatenated codes are
\begin{equation} \label{eq:distance_concatenated_upperbound_known}
d_1d_2 \leq d_{\min} \ \leq \ n_1d_2.
\end{equation}
In practice, concatenated codes often employ an interleaver between the inner and outer codes in order to increase the
minimum distance \cite{BenDivMonPol}.  In this case, while the upper bound in
\eqref{eq:distance_concatenated_upperbound_known} no longer holds, the bound in \eqref{eq:bound_concatenated} continues to
hold, since the code continues to possess all-symbol locality even after interleaving.

We observe that even when an interleaver is not used, \eqref{eq:bound_concatenated} is tighter than
\eqref{eq:distance_concatenated_upperbound_known} if both the  codes are MDS and the dimension of the first code $k_1 > 1$.
In this case $k_1=n_1-d_1+1$ and $k_2=n_2-d_2+1$. Then \eqref{eq:bound_concatenated} gives us that
\begin{eqnarray}
 d_{\min} \ & \leq & \ n_1n_2 - k_1k_2  +1 -\left(\left\lceil{\frac{k_1k_2}{k_1}}\right\rceil - 1\right)(d_1 - 1 ) \\ \nonumber
     & = & \ n_1(k_2+ d_2-1 ) - k_1k_2  +1 - (k_2 -1)(d_1)+k_2-1\\ \nonumber
     & = & \ n_1k_2+ n_1d_2 - n_1 - k_1k_2 + 1 - (k_2)(n_1-k_1+1) + n_1 -k_1 +1 + k_2 -1 \\ \nonumber
     & = & \ n_1d_2 -(k_1 -1) \\
     & < & n_1d_2,
\end{eqnarray}
when  $k_1 > 1$.

\section{Vector Codes} \label{sec:vec_code_prelims}

Our aim here is to extend the notions of locality to include codes whose local codes are codes such as regenerating codes.  Vector codes, by which we mean codes over a vector symbol alphabet, provide the appropriate setting for such an extension\footnote{As noted in Section~\ref{sec:intro}, these codes can equivalently also be regarded as array codes.}.
\vspace{0.1in}

\begin{defn}
An $\mathbb{F}_q$-linear \textit{vector code} of block length $n$ is a code ${\mathcal{C}}$ having a symbol alphabet $\mathbb{F}_q^{\alpha}$ for some $\alpha>1$, i.e.,
\begin{eqnarray}
\mathcal{C}& = & \left\{ \mathbf{c} \ = \ (\mathbf{c}_1, \mathbf{c}_2 \ldots, \mathbf{c}_n), \ \mathbf{c}_i \in
\mathbb{F}_q^{\alpha}, \text{  all  }  i \in [n] \right\},
\end{eqnarray}
satisfying the additional property that given $\mathbf{c}, \mathbf{c}' \in \mathcal{C}$ and $a,b \in \mathbb{F}_q$, \begin{eqnarray}
 a\mathbf{c} + b\mathbf{c}' & \triangleq &  (a\mathbf{c}_1 + b\mathbf{c}'_1, a\mathbf{c}_2 + b\mathbf{c}'_2 \ldots,
a\mathbf{c}_n + b\mathbf{c}'_n)
\end{eqnarray}
also belongs to $\mathcal{C}$, in which $a\mathbf{c}_i$ is simply the scalar multiplication of the vector $\mathbf{c}_i$.
\end{defn}

\vspace{0.1in}

We will refer to symbols from $\mathbb{F}_q$, $\mathbb{F}_q^{\alpha}$ as scalar and vector symbols respectively.  The
field $\mathbb{F}_q$ will be termed as the base field and the parameter $\alpha$ as the vector-size parameter\footnote{In the distributed storage context, $\alpha$ is also the node-size parameter as it denotes the number of symbols contained in a node. }.  Associated with the vector code $\mathcal{C}$ is an $\mathbb{F}_q$-linear scalar code $\mathcal{C}^{(s)}$ of length $N= n\alpha$, where $\mathcal{C}^{(s)}$ is obtained by expanding each vector symbol within a codeword into $\alpha$ scalar symbols (in some prescribed order).  Conversely, the scalar code $\mathcal{C}^{(s)}$ also uniquely
determines the vector code if one is given a-priori, the manner in which sets of $\alpha$ scalar code symbols are to be grouped together, to obtain the corresponding vector symbols.  We will assume the canonical grouping, in which the first $\alpha$ scalar symbols form the first vector code symbol etc. We also use $K$ to denote the dimension of the scalar code $\mathcal{C}^{(s)}$ and often refer to  it as the scalar dimension of the code $\mathcal{C}$.

Given  a generator matrix $G$ for the scalar code ${\cal C}^{(s)}$, the first code symbol in the vector code is naturally associated with the first $\alpha$ columns of $G$ etc.  We will refer to the collection of $\alpha$ columns of $G$ associated with the $i^{\text{th}}$ code symbol ${\bf c}_i$ as the $i^{\text{th}}$ thick column.   To avoid confusion, we will refer to the columns of $G$ themselves as thin columns and hence there are $\alpha$ thin columns per thick column of the generator matrix.   We will assume that the $\alpha$ thin columns comprising any thick
column are linearly independent which is equivalent to saying that as the codewords run through the code ${\cal C}$, the $i^{\text{th}}$ code symbol ${\bf c}_i$, takes on all possible values from $\mathbb{F}_q^{\alpha}$.    We will also use $W_i$ to denote the ($\alpha$-dimensional) subspace of $\mathbb{F}_q^K$ associated with the $\alpha$ thin columns making up the $i^{\text{th}}$ thick column.

Given a subset $\mathcal{I} \subseteq [n]$, we use $G|_{\mathcal{I}}$ to denote the restriction of $G$ to
the set of thick columns with indices lying in $\mathcal{I}$.  We will declare $\mathcal{I}$ to be an information set for ${\cal C}$, if
\bea \label{eq:info_set}
\text{rank}(G|_{\mathcal{I}}) & = & K
\eea
and if further, no proper subset of $\mathcal{I}$ possesses this property.   The requirement in \eqref{eq:info_set} is equivalent to stating that
\begin{eqnarray}
\sum_{i \in \mathcal{I}} W_i & = & \mathbb{F}_q^{K}.
\end{eqnarray}
Since the subspaces $\{W_i, i = 1, \ldots, n\}$ can have non-trivial intersection, it follows that information sets
can be of different cardinality.   We use $\kappa$ to denote the minimum cardinality of an information set:
\begin{eqnarray}
\kappa & \triangleq& \min_{\text{information sets $\mathcal{I}$
of $\mathcal{C}$}} |\mathcal{I}|,
\end{eqnarray}
and will refer to $\kappa$ as the quasi-dimension of the code $\mathcal{C}$ or $\text{q-dim}(\mathcal{C})$.  Any $\mathcal{I}$ such that $|\mathcal{I}| = \kappa$ will be referred to as a minimum cardinality information set.
{ \begin{note}
              Scalar codes correspond to vector codes with $\alpha=1$.  The quasi-dimension $\kappa$ of a scalar code equals its dimension $k$.
             \end{note}
}

The (Hamming) distance between any two codewords $\mathbf{c}$ and $\mathbf{c}'$ of $\mathcal{C}$ is the number of vector symbols in which $\mathbf{c}$ and $\mathbf{c}'$ differ.  Since $\mathcal{C}$ is $\mathbb{F}_q$-linear,
it follows that the minimum distance, $d_{\text{min}}$, of $\mathcal{C}$ is equal to the minimum Hamming weight of a non-zero codeword in $\mathcal{C}$.

We will refer to a vector code of block length $n$, scalar dimension $K$, minimum distance $d_{\min}$, vector-size parameter $\alpha$ and  quasi-dimension $\kappa$ as  an $[n , K, d_{\text{min}}, \alpha, \kappa]$ code. This notation will be simplified to $[n , K, d_{\text{min}}]$, whenever the vector-size parameter $\alpha$ and the quasi-dimension $\kappa$ is either clear from the context or else is not relevant to the discussion.    If $[N,K,D_{\text{min}}]$ are the parameters of the scalar code $\mathcal{C}^{(s)}$, it is easily verified that
\bea \label{eq:parameters} \nonumber
\left\lceil \frac{K}{\alpha}\right\rceil & \ \leq \ \kappa \ \leq & \left\lceil\frac{N-D_{\text{min}}+1}{\alpha}\right\rceil , \\
 \left\lceil\frac{D_{\text{min}}}{\alpha}\right\rceil & \ \leq \ d_{\text{min}} \ \leq & D_{\min}.
\eea
We define the {\em rate} $\rho(\mathcal{C})$ of an $[n,K,d_{\min},\alpha,\kappa]$ vector code $\mathcal{C}$ as the quantity
\begin{eqnarray}
 \rho(\mathcal{C})  & = & \frac{K}{n\alpha} .
 \end{eqnarray}
Since, $\left \lceil \frac{K}{\alpha}\right\rceil \ \leq \ \kappa$, it follows that
\bea
 \rho(\mathcal{C})  & \leq & \frac{\kappa}{n}.
\eea

\subsection{Singleton and Erasure Bounds} \label{sec:MDS_erasure_bounds}

The Singleton bound on the size $q^K$ of the vector code $\mathcal{C}$ yields
\begin{eqnarray} \label{eq:MDS_bound}
q^K  & \leq & \left(q^{\alpha}\right)^{n-d_{\text{min}}+1}  \ \ \ \ \text{ (Singleton bound on code size)},
\end{eqnarray}
which gives us:
\begin{eqnarray} \label{eq:MDS_bound_distance}
 d_{\text{min}} & \leq & n - \left\lceil \frac{K}{\alpha}\right\rceil + 1 \ \ \ \ \text{ (Singleton bound on minimum distance)}.
\end{eqnarray}
We will refer to codes achieving the Singleton bound \eqref{eq:MDS_bound} with equality as vector MDS codes. Several constructions of vector MDS codes are known in literature, for example see \cite{BlaFarTil, BlaBraBruMen, BlaBraBruMenVar, BlaRot, BlaBruVar}.

A second bound arises from noting that, given any information set $\mathcal{I}$, the minimum distance is upper bounded by
\begin{eqnarray}
d_{\text{min}} \leq n - |\mathcal{I}| + 1 \label{eq:erasure_bound_general_infoset}.
\end{eqnarray}
This follows since the minimality inherent in our definition of an information set implies the existence of a non-zero codeword which is zero on $(|\mathcal{I}|-1)$ symbols.  In particular, since $\kappa$ is the smallest possible size of an information set, we have that \begin{eqnarray}\label{eq:erasure_bound}
 d_{\text{min}} & \leq & n - \kappa + 1   \ \ \ \ \text{ (erasure bound)}.
\end{eqnarray}
The converse implication of \eqref{eq:erasure_bound_general_infoset} is that $n-d_{\min}+1$ is the largest possible size of  an information set for the code.  We will refer to \eqref{eq:erasure_bound}
as the \textit{erasure bound} for vector codes.
Note that since $\kappa \geq \left\lceil \frac{K}{\alpha}\right\rceil$ the erasure bound in
\eqref{eq:erasure_bound} is in general tighter than the Singleton bound in \eqref{eq:MDS_bound_distance}.

Equality in \eqref{eq:MDS_bound} holds only if $ K=\kappa \alpha$ and  in this case $d_{\text{min}} = n -  \frac{K}{\alpha}
+ 1$, whereas, equality in \eqref{eq:MDS_bound_distance} can hold even if $\alpha \nmid K$.   We will say that a vector code is systematic if through a sequence of elementary row
operations and thick-column permutations, the generator matrix $G$ can be reduced to the form $G = [I_{\kappa\alpha} \mid
P_{\kappa\alpha \times (n-\kappa)\alpha}]$.    Without loss of generality, one can assume that the generator matrix $G$ of a vector MDS code is in systematic form i.e., is of the form $G = [I \mid P]$, where $I$ is an identity matrix of size $K$ and
corresponds to
$\kappa$ thick columns, while $P$ is a $K \times (N - K)$ matrix. A characterization of a vector MDS code in terms of its generator matrix is presented next.  The proof is analogous to the scalar case and hence omitted.

\vspace{0.1in}

\begin{lem}\label{lem:MDS}
Any $[n,K,d_{\min},\alpha,\kappa]$ vector code $\mathcal{C}$ is MDS if and only if the generator matrix
can be represented in the form $G = [I \mid P]$, where the $K \times (N - K)$ matrix
\begin{equation*}
 P= \begin{bmatrix}
      G_{1,1} & G_{1,2} & \dots & G_{1,n-\kappa} \\
      G_{2,1} & G_{2,2} & \dots & G_{2,n-\kappa} \\
      \vdots & & \ddots  &\vdots \\
      G_{\kappa,1} & G_{\kappa,2} & \dots & G_{\kappa,n-\kappa} \\
    \end{bmatrix}
\end{equation*}
possesses the property that every square block submatrix of $P$ is invertible. Here, the $\{G_{i,j}\}$ are square sub-matrices of size $\alpha  \times \alpha$, and by a \text{block} submatrix, we mean a submatrix whose entries belong to the $\{G_{i, j}\}$.
\end{lem}

%
%

\vspace{0.1in}

\subsection{Puncturing and Shortening of a Vector Code}
Given any set $S \subseteq [n]$, we use $\mathcal{C}|_S$ to denote the restriction of the code to the set $S$ and will refer to this code as the code $\mathcal{C}$ punctured to set $S$.  Unlike in the scalar case, the quasi-dimension
of a punctured code, $\mathcal{C}|_S$, can be either larger or smaller than $\text{q-dim}(\mathcal{C})$.
\vspace{0.1in}

We define the shortened code $\mathcal{C}|^{S}$ as the code obtained by first restricting the attention to those codewords whose code symbols are zero on the complement $S^c$ of $S$ and then deleting the coordinates associated to $S^c$ leaving behind a code of length  $|S|$ \footnote{The generator matrix of the shortened code may not have the property that all thin columns associated with a thick column are linearly independent.  This issue, however, does not arise in the case of vector MDS codes.}. The lemma below describes the effect of shortening a vector MDS code. The proof is identical to that of the scalar case and is omitted.

\vspace{0.1in}

\begin{lem} \label{lem:short_MDS}
Given an $[n, K, d_{\text{min}} = n - \kappa + 1, \alpha, \kappa =\frac{K}{\alpha}]$ vector MDS code $\mathcal{C}$, and a set $S
\subseteq [n]$ such that $n - |S| < \kappa$, the shortened vector code $\mathcal{C}|^{S}$ is also vector MDS with parameters
$[|S|, \kappa' \alpha, d_{\text{min}}, \alpha, \kappa' = \kappa - (n - |S|)]$.
\end{lem}

\subsection{Regenerating Codes as Vector Codes}\label{sec:reg_codes_review}

 Let $\mathcal{C}$ denote an $((n, k, d), (\alpha, \beta),B)$ regenerating code, as discussed in Section \ref{sec:intro_regen}. The class of regenerating codes under consideration here will all be linear and will have the property that all the $\alpha$
scalar symbols contained within a node are linearly independent and hence these codes fall within the framework of vector
codes considered here. Recall that the reconstruction property of a regenerating code says that the entire file can be recovered given the contents of any set of $k$ nodes and hence it follows that the minimum distance $d_{min}$ of $\mathcal{C}$ is lower bounded by
\bea \label{eq:d_min_erasure}
d_{\min} & \geq & n-k+1.
\eea
The lemmas below deal with the quasi-dimension of MSR and exact-repair MBR regenerating codes as well as the impact
of puncturing and shortening these codes.

\vspace{0.1in}

\begin{lem} \label{lem:MSR_kappa}
Any MSR code (either exact or functional repair) is vector MDS, i.e., achieves \eqref{eq:MDS_bound} with equality, and
has quasi-dimension $\kappa = k$.
\end{lem}
\begin{proof}
The scalar dimension (file size) of an MSR code is given by $K = B = k\alpha$, which implies that the
quasi-dimension $\kappa \geq k$.   On the other hand, from the data reconstruction property, one can recover all the data by
connecting to any set of $k$ nodes and hence $\kappa=k$ which implies in turn that  $K=\kappa \alpha$. This along with  \eqref{eq:d_min_erasure} implies that the code is vector MDS.
\end{proof}

\vspace{0.1in}

\begin{note}
When we say that a functional-repair MSR code is vector MDS, we will mean that the code
remains vector-MDS after every repair operation.
\end{note}

\vspace{0.1in}

\begin{cor}\label{lem:systematic}
 The generator matrix $G$ of any MSR code can be represented in systematic form $G=[I_K \mid P]$.
\end{cor}

\vspace{0.1in}

\begin{lem}  \label{lem:MBR_kappa}
Any exact-repair MBR code is optimal with respect to the erasure bound and has quasi-dimension $\kappa = k$.
\end{lem}
\begin{proof}
The fact that the quasi-dimension $\kappa = k$ follows from properties of exact-repair optimal regenerating codes
discussed in \cite{ShaRasKumRam_rbt}.  The erasure bound optimality then follows from \eqref{eq:d_min_erasure}.
\end{proof}

\vspace{0.1in}

\begin{lem}\label{lem:punct_regen}
Suppose $\mathcal{C}$ is any $((n, k, d), (\alpha, \beta), B)$ regenerating code and if $S \subseteq [n]$ is such that $|S| > d$,
then the punctured code $\mathcal{C}|_S$ is also a regenerating code with parameters $(|S|, k, d, (\alpha, \beta), B)$.
\end{lem}

\vspace{0.1in}

\begin{lem}[Theorem $6$ of \cite{RasShaKum_pm}]\label{lem:short_MSR}
Suppose $\mathcal{C}$ is an $((n, k, d), (\alpha, \beta))$ MSR code and consider $S \subseteq [n]$ such that $\gamma \triangleq
n  - |S| < k$. Then the shortened code $\mathcal{C}|^S$ is also an MSR code with parameters $(n - \gamma, k - \gamma, d
- \gamma, (\alpha, \beta))$.
\end{lem}

\subsection{Uniform Rank Accumulation Codes} \label{sec:URA_code}

\begin{defn} \label{defn:URA_code}
Let $\mathcal{C}$ be an $[n, K, d_{\min}, \alpha, \kappa]$ vector code having generator matrix $G$ and let $S_i, 1
\leq i \leq n$ be an arbitrary subset of $i$ thick columns of $G$. The code $\mathcal{C}$ is said to be a Uniform Rank Accumulation(URA) code, i.e., a code possessing the URA property, if the restriction $G |_{S_i}$  of $G$ to $S_i$, has rank equal to \bean
\sum_{j=1}^{i} a_j,
\eean
for some set $\{a_1,a_2,\cdots, a_{n}\}$ of non-negative integers that are independent of the specific set $S_i$ of $i$
thick columns chosen.  \end{defn}

\vspace{0.1in}

We will refer to the sequence $\{a_i, 1 \leq i \leq n\}$ as the rank accumulation profile of the code $\mathcal{C}$. Under the
the definition of a vector code considered here, all the thin columns comprising a thick column in the generator matrix of a vector code are linearly independent and hence  we have $a_1 = \alpha$, the vector-size parameter of $\mathcal{C}$. It is also
straightforward to show that
\begin{equation}
\alpha = a_1 \geq a_2 \geq \cdots a_{n-2}\geq a_{n-1} \geq a_{n} \geq 0,
\end{equation}
and that
\begin{eqnarray}
\sum_{i=1}^{n}a_i & = & K.
\end{eqnarray}
Moreover, given that $\mathcal{C}$ has minimum distance $d_{\min}$, it follows that the last $d_{\min}-1$ elements of the
rank accumulation profile must equal $0$, i.e.,
\begin{eqnarray}
a_{n-i} & = & 0 , \ \ 0 \leq i \leq (d_{\min}-2),
\end{eqnarray}
and
\begin{eqnarray}
a_{n-d_{\min}+1} & > & 0.
\end{eqnarray}
\begin{note}\label{rem:URA_erasure_opt}
Whenever $\mathcal{C}$ is a URA code, since any set of $\kappa$ thick columns of $G$ form an information set for $\mathcal{C}$, it follows that $\mathcal{C}$ is optimal with respect to the erasure bound in \eqref{eq:erasure_bound}.
\end{note}

\begin{note} Clearly, any vector MDS code is a URA code and it follows from Lemma \ref{lem:MSR_kappa} that this is also true of MSR codes. For both these classes of codes we have
\begin{eqnarray}
 a_i \ = \ \alpha, \ 1 \leq i \leq \kappa &, & a_i \ = \ 0, \ \kappa + 1 \leq i \leq n.
\end{eqnarray}
One can also show using the information accumulation profile for exact regenerating codes (see \cite{ShaRasKumRam_rbt}),
that MBR codes are also URA codes.   In the case of the MBR codes, the rank accumulation profile $\{a_i\}$ is given by
\begin{eqnarray}
 a_i \ = \ \alpha - i + 1, \ 1 \leq i \leq \kappa &, & a_i \ = \ 0, \ \kappa + 1 \leq i \leq n.
\end{eqnarray}
\end{note}

\section{Locality in vector codes} \label{sec:locality_vec_codes}

In this section we define the notion of locality in the context of vector codes, in a manner analogous to that of scalar
codes. We will specifically consider codes with locality where all the local codes are URA codes and obtain a upper bound on
their minimum distance. Like in the scalar case, we will also deduce the structure of codes which achieve the bound with
equality. The discussion of minimum distance bounds for codes with locality, where local codes are not necessarily URA codes
is deferred until Section \ref{sec:kappa_bound}.

In all of the definitions below, $\mathcal{C}$ is an $[n, K, d_{\text{min}}, \alpha, \kappa]$ vector code possessing a $(K
\times n\alpha)$ generator matrix $G$.

\vspace*{0.2in}

\begin{defn}[$(r, \delta)$ locality] \label{defn:locality}
The $i^{\text{th}}$ vector code symbol, $i \in [n]$, of $\mathcal{C}$ is said to have  $(r, \delta)$ locality, $\delta \geq 2$, if there exists a punctured code of
${\mathcal C}$ with support containing $i$, whose length is at most $r + \delta - 1$, and whose minimum distance is at least
$\delta$, i.e., there exists a subset $S_i \subseteq [n]$ such that
\begin{itemize}
\item  $i \in S_i,  \ |S_i| \leq r + \delta - 1$ and
\item  $d_{\text{min}}\left(\mathcal{C}|_{S_i}\right) \geq \delta$.
\end{itemize}
\end{defn}
It follows from the erasure bound, given in \eqref{eq:erasure_bound}, that $\text{q-dim}(\mathcal{C}|_{S_i}) \leq r$.

\vspace{0.1in}
\begin{defn}[$(r, \delta)$ information locality] \label{defn:info_locality}
The code $\mathcal{C}$ is said to have $(r,\delta)$ information locality if there exist a set of punctured codes  $ \{\mathcal{C}_i \}_{i \in \mathcal{L}}$ of $\mathcal{C}$ with respective supports $\{S_i\}_{i \in \mathcal{L}}$ such that
\begin{itemize}
\item  $|S_i| \leq r + \delta - 1$,
\item  $d_{\text{min}}\left(\mathcal{C}_i\right) \geq \delta$, and
\item  $\displaystyle\text{Rank}(G|_{\cup_{i \in \mathcal{L}}S_i})=K$.
\end{itemize}
Here $\mathcal{L}$ denotes the index set for the local codes.
\end{defn}
\vspace{0.1in}
If further $\cup_{i \in \mathcal{L}}S_i=[n]$, then the code is said to have $(r, \delta)$ all-symbol locality.

The case of locality in vector codes with $\delta = 2$ has been previously considered in \cite{PapDim}, where it was shown that
under $(r, \delta = 2)$ all-symbol locality, the minimum distance $d_{\text{min}}$ of $\mathcal{C}$ is upper bounded by
\bea
 d_{\text{min}} & \leq & n - \left \lceil \frac{K}{\alpha} \right \rceil + 1 - \left(
\left \lceil \frac{K}{r\alpha}\right\rceil - 1\right). \label{eq:dimiakis_delta_eq_2}
\eea

From an implementation point of view, it is desirable that the local codes be identical and this prompts the definition of
exact locality. We define the code $\mathcal{C}$ to have exact $(r,\delta)$ information locality if $\mathcal{C}$ has
$(r,\delta)$ information locality such that $|S_i| = r + \delta - 1$ and $d_{\text{min}}\left(\mathcal{C}_i\right) =
\delta$, $\forall i \in \mathcal{L}$. In addition, if $\cup_{i \in \mathcal{L}}S_i=[n]$, then the code is said to
have exact $(r, \delta)$ all-symbol locality.


Let ${\cal U}$ (for Uniform rank accumulation) denote the class of $\mathbb{F}_q$-linear vector codes ${\cal C}$, where each
code $\mathcal{C}$ is an $[n, K, d_{\text{min}}, \alpha, \kappa]$ vector code
\bit
\item possessing exact $(r, \delta)$ information locality with $\delta \geq 2$,
\item whose associated local codes $\{ \mathcal{C}_i\}_{i \in \mathcal{L}}$ are URA
codes (described in Section~\ref{sec:URA_code}) with rank accumulation profile $\{a_i,  i \in [r+\delta-1] \}$.
\eit
Note that since URA codes are erasure optimal (see Remark \ref{rem:URA_erasure_opt}), q-dim$(\mathcal{C}_i) = r, \forall i
\in \mathcal{L}$ and hence $a_{r +1}, \ldots , a_{n_L}$ are all zeros. We use $n_L,K_L$ to denote the block length and scalar
dimension of the local codes ${\cal C}_i$ respectively, i.e.,
\bean
n_L & \triangleq & r + \delta - 1, \\
K_L & \triangleq & \sum_{i = 1}^{n_L}a_i.
\eean

We will now present an upper bound on the minimum distance $d_{\min}$ of the code $\mathcal{C}$, whenever  $\mathcal{C} \in
\mathcal{U}$. Subsequently, under certain assumptions, we also identify necessary conditions of optimality with respect to
the minimum distance bound. We begin by introducing some terminology that will be used to describe the bound.

\vspace*{0.05in}

\subsection{Sub-Additivity}

Let us extend the finite length vector $(a_1,a_2,\cdots,a_{n_L})$ to a periodic semi-infinite sequence $\{a_i\}_{i=1}^{\infty}$ of period $n_L$ by defining
\bea
a_{i+j n_L} & = & a_{i}, \ 1 \leq i \leq n_L, \  j \geq 1.
\eea
Let $P(\cdot)$ denote the sequence of leading sums of this semi-infinite sequence, i.e.,
\bea
P(s) & = & \sum_{i=1}^s a_i \label{eq:P_fun}, \ \ s\geq 1.
\eea
It follows from the periodicity of $\{a_i\}_{i=1}^{\infty}$ that\footnote{It turns out to be more convenient to have the
range of $u_0$ as $1 \leq u_0 \leq n_L$, as opposed to the more conventional range $0 \leq u_o \leq (n_L-1)$.}
\bea
P(u_1n_{L}+u_0) & = & u_1K_L+P(u_0), \ \ u_1 \geq 0, \ \ 1 \leq u_0 \leq n_L.
\eea
With respect to the finite length vector $(a_1,a_2,\cdots,a_{n_L})$, let $Q(\cdot)$ represent the trailing-sum function given by
\bea
Q(s) & = & \sum_{i=n_L-(s-1)}^{n_L} a_i, \  1 \leq s \leq n_L. \label{eq:Q_fun}
\eea
We extend the definitions of $P(\cdot),Q(\cdot)$ by setting $P(0)=Q(0)=0$.
It can be verified that
\ben
\item For $s$ in the range $0 \leq s \leq n_L$, $P(s)  \geq  Q(s)$,
\item $P(\cdot)$ is sub-additive, i.e.,
\bea \label{eq:sub_additive}
P(s+s') & \leq & P(s) + P(s'), \text{  for all $s,s' \geq0$},
\eea
\item the sum $P(s)+Q(s')$ satisfies
\bea \label{eq:P_plus_Q}
P(s)+Q(s') & \leq & P(s+s'), \text{  for all } \ \ s\geq0, \ \ 0 \leq s' \leq n_L.
\eea
\een
We next define the function $P^{(\text{inv})}$ by setting $P^{(\text{inv})}(\nu)$, for $\nu \geq 1 $, to be the smallest integer $s$ such
that $P(s) \geq \nu$, i.e., $P^{(\text{inv})}(\nu)  =  s$, where $s>0$ is uniquely determined from $P(s-1) < \nu \leq
P(s)$.
It can be verified that
\bea
P^{(\text{inv})}(v_1K_L+v_0)  & = & v_1n_L+P^{(\text{inv})}(v_0), \ \ v_1\geq 0, \ \ 1 \leq v_0 \leq K_L.
\eea
As a special case, it follows that
\bea \label{eq:locality_Pinv_temp}
P^{(\text{inv})}(v_1K_L) & = & (v_1-1)n_L+r.
\eea

\subsection{Upper Bound on Minimum Distance Under Exact Information Locality and Uniform Rank Accumulation}

\vspace*{0.1in}

\begin{thm} \label{thm:URA_bound}
Let ${\cal C}$ belong to Class ${\cal U}$.
Then the minimum distance of ${\cal C}$ is upper bounded by
\bea \label{eq:dmin_Pinv_bound}
d_{\min} & \leq & n-P^{(\text{inv})}(K)+1.
\eea
When $K_L \mid K$, the bound takes on the form
\bea \label{eq:dmin_Pinv_bound_BdivK}
d_{\min} & \leq & n-\left( \frac{K}{K_L} \right)  r + 1  - \left( \frac{K}{K_L} -1 \right) (\delta-1).
\eea
\end{thm}

\vspace{0.1in}

\begin{cor} \label{cor:URA_bound}
Let ${\cal C}$ belong to Class ${\cal U}$. Then given $n$, $d_{\min}$, the scalar dimension of $\mathcal{C}$ is upper
bounded by
\bea
K &\leq & P(n-d_{\min}+1). \nonumber
\eea
\end{cor}

\vspace{0.1in}

We say that ${\cal C}$ is distance-optimal if $d_{\min} = n-P^{(\text{inv})}(K)+1$ and rate-optimal if $K =
P(n-d_{\min}+1)$.

The following lemma, which is the analog of the Fact $1$ of \cite{GopHuaSimYek}, for the case of vector codes,  is used in
the proof of Theorem \ref{thm:URA_bound}.

\vspace{0.1in}

\begin{lem} \label{lem:fact_dmin}
Given any set $T\subseteq [n]$ such that $\text{rank}\left(G|_T\right) < K$, we have
\begin{eqnarray} \label{eq:fact_dmin}
d_{\text{min}} \leq n - \mid T \mid
\end{eqnarray}
with equality iff $T \subseteq [n]$ is of largest size such that $\text{Rank}(G|_T) < K$.
\end{lem}

\vspace{0.1in}

\begin{proof}[Proof of Theorem \ref{thm:URA_bound}]  The proof proceeds along the lines of the proof of Theorem
\ref{thm:scalar_info_locality}. We begin by applying Algorithm~\ref{alg:URA_bound} below, to construct a large set $T
\subseteq [n]$ such that $\text{rank}\left(G|_T\right) < K$. The subspaces $\{V_i\}$ appearing in the algorithm correspond to
the column-space of $G|_{S_i}$, where $S_i$ denotes the support of the local codes $\mathcal{C}_i$, i.e.,
\bean
V_i & = & \sum_{\ell \in S_i} W_{\ell},
\eean
where we recall $W_{\ell}$ to be the span of the $\alpha$ thin columns comprising the ${\ell}^{\text{th}}$ thick column of
$G$.

\begin{algorithm}
\caption{Used in the Proof of Theorem \ref{thm:URA_bound}}
\begin{algorithmic}[1]
\State Let  $T_0 = \{\ \}, \ \ j = 0$
\While {$1$}
    \State Pick $i \in \mathcal{L}$ such that $V_i \nsubseteq \sum_{\ell \in T_j}W_{\ell}$
 \If {$\text{Rank}\left(G|_{T_j \cup S_i} \right) < K $}
      \State $j = j + 1$
      \State $T_{j} = T_{j-1} \cup S_i $
    \ElsIf {$\text{Rank}\left(G|_{T_j \cup S_i} \right) = K $}
      \State  Pick any maximal subset $S_{\text{end}}$ of $S_i$ such that  $\text{Rank}\left(G|_{T_j \cup S_{\text{end}} } \right) < K$
      \State $\nu_{\text{end}} = K - \text{Rank}\left(G|_{T_j \cup S_{\text{end}} } \right)$
      \State $ j = j + 1$
      \State $T_{j} = T_{j-1} \cup S_{\text{end}}$
      \State Exit
    \EndIf
\EndWhile
\end{algorithmic} \label{alg:URA_bound}
\end{algorithm}
Let Algorithm \ref{alg:URA_bound} run to $J$ iterations.  Let $s_j, \nu_j$,  $1 \leq j \leq J$, denote the incremental rank and support size respectively, i.e.,
\bea
s_j & = & \mid T_j \mid \ - \ \mid T_{j-1} \mid, \\
\nu_j & = & \text{Rank}\left(  G|_{T_{j}} \right) \ - \ \text{Rank}\left(  G|_{T_{j-1}} \right).
\eea
Note further that $S_{\text{end}}$ is chosen in such a way there exists a choice of thick column in the last stage such that, adding this thick column to the support, will cause the accumulated rank to equal $K$.  Let the integer $\sigma$
represents the amount of overlap in support between the final local code and the prior $J-1$ local codes, i.e.,
\bea
\sigma & = & \mid T_{J-1} \cap S_{\text{end}}\mid .
\eea
Note that under the algorithm, it is possible that the final incremental support $s_J=0$. In addition, the sum $\sigma+s_J$
is upper bounded by $(r-1)$.  This last statement follows by first noting that the local codes are erasure optimal
with q-dim$(\mathcal{C}_i) = r$ and hence as a result if  $\sigma+s_J=r$, then this will result in rank $K$ after the last
step of the algorithm (but rank cannot reach $K$ in the algorithm). The cumulative rank added (assuming $S_{\text{end}}$
and one more thick column) is then upper bounded by
\bea
K &= &  \sum_{j=1}^{J-1} \nu_j + (\nu_J + \nu_{\text{end}}) \\
& \leq & \sum_{j=1}^{J-1} Q(s_j) + \left(P(\sigma+s_J+1)-P(\sigma)\right)  \label{eq:chain_temp_start}\\
& \stackrel{(a)}{\leq} & \sum_{j=1}^{J-1} Q(s_j) + P(s_J+1) \\
& = & \sum_{j=1}^{J-2} Q(s_j) + (Q(s_{J-1})+P(s_J+1)) \\
& \stackrel{(b)}{\leq} & \sum_{j=1}^{J-2} Q(s_j) + P(s_{J-1}+s_J+1))\\
& \leq & P\left(1+\sum_{j=1}^Js_i\right) \label{eq:chain_temp_end},
\eea
where $(a)$ and $(b)$ respectively follow from \eqref{eq:sub_additive} and \eqref{eq:P_plus_Q}.
Hence it follows that
\bea \label{eq:dmin_in_proof_tmp1}
\ \mid T_J\mid +1  \ = \ \sum_{j=1}^Js_j + 1 & \geq & P^{(\text{inv})}(K),
\eea
so that
\bea \label{eq:dmin_in_proof_tmp2}
\mid T_J \mid  & \geq & P^{(\text{inv})}(K)-1 ,
\eea
which from Lemma \ref{lem:fact_dmin}, leads to
\bea
d_{\min} & \leq & n-P^{(\text{inv})}(K)+1.    \label{eq:dmin_in_proof}
\eea
When $K_L \mid K$, this simplifies to
\bean
d_{\min} & \leq & n- P^{(\text{inv})}(K)+1 \\
& \stackrel{(a)}{=} & n-\left(\frac{K}{K_L}-1\right)n_L-r + 1 \\
& = & n-\left(\frac{K}{K_L}-1\right)(r+\delta-1) -r + 1\\
& = & n-\left(\frac{K}{K_L}\right)r+1  -\left(\frac{K}{K_L}-1\right)(\delta-1),
\eean
where $(a)$ follows from \eqref{eq:locality_Pinv_temp}.
\end{proof}

\subsection{Structure of Optimal Codes} \label{sec:URA_rank_bound}

\vspace{0.1in}

\begin{defn}
We will say that the leading-sum function $P(\cdot)$ is strictly sub-additive in the range $[n_L]$, if for any $s \geq 1, s'
\geq 1$ such that $s + s' \leq n_L$, we have $P(s + s')  <  P(s) + P(s')$.
\end{defn}

\vspace*{0.1in}
It can be easily verified that a necessary and sufficient condition for strict sub-additivity is that $a_1>a_2$.
\vspace{0.1in}

\begin{thm} \label{thm:URA_rank_bound}
Let ${\cal C}$ belong to class ${\cal U}$ and also assume that $\mathcal{C}$ is both distance and rate optimal, i.e.,
$d_{\min}= n - P^{\text{inv}}(K) + 1$ and $K = P(n - d_{\min} + 1)$. Also, let $u_1 = \lceil \frac{K}{K_L} \rceil - 1$.
Following observations could be made with regards to the structure of the local codes:
\begin{enumerate}[(a)]
\item
If the leading-sum function $P$ is  strictly sub-additive in the range $[n_L]$, then
\begin{enumerate}[(i)]
\item the local codes $\{ \mathcal{C}_i \}_{i \in \mathcal{L}}$ must all have disjoint supports, i.e., $S_i \cap S_j = \phi,
\ \forall i, j \in \mathcal{L}$ and
\item for distinct $i_1, i_2, \cdots i_{u_1}, i_{u_1+1} \in \mathcal{L}$, it must be that
\begin{eqnarray}
 V_{i_{\ell}} \cap \left( \displaystyle\sum_{\substack{j=1 \\ j \neq \ell}}^{u_1}V_{i_j} \right) & = & \mathbf{0}, \ \forall \ 1 \leq \ell \leq u_1, \text{and}\label{eq:rank_disjoint_1}
\end{eqnarray}
\begin{equation}
   V_{i_{\ell}} \nsubseteq \left( \displaystyle\sum_{\substack{j=1 \\ j \neq \ell}}^{u_1+1}V_{i_j} \right), \ \forall \ 1 \leq \ell \leq u_1 + 1. \label{eq:rank_disjoint_2}
\end{equation}
\end{enumerate}
\item If the scalar dimension $K$ is a multiple of $K_L$, i.e. $K_L|K$, then once again
\begin{enumerate}[(i)]
\item it must be that the local codes $\{ \mathcal{C}_i\}_{i \in \mathcal{L}}$ must all have disjoint supports.
\item Furthermore, for distinct $i_1, i_2, \cdots i_{u_1}, i_{u_1+1} \in \mathcal{L}$, it must be that
\begin{eqnarray}
 \ V_{i_{\ell}} \cap \left(\displaystyle \sum_{\substack{j=1 \\ j \neq \ell}}^{u_1+1}V_{i_j} \right)  = \mathbf{0}, \ \forall \ 1 \leq \ell \leq u_1+1.  \label{eq:rank_disjoint_3}
\end{eqnarray}
\end{enumerate}
\end{enumerate}
\end{thm}

\vspace{0.1in}

\begin{proof}
We are given that
\begin{eqnarray}
 d_{\min} & = & n - P^{\text{inv}}(K) + 1 \label{eq:struc_theorem_vector_tmp1}\\
 K & = & P(n - d_{\min} + 1) \label{eq:struc_theorem_vector_tmp2}.
\end{eqnarray}
Referring back to the proof of Theorem \ref{thm:URA_bound}, we then see that \eqref{eq:dmin_in_proof} is an equality and
hence so must be \eqref{eq:dmin_in_proof_tmp1} and \eqref{eq:dmin_in_proof_tmp2}. As a result, we get that
$\sum_{j = I}^{J}s_j + 1= P^{\text{inv}}(K)$ which implies that
\begin{eqnarray}
P\left(\sum_{j = I}^{J}s_j + 1\right) & = & P\left(P^{\text{inv}}(K)\right) \nonumber \\
& \stackrel{(a)}{=} & P(n - d_{\min} + 1) \nonumber \\
& \stackrel{(b)}{=} & K,
\end{eqnarray}
where $(a)$ and $(b)$ respectively follow from \eqref{eq:struc_theorem_vector_tmp1} and \eqref{eq:struc_theorem_vector_tmp2}.
This means that the chain of inequalities \eqref{eq:chain_temp_start}-\eqref{eq:chain_temp_end} have equality at every step.
The chain is reproduced below for convenience of further analysis:
\bea
K &= &  \sum_{j=1}^{J-1} \nu_j + (\nu_J + \nu_{\text{end}}) \label{eq:chain_temp_start_new}\\
& \stackrel{(i)}{\leq} & \sum_{j=1}^{J-1} Q(s_j) + P(\sigma+s_J+1) - P(\sigma) \\
& \stackrel{(ii)}{\leq} & \sum_{j=1}^{J-1} Q(s_j) + P(s_J+1) \\
& = & \sum_{j=1}^{J-2} Q(s_j) + \left( Q(s_{J-1})+P(s_J+1) \right)\\
&\stackrel{(iii)}{\leq}& \sum_{j=1}^{J-2} Q(s_j) + P(s_{J-1}+s_J+1) \\
&\stackrel{(iv)}{\leq} & P\left(\sum_{j=1}^Js_j+1\right) \\
& = & K.\label{eq:chain_temp_end_new}
\eea
Also let $u_0$ to be such that $P^{\text{inv}}(K) = \sum_{j=1}^Js_j+1 = u_1n_L + u_0$, where $u_1 = \lceil \frac{K}{K_L}
\rceil - 1$. Then note that, since $P^{\text{inv}}(P(\sum_{j=1}^Js_j+1)) = \sum_{j=1}^Js_j+1$, we get that $u_0$ must be in
the range $1 \leq u_0 \leq r$. We now analyze the conditions for various equalities in the above chain for the
two cases: $(a)$ $P$ is strictly sub-additive in the range $[n_L]$ and $(b)$ $K_L|K$.

\begin{enumerate}[(a)]
\item
Assume that $P$ is strictly sub-additive in the range $[n_L]$. Equality in $(ii)$ coupled with the strict sub-additivity of
$P$ implies that $\sigma=0$, implying that the last code added was support disjoint from the rest.
Towards analyzing the equality conditions in $(iii)$ and $(iv)$, we first note that for $s\geq 0$, $\delta \leq s' \leq n_L$,
the equality
\bean
P(s)+Q(s') & = & P(s+s')
\eean
can happen for a strictly sub-additive $P$ iff either $s'=n_L$ or else, $s+s'$ is a multiple of $n_L$.
A little thought will now show that equality can hold in $(iii),(iv)$ iff either $s_j=n_L, 1 \leq j \leq J-1$ or
if there exists $1 \leq \ell \leq J-1$ such that
\bean
s_{\ell}+1+s_{J} & = & n_L \\
s_j & = & n_L, \ 1 \leq j \leq (J-1), \ j \neq \ell.
\eean
In the latter case, this would imply that $\sum_{j=1}^Js_j+1$ is a multiple of $n_L$, which we realize as a contradiction by
noting from the above discussion that $1 \leq u_0 \leq r$. Thus, we get that $s_j=n_L, 1 \leq j \leq J-1$
and $s_J+1=u_0$. It then follows from this that $J=u_1+1$ and the first $J-1$ ($=u_1$) local codes are support
disjoint. From our earlier observation, even the last local code was support disjoint, hence it follows that the $J$ local
codes encountered in the Algorithm~\ref{alg:URA_bound} are support disjoint. Now, the fact that the local codes not
encountered during Algorithm~\ref{alg:URA_bound} also have disjoint supports can be proved in a manner similar to the
proof of Part $(b)$ of Theorem \ref{thm:scalar_info_locality_equality}.

We proceed next to prove the claims in \eqref{eq:rank_disjoint_1} and \eqref{eq:rank_disjoint_2}. Towards this, first note
that equality in $(i)$ in the above chain, coupled with the previous observation that $s_j=n_L, 1 \leq j \leq J-1$
and $s_J+1=u_0$ implies that
\bean
 \nu_j &=& Q(n_L) \ = \ K_L, \ 1 \leq j \leq J-1, \\
 \nu_J + \nu_{\text{end}}&=&P(u_0).
 \eean
Thus if $i_1, i_2, \cdots, i_{u_1}, i_{u_1+1} \in \mathcal{L}$ are such that $\mathcal{C}_{i_j}$
was picked in the $j^{th}$ step, $1 \leq j \leq J = u_1 + 1$,  of Algorithm \ref{alg:URA_bound}, then it follows that
 \bean
 \ V_{i_{\ell}} \cap \left( \displaystyle\sum_{\substack{j=1 \\ j \neq \ell}}^{u_1}V_{i_j} \right)  = \mathbf{0}, \ \forall \ 1 \leq \ell \leq u_1,
 \eean
and that,
\bean
 V_{i_{\ell}} \nsubseteq \left( \displaystyle\sum_{\substack{j=1 \\ j \neq \ell}}^{u_1+1}V_{i_j} \right), \ \forall \ 1 \leq \ell \leq u_1 + 1.
\eean
This is because we have $\text{dim}(\sum_{{j=1 }}^{u_1+1}V_{i_j})=K > u_1K_L$ and $\text{dim}(V_{i_j})=K_L, \ 1 \leq j \leq
u_1+1$. The fact that above observations hold good for any set of ordered indices $i_1, i_2, \cdots, i_{u_1}, i_{u_1+1}$
belonging to $\mathcal{L}$ can be proved in a manner similar to the proof of Part $(c)$ of Theorem
\ref{thm:scalar_info_locality_equality}.

\item
We next consider the case when $K_L|K$ and analyze the conditions for various equalities in the chain
\eqref{eq:chain_temp_start_new}-\eqref{eq:chain_temp_end_new}. First of all, note that since $K_L|K$,  $P^{\text{inv}}(K) =
u_1n_L + r$ (i.e., $u_0 = r$). Next, we note that since $\nu_j \leq K_L, 1 \leq j \leq J - 1$, it follows that
\begin{equation} \label{eq:J_lower_bound}
 J\geq  \frac{K}{K_L}  = u_1+1.
\end{equation}
We will next show that the number $J$ of iterations in the algorithm equals $u_1+1$. Towards this, consider the inequalities
$(iii)$ and $(iv)$ in the above chain. It can be shown that for any $s = q_1n_L + q_0, \ q_1 \geq 0, \ 1 \leq q_0 \leq
n_L $ and $s'$ such that $\delta \leq s' \leq n_L$, the equality
\bean
P(s)+Q(s') & = & P(s+s')
\eean
can happen only if $s+s' \geq (q_1+1)n_L$. It then follows that equalities in $(iii)$ and $(iv)$ can happen only if
$\sum_{j=1}^Js_j+1 = u_1n_L+r \geq (J-1)n_L$ which gives us that $J \leq u_1 +1$.  When coupled with
\eqref{eq:J_lower_bound}, we obtain that $J=u_1+1$ and hence
\bean
\sum_{j=1}^{J-1}s_j + (s_J+1) = (J-1)n_L +r.
\eean
Further, noting that $1\leq s_j \leq n_L, \ 1 \leq j \leq J-1$ and $1 \leq s_J+1 \leq r$, it  follows that
\begin{eqnarray}
 s_j = \begin{cases} n_L, & \mbox{if } 1 \leq j \leq J-1, \\ r-1, & \mbox{if } j=J. \end{cases}
\end{eqnarray}
Also, recall that $\sigma + s_J \leq r-1$, and thus we get that $\sigma=0$. Hence it follows that the $J$
local codes encountered in the Algorithm~\ref{alg:URA_bound} are support disjoint. Moreover,  equality in $(i)$ in the chain,
along with the above observation regarding disjointness also implies that
\bean
\nu_j &=& Q(n_L) \ = \ K_L ; \ 1 \leq j \leq J-1, \\
\nu_J + \nu_{\text{end}}&=&P(r)=K_L.
\eean
We note that this implies that if $i_1, i_2, \cdots, i_{u_1}, i_{u_1+1} \in \mathcal{L}$ are such that $\mathcal{C}_{i_j}$
was picked in the $j^{th}$ step of Algorithm \ref{alg:URA_bound}, then it follows that
\bean
\ V_{i_{\ell}} \cap \left( \displaystyle\sum_{\substack{j=1 \\ j \neq \ell}}^{u_1+1}V_{i_j} \right)  = \mathbf{0}, \ \forall
\ 1 \leq \ell \leq u_1+1,
\eean
The rest of the proof follows along the same lines as the proof of Part $(a)$.
\end{enumerate}

\end{proof}

\section{MSR-Local Codes}\label{sec:msr_local_codes}

We show in this and the next section, how it is possible to construct vector codes with locality, such that the constituent
local codes are regenerating codes, thereby simplifying node repair in two respects.    Node repair can be carried out on
average, by accessing a small number of nodes while downloading an amount of data that is not much more than what the data
node stores. The present section will focus on the construction of optimal codes with information locality in which the local
codes are MSR codes. A more formal definition appears below.

\vspace*{0.2in}

\begin{defn}
Let $\mathcal{C}$ be an $[n,K,d_{\min},\alpha]$ vector code over $\mathbb{F}_q$ possessing $(r,\delta)$ information locality.   Let $G$ be the generator matrix for the code.  Then $\mathcal{C}$ is said to be an MSR-local code with $(r,\delta)$ information locality, provided
\bit
\item the code $\mathcal{C}$ can be punctured so as to yield $m$ local codes $\mathcal{C}_i$ in which the $i^\text{th}$ local code is an $(n_L,r,d)$-MSR code with $n_L=(r+\delta-1)$,
\item and if the $i^\text{th}$ local code has support $S_i$, and $S = \cup_{i=1}^mS_i$, then
\bean
\text{Rank}\left( G|_{S} \right) & = & K.
\eean
\eit
If in addition, $S=[n]$, we will say that $\mathcal{C}$ is an MSR-local code with $(r,\delta)$ all-symbol locality.
\end{defn}

\vspace*{0.2in}

For convenience, we will simply write MSR-local code in place of MSR-local code with $(r,\delta)$ information locality and all-symbol MSR-local code in place of MSR-local code with $(r,\delta)$ all-symbol locality.

Four constructions of MSR-local codes are presented in this section of which the first two are explicit.  The third construction will prove the existence, over large enough fields, of MSR-local codes for a wider range of code parameters than is possible under the two explicit constructions. The fourth construction will establish the existence of all-symbol MSR-local codes whenever $n_L \mid n$.
Throughout this section we will assume that $\delta \geq 3$ as it turns out that $\delta=2$ result in codes where the local codes have trivial regeneration ($\beta=\alpha$).

 \subsection{MSR Codes and Uniform Rank Accumulation}

Let $\mathcal{B}$ be an $((n_L,r,d), (\alpha,\beta),K_L)$ MSR code.  It can be seen either from the rank accumulation profile of MSR codes presented in \cite{ShaRasKumRam_rbt} or from the fact that MSR codes are vector MDS codes that $\mathcal{B}$ has uniform rank accumulation. The rank accumulation profile $\{a_i\}$ is given by
\bean
a_i & = & \left\{ \begin{array}{ll}
\alpha, & 1 \leq i \leq r, \\
 0,  & r+1 \leq i \leq n_L . \end{array} \right.
\eean
It follows that
\bean
K_L & = & \sum_{i=1}^{n_L} a_i \ = \ r \alpha.
\eean
Next, let
\bean
K & = & v_1K_L+v_0 \ = \ v_1 (r\alpha) + v_0 , \  \ 1 \leq v_0 \leq K_L.
\eean
Then we have that
\bean
P^{(\text{inv})}(K) & = & v_1n_L + P^{(\text{inv})}(v_0) \\
& = & v_1n_L + \left \lceil \frac{v_0}{\alpha} \right \rceil \\
& = & v_1 (\delta-1) + v_1r + \left \lceil \frac{v_0}{\alpha} \right \rceil \\
& = & v_1(\delta-1) + \left \lceil \frac{K}{\alpha} \right \rceil \\
& = & \left( \left \lceil \frac{K}{r\alpha} \right \rceil -1 \right) (\delta-1) + \left \lceil \frac{K}{\alpha} \right \rceil .
\eean
It follows that for codes with exact $(r,\delta)$-MSR Locality, we have that
\bea \nonumber
 d_{\min} & \leq & n+1-P^{(\text{inv})}(K) \\ \label{eq:URA_MSR_bound}
 & = & \left( n - \left\lceil \frac{K}{\alpha}\right\rceil + 1\right)  -\left( \left\lceil\frac{K}{r\alpha}\right\rceil -1 \right)(\delta-1).
\eea

\vspace*{0.2in}

\begin{note}
Assuming that it is possible to construct codes that satisfy the bound on $d_{\min}$ in \eqref{eq:URA_MSR_bound} for any given value of $K$, we see that the largest scalar dimension for a given value of $d_{\min}$ results when $\alpha$ divides $K$. All MSR-local  codes presented in this section achieve the bound on $d_{\min}$ of \eqref{eq:URA_MSR_bound} and have $\alpha \mid K$ and hence are rate optimal.
\end{note}

 \vspace*{0.1in}

\subsection{Sum-Parity MSR-Local Codes}

\begin{constr} \label{constr:sum_msr}
Let $\mathcal{C}_0$ be an $((n_L+\Delta, r, d), (\alpha,\beta))$ MSR code with $n_L=(r+\delta-1)$ such that $d \leq r+\delta-2$.  Let $G_0 =
[G_L \mid Q_{\Delta}]$ be a generator matrix of $\mathcal{C}_0$, where $G_L$ and $Q_{\Delta}$ are matrices of size
$(r\alpha \times n_L\alpha)$ and $(r\alpha \times \Delta\alpha)$ respectively.   By Lemma
\ref{lem:punct_regen}, we know that the matrix $G_L$ generates an $((n_L , r, d), (\alpha, \beta))$ MSR code obtained by puncturing $\mathcal{C}_0$ in the symbols associated with the matrix $Q_{\Delta}$.    Next
consider the code $\mathcal{C}$ with generator matrix $G$ given by
\begin{equation}  \label{eq:generator_matrix_MSR_sum_parity}
G = \left [ \begin{array}{ccc|c} G_L &&& Q_\Delta \\ & \ddots && \vdots \\ &&G_L& Q_\Delta \end{array} \right],
\end{equation}
in which both matrices $G_L$ and $Q_{\Delta} $ appear $m \geq 1$ times.
\end{constr}

The theorem below identifies the parameters of the code so constructed and proves the construction to yield MSR-local codes with minimum distance $d_{\min}$ achieving the bound given in \eqref{eq:URA_MSR_bound} whenever $\delta \geq \Delta$.

\vspace{0.1in}

\begin{thm}\label{thm:sum_msr_optimality}
Consider the code $\mathcal{C}$ constructed in Construction~\ref{constr:sum_msr} in which the parameters $\delta,\Delta$ are chosen such that $\delta \geq \Delta $. Then the code $\mathcal{C}$ is an MSR-local code with $(r,\delta)$ information locality, and has
\ben[(a)]
\item length $n = mn_L  + \Delta$ and vector-size parameter $\alpha$,
\item $m$ support-disjoint local codes each of which is MSR with parameters $((n_L , r, d), (\alpha, \beta))$ and possessing generator matrix $G_L$,
\item scalar dimension $K=mr\alpha$ and
\item minimum distance $d_{\min}$ satisfying the $K$-bound \eqref{eq:URA_MSR_bound},
\begin{eqnarray} \label{eq:sum_msr_thm}
 d_{\text{min}} & = & n -  \frac{K}{\alpha} + 1 - \left( \frac{K}{r\alpha}-
1\right)(\delta - 1) \\
& = & \delta+\Delta.
\end{eqnarray}
\een
\end{thm}

\begin{proof} (a),(b),(c) are evident from the construction.    To prove (d), we will first show the equivalence between the two expressions for $d_{\min}$ provided.   Since $K=mr\alpha$ and $n=m(r+\delta-1)+\Delta$, we have that
\begin{eqnarray} \label{eq:sum_msr_proof_temp1}
 n -  \frac{K}{\alpha} + 1 - \left( \frac{K}{r\alpha}-1\right)(\delta - 1) & = &  n-mr+1 - (m-1)(\delta-1), \\
&= & \Delta+\delta.
\end{eqnarray}
Thus, it suffices to show that any non-zero
codeword $\mathbf{c}$ has Hamming weight, $\text{wt}(\mathbf{c}) \geq \delta + \Delta$. First of all, note that if
$\mathbf{c}$ has non-zero components belonging to two or more local codes, then clearly $\text{wt}(\mathbf{c}) \geq 2\delta
\geq \delta + \Delta$, since all local codes themselves have minimum distance $\delta$ and $\delta \geq \Delta$ by hypothesis. Next, consider the complementary case
where the non-zero components of $\mathbf{c}$ are restricted to one of the local codes and the global parities. By
inspecting the generator matrix $G$ given in \eqref{eq:generator_matrix_MSR_sum_parity}, it can be seen that when the
all-zero code symbols corresponding to remaining $(m-1)$ local codes are deleted from each codeword, the resultant punctured codeword lies in the row-space of $G_0 = [G_L \mid Q_{\Delta}]$.
The proof now follows by noting that $G_0$ generates an MSR code of minimum distance $\delta + \Delta$.
\end{proof}

\subsection{Pyramid-Like MSR-Local Codes} \label{sec:pyramid_msr}

The construction below mimics the construction of pyramid codes in \cite{HuaCheLi}, with the difference that we are now dealing with vector symbols in place of scalars and local MSR codes in place of local MDS codes.

\vspace{0.1in}

\begin{constr} \label{constr:pyramid_msr}
Let $\mathcal{C'}$ be an  $((n' = m r + \delta-1 + \Delta, k' = m r ,d), (\alpha, \beta))$ exact repair MSR code  such that
$d \leq n'-\Delta-1 = mr + \delta-2$. Let the (systematic) generator matrix $G'$ of $\mathcal{C'}$ be given by
\begin{eqnarray}
G' & = & \left [ \begin{array}{c|c|c} I_{m r\alpha} & Q & Q' \end{array} \right],
\end{eqnarray}
where $I_{m r\alpha}$ denotes an identity matrix of size $m r\alpha$ and the matrices $Q, Q'$ are respectively of size
$(m r \alpha \times (\delta-1)\alpha)$ and $(m r \alpha \times \Delta\alpha )$. From Lemma \ref{lem:punct_regen},
it follows that the ``punctured'' generator matrix $G'' \triangleq [I_{m r\alpha } \mid Q ]$ generates an $((n' - \Delta, k',d), (\alpha, \beta))$
MSR code; let us call it $\mathcal{C}''$.  Let the matrix $G''$ be represented in block-matrix form as shown below:
\begin{eqnarray}
G'' \ = \  [I_{m r\alpha } \mid Q ] & = & \left [ \begin{array}{cccc} I_{r\alpha} & & & Q_1  \\ &\ddots & & \vdots  \\ &
& I_{r\alpha} & Q_{m}  \end{array} \right],
\end{eqnarray}
where
$Q_i, 1 \leq i \leq m$  are matrices of size $ (r \alpha \times (\delta-1)\alpha)$.   Then the generator matrix $G$ of the desired code $\mathcal{C}$ is obtained by splitting and rearranging the columns of $Q$, as shown below
\begin{eqnarray}
G  & = & \left [ \begin{array}{cccccc|c} I_{r\alpha}  & Q_1  &&& && \\  && \ddots& \ddots && &Q' \\ &&&& I_{r\alpha}  &
Q_{m} & \end{array} \right].
\end{eqnarray}
Clearly, the code $\mathcal{C}$ has $(r, \delta)$ information locality, where the local codes are generated by $[I_{r
\alpha} \mid Q_i ], \ i \in [m]$. It can also be observed that all the local codes are shortened codes of $\mathcal{C}''$
and from Lemma \ref{lem:short_MSR}, it follows that these are all MSR. Thus, we conclude that the code $\mathcal{C}$ is an MSR-local code.
\end{constr}

The theorem below identifies the parameters of the code so constructed, and proves optimality with respect to minimum distance.

\vspace{0.1in}

\begin{thm}\label{thm:pyramid_msr_optimality} Construction \ref{constr:pyramid_msr} gives us an MSR-local code with $(r, \delta)$ information locality, and parameters
\ben[(a)]
\item $K=mr\alpha$, $n=m(r+\delta - 1)+\Delta$, $\alpha = (d - r +1)\beta$,
\item  $d_{\text{min}}$ satisfying
\bean
d_{\min}  & = & n -  \frac{K}{\alpha} + 1 - \left( \frac{K}{r\alpha}-
1\right)(\delta - 1) .
\eean
\een
Thus the code is optimal with respect to the $K$-bound in \eqref{eq:URA_MSR_bound} on minimum distance.
\end{thm}

\begin{proof} As in the proof of Theorem~\ref{thm:sum_msr_optimality}, it suffices to prove that
\begin{eqnarray}
 d_{\text{min}} &= & n -  \frac{K}{\alpha} + 1 - \left( \frac{K}{r\alpha}-
1\right)(\delta - 1) \\
& = & \delta + \Delta.
\end{eqnarray}
However, by inspecting the generator matrices $G$ and $G'$, it is clear that the minimum distance of $\mathcal{C}$ is no
less than that of $\mathcal{C}'$. The theorem now follows by noting that $\mathcal{C}'$ is an MSR code with minimum distance $d_{\min}(\mathcal{C}') = \delta + \Delta$.
\end{proof}

\vspace{0.1in}

\begin{note}
The existence of MSR codes for all possible $[n,k,d]$ has been shown in \cite{CadJafMalRamSuh} and these codes could be used
as the codes $\mathcal{C}_0, \mathcal{C}'$ in the two constructions above.  In terms of known, explicit constructions, the
code $\mathcal{C}_0$ can be picked from the product-matrix class~\cite{RasShaKum_pm} of MSR codes.  The product-matrix
construction requires $d \geq 2r - 2$, which combined with $d \leq r+\delta-2$, leads to the constraint $r \leq \delta$ on
the applicability of this construction in Theorem~\ref{thm:sum_msr_optimality}.  When combined with the requirement $d \leq
mr + \delta-2$, it leads to the constraint $m r \leq \delta$ on the applicability of this construction in
Theorem~\ref{thm:pyramid_msr_optimality}.
\end{note}
\vspace{0.1in}

\subsection{Existence of MSR-Local Codes when $K=mr\alpha$} \label{sec:msr_info_locality_existence}

\vspace{0.1in}

\begin{thm}\label{thm:msr_info_locality_existence}
Given the existence of an $((n_L, r, d), (\alpha, \beta))$ exact-repair MSR code with $n_L=(r+\delta-1)$, there exists an MSR-local code with $(r, \delta)$ information locality over $\mathbb{F}_q$, and $d_{\min}$ achieving the $K$-bound, i.e.,
\begin{eqnarray} \label{eq:msr_info_locality_existence_thm}
 d_{\min} & = & n - \frac{K}{\alpha} + 1 - \left(\frac{K}{r\alpha} - 1\right)(\delta-1),
\end{eqnarray}
with $K=mr \alpha$, for some integer $m \geq 2$, whenever $q > {n \choose mr}$.
\end{thm}

\vspace{0.1in}

\begin{proof}
See Appendix \ref{app:msr_infolocality_existence_proof}.
\end{proof}

\vspace{0.1in}

We note that unlike in Theorems \ref{thm:sum_msr_optimality} and \ref{thm:pyramid_msr_optimality}, there is no constraint here on the
repair degree $d$ involving $r$ and $\delta$, and thus Theorem \ref{thm:msr_info_locality_existence} is applicable for a wider range of
parameters than are Theorems \ref{thm:sum_msr_optimality} and \ref{thm:pyramid_msr_optimality}.

\subsection{Existence of MSR-Local Codes with All-symbol Locality} \label{sec:msr_all_symbol_locality}

\begin{thm}\label{thm:msr_all_symbol}
Given the existence of an $((n_L, r, d), (\alpha, \beta))$ exact-repair MSR code, where $n_L=r+\delta-1$ there exists an $[n, K, d_{\min}, \alpha]$ MSR-local code $\mathcal{C}$ with $(r, \delta)$ all-symbol locality over $\mathbb{F}_q$, such that $d_{\min}$ achieves the $K$-bound with equality
\begin{eqnarray} \label{eq:msr_allsymbol_thm}
 d_{\min} & = & n -  \frac{K}{\alpha}  + 1 - \left( \left\lceil \frac{K}{r\alpha} \right\rceil - 1\right)(\delta-1),
\end{eqnarray}
whenever
\bit
\item $K=\ell \alpha$ for some positive integer $\ell\geq r$,
\item $n=mn_L$ for some positive integer $m\geq \frac{\ell}{r}$ and
\item field size $q > {n \choose \ell}$.
\eit
\end{thm}

\begin{proof}
See Appendix \ref{app:msr_allsymbol_proof}.
\end{proof}

\section{MBR-Local codes} \label{sec:mbr_local_codes}

The present section will focus on the construction of optimal codes with locality in which the local
codes are MBR codes.   More formally, we have
\vspace{.1in}
\begin{defn}
Let $\mathcal{C}$ be an $[n,K,d_{\min},\alpha]$ vector code over $\mathbb{F}_q$ possessing exact $(r,\delta)$ information locality.   Let $G$ be the generator matrix for the code.  Then $\mathcal{C}$ is said to be an MBR-local code with $(r,\delta)$ information locality, provided
\bit
\item the code $\mathcal{C}$ can be punctured so as to yield $m$ local codes $\mathcal{C}_i$ in which the $i^\text{th}$ local code is an $(n_L,r,d)$-MBR code with $n_L=(r+\delta-1)$,
\item and that if the $i^\text{th}$ local code has support $S_i$, and $S = \cup_{i=1}^mS_i$, then
\bean
\text{Rank}\left( G|_{S} \right) & = & K.
\eean
\eit
If in addition, $S=[n]$, we will say that $\mathcal{C}$ is an MBR-local code with $(r,\delta)$ all-symbol locality.
\end{defn}
\vspace{.1in}

As with MSR-local codes, we will write MBR-local code in place of MBR-local code with $(r,\delta)$ information locality and all-symbol MBR-local code in place of MBR-local code with $(r,\delta)$ all-symbol locality.

Two constructions of optimal MBR codes will be presented in this section.  The first is an explicit construction of an MBR-local code that can be applied whenever $K_L \mid K$, where $K_L$ and $K$ denote the scalar dimension of the local and global codes respectively.  In the second construction, we show the existence of all-symbol MBR-local codes whenever $K_L \mid K$ and in addition, $n_L \mid n$. In both cases, optimality is with respect to the bound
 \begin{eqnarray}\label{eq:URA_MBR_bound}
d_{\min} & \leq & n-\frac{K}{K_L}  r +1  - \left( \frac{K}{K_L} -1 \right) (\delta-1),
\end{eqnarray}
appearing in Theorem \ref{thm:URA_bound}.
The MBR codes appearing in both these constructions are the repair-by-transfer MBR codes presented in \cite{ShaRasKumRam_rbt} and described in Example \ref{eg:pentagon} of the present paper.

\subsection{MBR-Local Codes with $(r,\delta)$ Information Locality } \label{sec:mbr_info_locality}

 \begin{constr} \label{constr:mbr_info_locality}

The aim of this construction is to build an optimal MBR-local code $\mathcal{C}$ with $(r,\delta)$ information locality composed of $m$ support-disjoint MBR codes $\{\mathcal{C}_i\}_{i=1}^{m}$ each of which is a repair-by-transfer (RBT) $((n_L,r,d), (\alpha,\beta),K_L)$ MBR code along with $\Delta$ global parity symbols.  The parameters of each RBT MBR code satisfy
\bean
n_L =  r+\delta-1, \ \ \alpha = d  = n_L-1,  \ \  \beta =  1, \  \ K_L = r\alpha - {r \choose 2} .
\eean
Thus the desired global code $\mathcal{C}$ will have length $n= m n_L+\Delta$ and scalar rank $K = m K_L$.
The construction will proceed in three stages:

{\em Stage 1}: Let us define $N_L= {n_L \choose 2}$ and set $\Delta_L = N_L-K_L+1$.  Then in the first stage, a pyramid code $\mathcal{A}$ (see Section~\ref{sec:scalar_local}) with $(K_L,\Delta_L)$-information locality is constructed that is composed of $m$ support-disjoint local codes $\{\mathcal{A}_i\}_{i=1}^{m}$ and $\Delta \alpha$ global parities.   In other words, each of the disjoint local codes $\mathcal{A}_i$ in the pyramid code is an MDS code with parameters $[N_L,K_L,\Delta_L]$ and the overall code ${\cal A}$ has parameters $[m N_L+\Delta \alpha, m K_L,\Delta_L+\Delta \alpha]$.

 {\em Stage 2}: In the second stage, the $N_L$ symbols that correspond to the $i^\text{th}$ local code $\mathcal{A}_i$ are regarded as MDS-coded symbols with parameters $[N_L,K_L,\Delta_L]$ and used to construct a repair-by-transfer $((n_L,r,d), (\alpha,\beta),K_L)$  MBR code.

{\em Stage 3:} In the final stage, the $\Delta \alpha$ global parities are collected into $\Delta$ groups of $\alpha$ symbols each with each group representing the contents of one of the $\Delta$ global parity nodes.

This completes the construction.

\end{constr}

\vspace{0.1in}
An example of Construction \ref{constr:mbr_info_locality}, is illustrated in Figure \ref{fig:mbr_local}.
\vspace{0.1in}

\begin{thm}\label{thm:mbr_info_locality}

Construction \ref{constr:mbr_info_locality} results in an MBR-local code $\mathcal{C}$ with $(r,\delta)$ information locality composed of $m$ support-disjoint local codes $\{\mathcal{C}_i\}_{i=1}^{m}$ each of which is a repair-by-transfer (RBT) $((n_L,r,d), (\alpha,\beta),K_L)$ MBR code along with $\Delta$ global parity symbols.  The parameters of each RBT MBR code satisfy
\bean
n_L =  r+\delta-1, \ \ \alpha = d  = n_L-1,  \ \  \beta =  1, \  \ K_L = r\alpha - {r \choose 2} .
\eean
Thus code $\mathcal{C}$ has length $n= m n_L+\Delta$, scalar dimension $K = m K_L$ and $d_{\min}$ satisfying the upper bound in \eqref{eq:URA_MBR_bound} given by
 \begin{eqnarray}
d_{\min} & = & n-\frac{K}{K_L}  r +1  - \left( \frac{K}{K_L} -1 \right) (\delta-1) \\
& = & \Delta+\delta.
\end{eqnarray}
\end{thm}

\begin{proof}  All claims in the theorem are clear with the exception of the claim concerning the minimum distance.
 Since $K_L \mid K$, an upper bound on $d_{\min}$ from \eqref{eq:URA_MBR_bound} is given by
\begin{eqnarray}
d_{\min} & \leq & n-\frac{K}{K_L}  r +1  - \left( \frac{K}{K_L} -1 \right) (\delta-1) \label{eq:thm_rbt_MBR_temp2} \\
& = & (n-m r+1) -(m-1)(\delta-1), \\
& = & n-m(r+\delta-1) + \delta \\
& = & \delta + \Delta . \label{eq:MBR_equality}
\end{eqnarray}
To show that the code satisfies the above bound with equality, it suffices to show that any pattern of $\delta + \Delta - 1$ erasures can be corrected by the code. Towards this, we note that the scalar code $\mathcal{A}$ employed in Construction \ref{constr:mbr_info_locality} has minimum distance given by
\bean
D_{\min} & = & \Delta_L+\Delta \alpha \ = \ \Delta \alpha + {\delta-1 \choose 2} +1 .
\eean
Given any pattern of $\delta + \Delta - 1$ erasures in the vector code $\mathcal{C}$, by using the structure of the repair-by-transfer MBR local codes, we will evaluate the number of scalar symbols of the pyramid code,
$\mathcal{A}$, that are erased and show that this number is at most $D_{\text{min}} - 1$. This would imply that the
pyramid code, $\mathcal{A}$, can recover from this many erasures and thus, so can the vector code
$\mathcal{C}$.

As a first step we note that at least $\delta-1$ vector code symbols out of any pattern of $\delta+\Delta-1$ erased vector code symbols, come from the union of the local codes. We will now argue that the maximum number of scalar symbols lost on erasing the first $\delta-1$ vector code symbols from the union of the local codes is  ${\delta -1 \choose 2}$.

Assume that a given pattern of $\delta+\Delta-1$ vector code symbol erasures is given. For this pattern, we further restrict ourselves to the first $\delta-1$ vector code symbols erased from the union of the local codes. Let $\gamma_i, 1\leq i\leq m$ be the number of code-word symbols erased from the $i^{th}$ local code. Note that $0 \leq \gamma_i \leq \delta -1 < n_L,\ 1 \leq i \leq m$ and $\sum_{i=1}^{m}\gamma_i=\delta-1$. Next, we observe that the number of scalar symbols lost when $\delta-1$ nodes are lost from the union of local codes equals
\bean
\sum_{i=1}^{m} {\gamma_i \choose 2} \leq {\sum_{i=1}^{m}\gamma_i \choose 2} = {\delta -1 \choose 2}.
\eean
The loss of $\Delta$ more vector code symbols can cause the loss of at-most $\Delta\alpha$ more scalar code symbols. Thus we have that the maximum number $L$ of scalar code symbols lost as a result of $\Delta +\delta -1$ erasures, is
given by \bean
L \leq \Delta \alpha + {\delta-1 \choose 2}  \ = \ \Delta_L - 1 + \Delta \alpha  \ = \ D_{\min} -1,
\eean
and the result follows.
\end{proof}

\vspace{0.1in}

\subsection{Existence of Optimal MBR-Local Codes with $(r,\delta)$ All-Symbol Locality} \label{sec:mbr_allsymbol_locality}

\vspace{0.1in}
 \begin{constr} \label{constr:mbr_allsymbol_existence}

The aim of this construction is to build an optimal MBR-local code $\mathcal{C}$ with $(r,\delta)$-all-symbol locality composed of $m$ support-disjoint MBR codes $\{\mathcal{C}_i\}_{i=1}^{m}$ each of which is a repair-by-transfer (RBT) $((n_L,r,d), (\alpha,\beta),K_L)$ MBR code, using optimal scalar all-symbol locality codes, whose existence has been shown in Theorem \ref{thm:scalar_allsymbol_existence}.  The parameters of each RBT MBR code satisfy
\bean
n_L =  r+\delta-1, \ \ \alpha = d  = n_L-1,  \ \  \beta =  1, \  \ K_L = r\alpha - {r \choose 2} .
\eean
The desired global code $\mathcal{C}$ will have length $n= mn_L$ and  scalar dimension $K$ that is assumed to be a multiple $K = \ell K_L$ for some positive integer $\ell \leq m$.
The construction will proceed in two stages:

{\em Stage 1}: Let us define $N_L= {n_L \choose 2}$ and set $\Delta_L = N_L-K_L+1$.  Then in the first stage, a scalar code $\mathcal{A}$ with $(K_L,\Delta_L)$ all-symbol locality, of length $mN_L$, dimension $K=\ell K_L$ and which moreover, is optimal with respect to the bound on minimum distance is assumed to be given. The existence of such a code is shown in Theorem \ref{thm:scalar_allsymbol_existence}. As $K_L \mid K$, it also follows from Theorem \ref{thm:scalar_info_locality_equality}, that $\mathcal{A}$ is composed of $m$ support-disjoint MDS codes local $\{\mathcal{A}_i\}_{i=1}^{m}$ with parameters $[N_L,K_L,\Delta_L]$.  The global scalar code ${\cal A}$ has parameters $[m N_L, \ell K_L, \Delta_L+(m-\ell)N_L]$.

 {\em Stage 2}: In the second stage, the $N_L$ symbols that correspond to the $i^{\text{th}}$ local code $\mathcal{A}_i$ are regarded as MDS-coded symbols with parameters $[N_L,K_L,\Delta_L]$ and used to construct a repair-by-transfer $((n_L,r,d), (\alpha,\beta),K_L)$  MBR code.

This completes the construction.

\end{constr}

\vspace{0.1in}
An example of Construction \ref{constr:mbr_allsymbol_existence}, is illustrated in Figure \ref{fig:mbr_local_all_symbol}.
\vspace{0.1in}

\begin{thm}\label{thm:mbr_allsymbol_existence}

Construction \ref{constr:mbr_allsymbol_existence} results in an MBR-local code $\mathcal{C}$ with $(r,\delta)$ all-symbol locality composed of $m$ support-disjoint local codes $\{\mathcal{C}_i\}_{i=1}^{m}$ each of which is a repair-by-transfer (RBT) $((n_L,r,d), (\alpha,\beta),K_L)$ MBR code along with $\Delta$ global parity symbols.  The parameters of each RBT MBR code satisfy
\bean
n_L =  r+\delta-1, \ \ \alpha = d  = n_L-1,  \ \  \beta =  1, \  \ K_L = r\alpha - {r \choose 2} .
\eean
Thus code $\mathcal{C}$ has  $n= mn_L$ and scalar rank $K = \ell K_L$ for some positive integer $\ell \leq t$ and $d_{\min}$ satisfying the upper bound in \eqref{eq:URA_MBR_bound} given by
 \begin{eqnarray}
d_{\min} & = & n-\frac{K}{K_L}  r +1  - \left( \frac{K}{K_L} -1 \right) (\delta-1) \\
& = & (m-\ell)n_L+\delta.
\end{eqnarray}
\end{thm}

\begin{proof}
See Appendix \ref{app:mbr_allsymbol_existence_proof}.
\end{proof}

\vspace{0.1in}

\begin{note}
It follows from conditions derived in Theorem \ref{thm:URA_rank_bound} of Section \ref{sec:URA_rank_bound}, that both constructions presented in this section are rate optimal. We also note that the rank accumulation profile is strictly sub-additive for MBR codes and hence it is not possible to construct any $d_{\min}$ optimal MBR-local codes without support disjoint local codes.
\end{note}

\section{Bound on $d_{\min}$ Based on Quasi-Dimension} \label{sec:kappa_bound}
In this section, we derive bounds on minimum distance for vector codes possessing locality. Unlike in Section \ref{sec:locality_vec_codes}, we do not assume exact locality in this section.

The case for locality in vector codes with $\delta = 2$ has been previously considered in \cite{PapDim}, where it was shown that
$d_{\text{min}}$, under $(r, \delta = 2)$-all-symbol locality, is upper bounded by
\bea
 d_{\text{min}} & \leq & n - \left \lceil \frac{K}{\alpha} \right \rceil + 1 - \left(
\left \lceil \frac{K}{r\alpha}\right\rceil - 1\right). \label{eq:dimakis_delta_eq_2}
\eea

\subsection{Bound on Minimum Distance for Vector Codes with Locality} \label{sec:dmin_bound_vector_codes}

We obtain below an upper bound on the minimum distance of a vector code, in the presence of $(r, \delta)$ information
locality, that holds for all $\delta \geq 2$ and which when specialized to the case $\delta=2$, is in general, tighter than the bound in \eqref{eq:dimakis_delta_eq_2}.

\vspace{0.1in}

\begin{thm}\label{thm:kappa_bound}
Consider an $[n, K, d_{\text{min}},\alpha, \kappa]$ vector code $\mathcal{C}$ with $(r, \delta)$ information locality. Let the local punctured codes of length at-most $n_L=r+\delta-1$ and minimum distance at-least $\delta$ be $\left\{ \mathcal{C}_i, i \in \mathcal{L}\right\}$ and their supports be $S_i, i \in \mathcal{L}$ respectively. Set $S={\cup_{i \in \mathcal{L}}S_i}$.  Then, the minimum distance of $\mathcal{C}$ is upper bounded by
\begin{eqnarray}
 d_{\text{min}} & \leq & n - |\mathcal{I}_0| + 1 - \left(\left \lceil \frac{|\mathcal{I}_0|}{r}\right \rceil - 1\right)(\delta
- 1)  \ \ \ \ (\mathcal{I}_0\text{-bound}) \label{eq:locality_tighter} \\
 & \leq & n - \kappa + 1 - \left(\left \lceil \frac{\kappa}{r}\right \rceil - 1\right)(\delta
- 1), \label{eq:locality_looser_kappa}  \ \ \ \ \ \ \ \ \ \ (\kappa\text{-bound})\\
& \leq & n - \left\lceil \frac{K}{\alpha}\right\rceil + 1 - \left(\left \lceil \frac{K}{r\alpha}\right\rceil- 1\right)(\delta
- 1), \label{eq:locality_looser_dimakis} \ \ \ (K\text{-bound})
\end{eqnarray}
where $\mathcal{I}_0$ is a minimum cardinality information set for $\mathcal{C}|_S$ or equivalently, equals the quasi-dimension of
$\mathcal{C}|_S$, i.e.,
$|\mathcal{I}_0| = \text{q-dim}\left(\mathcal{C}|_S\right)$.
\end{thm}

\vspace{0.1in}

\vspace{0.1in}

\begin{proof} See Appendix \ref{app:kappa_bound_proof}.

\end{proof}

\begin{note}
For the same set of code parameters $[n, K, d_{\text{min}},\alpha]$ , it is possible to construct codes having different values of $\kappa$ and $\mathcal{I}_0$.  Thus $\kappa$ and $\mathcal{I}_0$ depend upon finer structural details of both global and local code.
Thus while the $K$-bound is a global bound on the minimum distance, the $\kappa$ and $\mathcal{I}_0$-bounds may be regarded as {\em structure}-dependent bounds.
\end{note}

\vspace{0.1in}

The MSR-local code constructions presented in Section \ref{sec:msr_local_codes} of this paper achieve the $K$-bound given in \eqref{eq:locality_looser_dimakis} with equality.   On the other hand, the MBR-local code constructions in Section \ref{sec:mbr_local_codes} are examples of codes that do not meet the $K$-bound but which achieve the $\mathcal{I}_0$-bound.   It also turns out that amongst the codes constructed using Construction \ref{constr:dimakis}, there are examples of codes that achieve the $\kappa$ bound but which do not achieve the $K$-bound.

We now discuss necessary conditions for achieving equality in the $\mathcal{I}_0$-bound, whenever
$r$ divides $|\mathcal{I}_0| $ where $\mathcal{I}_0$ is as defined in the statement of
Theorem \ref{thm:kappa_bound}, and is a reference to a minimum cardinality information set for
the restriction of the code $\mathcal{C}$ to the union of the support of the local codes.

\vspace{0.1in}

\begin{thm}\label{thm:kappa_bound_equality}
Consider an $[n, K, d_{\text{min}},\alpha,\kappa]$ vector code $\mathcal{C}$ having $(r, \delta)$ information locality that is optimal with respect to the $\mathcal{I}_0$-bound.  We assume further that $r \mid |\mathcal{I}_0 |$ and set $\frac{|\mathcal{I}_0|}{r}=t$.  Let $\{ \mathcal{C}_i \}_{i \in \mathcal{L}}$, be the set of all local codes whose length is at most $r+\delta -1$ and distance is at least $\delta$ and let $\{S_i\}_{i \in \mathcal{L}}$ respectively be their supports.  Then
\begin{enumerate}[(a)]
\item  $\mathcal{C}_i$ is an $[ r + \delta - 1, \text{dim}(\mathcal{C}_i) \leq r \alpha, \delta, \alpha, r ]$ erasure optimal code $\forall \ i \in \mathcal{L}$, and
\item for distinct $i_1, i_2 \in \mathcal{L}$, the codes $\mathcal{C}_{i_1}$ and  $\mathcal{C}_{i_2}$ are support disjoint, i.e.,
\bea
S_{i_1} \cap S_{i_2} &=& \phi,
\eea
\end{enumerate}
\end{thm}

\begin{proof} The proof is along the same lines as the proof of Theorem \ref{thm:scalar_info_locality_equality} for the scalar case and hence omitted.
\end{proof}

\vspace{0.1in}

\subsection{Optimal Vector Codes with Locality through Stacking} \label{sec:stacking_info_locality}

By stacking $\alpha$ scalar codes with locality, one trivially obtains a vector code with locality.    More specifically, let $\mathcal{B}$ be a scalar local code having parameters $[n,k,d_{\min}]$ and let $\mathcal{C}$ be the code obtained by stacking $\alpha$ codewords, each drawn from $\mathcal{B}$, to obtain a codeword from $\mathcal{C}$.    It is straightforward to verify that $\mathcal{C}$ has code parameters $[n,K,d_{\min},\kappa]$ where $K=k \alpha$ and $\kappa=k$.   By numerically comparing the bounds \eqref{eq:bound_info_locality} and \eqref{eq:locality_looser_dimakis} on minimum distance in the scalar and vector case respectively, it follows that the vector code is optimal with information or all-symbol locality depending whenever the scalar code is optimal in the same sense. This observation is made formal in the following theorem.

\vspace{0.1in}

\begin{thm} \label{thm:vec_local_stacking}
For any set of parameters $n, \kappa, \alpha, r, \delta$, $\delta \geq 2$, the following optimal vector codes can be constructed via stacking:
\begin{enumerate}
\item an explicit $(r, \delta)$ information locality code,
\item an explicit $(r, \delta)$ all-symbol locality code, whenever $n=\lceil \frac{\kappa}{r} \rceil
(r+\delta-1)$,
\item a non-explicit $(r, \delta)$ all-symbol locality code, whenever $(r+\delta-1)|n$ and the field size $q \geq \kappa n^{\kappa}$.
\end{enumerate}
The minimum distance of all the three classes of codes is given by the equality in the $K$-bound.
\end{thm}

\begin{proof}
Each of the three classes of the codes are respectively obtained by stacking $\alpha$ independent codewords of the following classes of optimal scalar codes with locality:
\begin{enumerate}
\item pyramid codes which are explicit $(r, \delta)$ information locality codes,
\item parity splitting codes which are explicit $(r, \delta)$ all-symbol locality codes, whenever $n=\lceil \frac{k}{r} \rceil
(r+\delta-1)$,
\item $(r, \delta)$ all-symbol locality codes whose existence is known, whenever $(r+\delta-1) \mid n$ and the field size $q \geq kn^k$.
\end{enumerate}
\end{proof}

\subsection{A Class of Optimal and Explicit $(r, \delta)$ All-Symbol Locality Vector Codes}

An explicit construction for obtaining optimal codes with $(r, \delta)$ all-symbol locality, for the case of $\delta = 2$ was
presented in \cite{PapDim}. This construction has a straightforward extension for any arbitrary $\delta \geq 2$. The construction as well as its extension are described below.

\vspace{0.1in}

\begin{constr} \label{constr:dimakis}
Pick a message matrix $M$ of size $r \times k', k' > 0,$ such that $(r+1)|n$. The
encoding takes place in two stages:  in the first stage, the message matrix is encoded by a product code, wherein the row code is
chosen as an $[n, k']$ MDS code and the column code is a parity-check code. Let $\mathbf{c}'$ denote the $((r+1) \times n)$
codeword array obtained after the first stage. In the second stage, the set of all columns of $\mathbf{c}'$ is partitioned into
contiguous sets of size $r+1$ each and the $i^{\text{th}}, \ 1 \leq i \leq (r+1)$ row of each partition is cyclically
by $(i-1)$ scalar symbols. As an illustration, the first partition after the cyclic permutation would look like
\begin{eqnarray}
 \left[ \begin{array}{ccccc}
  c_{1,1} & c_{1,2} & c_{1,3} & \ldots & c_{1,r+1} \\
  c_{2,r+1} & c_{2,1} & c_{2,2} & \ldots & c_{2,r} \\
  \vdots & &&& \vdots \\
  c_{r+1,2} & c_{r+1,3} & c_{r+1,4} & \ldots &c_{r+1,1}
 \end{array} \right].
\end{eqnarray}
Note that the code has $\alpha = (r + 1)$. In this encoded structure, it is clear that every column of the codeword array is
locally covered by an $[r + 1,r, 2]$ code, in which the remaining $r$ columns come from the partition to which the column belongs to. For example, if the first column of the first partition fails, this column can be recovered by accessing the entire contents of all the remaining $r$ columns of this partition and computing the parities. Thus the code has $(r, \delta = 2)$-all-symbol locality. It was also shown that
whenever $(r+1) \nmid k'$, the minimum distance of the code is given by the equality condition in
\eqref{eq:locality_looser_dimakis}.

{\em The extension} In the extension, the requirement  $(r+1)|n$ is replaced by the condition $(r + \delta -1)|n$.  To construct a code with $(r, \delta)$-all-symbol locality for the general case $\delta \geq 2$, we just use an $[r + \delta - 1,
r, \delta]$ MDS code as the column code in place of the parity-check code, when building the product
code.  Thus in the second stage, the set of columns of  $\mathbf{c}'$ are
partitioned into contiguous sets of size $(r + \delta -1)$ each and a similar cyclic permutation is carried out as before. It is not hard to verify that this code has $(r, \delta)$-all-symbol locality. We summarize the above discussion about the construction,
its optimality and rate in the theorem below.
\end{constr}

\vspace{0.1in}

\begin{thm} \label{thm:dimakis_extn}
Given any $n, \kappa, r, \delta$ such that $(r + \delta - 1)|n$, Construction \ref{constr:dimakis} yields a vector code with
$(r, \delta)$-all-symbol locality, where the parameter $k'$ in Construction \ref{constr:dimakis} is chosen as
\begin{eqnarray} \label{eq:thm_dimakis_extn_temp1}
k' & = & \kappa + \left( \left\lceil\frac{\kappa}{r}\right\rceil - 1\right)(\delta -1 ).
\end{eqnarray}
The code has $\alpha = r+\delta - 1$, rate $\rho = \frac{k'r}{n\alpha}$ and minimum distance, $d_{\text{min}}$ achieving the $\kappa$-bound with equality,  given by \eqref{eq:locality_looser_kappa}.
\end{thm}
\begin{proof}
First of all, note that the parameter $k'$ in \eqref{eq:thm_dimakis_extn_temp1} is chosen such that the code obtained
through Construction \ref{constr:dimakis} will have quasi-dimension $\kappa$. This follows from the fact if $k' = \theta(r +
\delta - 1) + \gamma, \ \theta >0, \ 0 \leq \gamma \leq r+ \delta - 2$, then the quasi dimension of the code obtained is
given by
\begin{eqnarray}
\kappa & = & \begin{cases}
              \theta r + \gamma, \ \text{if } 0 \leq \gamma \leq r -1, \\
             \theta r + r,  \ \text{if } r \leq \gamma \leq r +\delta - 2.
	    \end{cases}
\end{eqnarray}
Clearly, since each row of the code is $[n, k']$ MDS , any $n - k'$ erasures can be tolerated by the overall vector code
and hence the minimum distance of the code can be lower bounded as
\begin{eqnarray}
d_{\text{min}} & \geq & n - k' + 1 \\
	      & = & n - k + 1 - \left( \left\lceil\frac{\kappa}{r}\right\rceil - 1\right)(\delta -1 ).
\label{eq:thm_dimakis_extn_temp2}
\end{eqnarray}
Combining \eqref{eq:thm_dimakis_extn_temp2} with Theorem \ref{thm:kappa_bound},
the claim about
the minimum distance follows.
\end{proof}

\vspace{0.1in}

\begin{note} The following comments are in order regarding Construction \ref{constr:dimakis}.
\begin{enumerate}

\item Construction \ref{constr:dimakis}, whenever $\delta >2$, is an instance where the $K$-bound on the minimum
distance is not always achievable. For example, if $\delta =3, r = 4$ and
$k' = 9 = (r + \delta - 1) + 3$, then $\kappa = r + 3 = 7$. But $K = k' r = 36$ and hence $\frac{K}{\alpha} = \frac{36}{6} = 6 <
\kappa$. Thus, \eqref{eq:locality_looser_kappa} is strictly tighter than \eqref{eq:locality_looser_dimakis}.

\item Unlike the optimal codes presented in Theorem \ref{thm:vec_local_stacking}, the rate of the optimal code obtained via Construction \ref{constr:dimakis} can
be less than $\frac{\kappa}{n}$;  the parameters given above constitute an example.

\end{enumerate}
\end{note}

\bibliographystyle{IEEEtran}
\bibliography{CLG}

\begin{thebibliography}{10}
\providecommand{\url}[1]{#1}
\csname url@samestyle\endcsname
\providecommand{\newblock}{\relax}
\providecommand{\bibinfo}[2]{#2}
\providecommand{\BIBentrySTDinterwordspacing}{\spaceskip=0pt\relax}
\providecommand{\BIBentryALTinterwordstretchfactor}{4}
\providecommand{\BIBentryALTinterwordspacing}{\spaceskip=\fontdimen2\font plus
\BIBentryALTinterwordstretchfactor\fontdimen3\font minus
  \fontdimen4\font\relax}
\providecommand{\BIBforeignlanguage}[2]{{%
\expandafter\ifx\csname l@#1\endcsname\relax
\typeout{** WARNING: IEEEtran.bst: No hyphenation pattern has been}%
\typeout{** loaded for the language `#1'. Using the pattern for}%
\typeout{** the default language instead.}%
\else
\language=\csname l@#1\endcsname
\fi
#2}}
\providecommand{\BIBdecl}{\relax}
\BIBdecl

\bibitem{PraKamLalKum}
N.~Prakash, G.~M. Kamath, V.~Lalitha, and P.~V. Kumar, ``{Optimal linear codes
  with a local-error-correction property},'' in \emph{{Proc. IEEE Int. Symp.
  Inf. Theory (ISIT)}}, Cambridge, MA, Jul. 2012, pp. 2776--2780.

\bibitem{PraKamLalKum_arxiv}
\BIBentryALTinterwordspacing
------, ``{Optimal linear codes with a local-error-correction property},''
  2012. [Online]. Available: \url{arXiv:1202.2414}
\BIBentrySTDinterwordspacing

\bibitem{HuaSimXu_etal_azure}
C.~Huang, H.~Simitci, Y.~Xu, A.~Ogus, B.~Calder, P.~Gopalan, J.~Li, and
  S.~Yekhanin, ``{Erasure coding in windows azure storage},'' in
  \emph{{Proc.2012 USENIX Annual Technical Conference (ATC)}}, Boston, MA,
  2012, pp. 15--26.

\bibitem{hadoop}
\BIBentryALTinterwordspacing
``{Hadoop}.'' [Online]. Available: \url{http://hadoop.apache.org}
\BIBentrySTDinterwordspacing

\bibitem{hdfs_raid}
D.~Borthakur, R.~Schmit, R.~Vadali, S.~Chen, and P.~Kling, ``{HDFS RAID},''
  \emph{Tech talk. Yahoo Developer Network}, 2010.

\bibitem{DimGodWuWaiRam}
A.~Dimakis, P.~Godfrey, Y.~Wu, M.~Wainwright, and K.~Ramchandran, ``{Network
  coding for distributed storage systems},'' \emph{IEEE Trans. Inf. Theory},
  vol.~56, no.~9, pp. 4539--4551, Sep. 2010.

\bibitem{GopHuaSimYek}
P.~Gopalan, C.~Huang, H.~Simitci, and S.~Yekhanin, ``{On the Locality of
  Codeword Symbols},'' \emph{IEEE Trans. Inf. Theory}, vol.~58, no.~11, pp.
  6925--6934, Nov. 2012.

\bibitem{ShaRasKumRam_rbt}
N.~B. Shah, K.~V. Rashmi, P.~V. Kumar, and K.~Ramchandran, ``{Distributed
  Storage Codes With Repair-by-Transfer and Nonachievability of Interior Points
  on the Storage-Bandwidth Tradeoff},'' \emph{IEEE Trans. Inf. Theory},
  vol.~58, no.~3, pp. 1837--1852, Mar. 2012.

\bibitem{RasShaKum_pm}
K.~V. Rashmi, N.~B. Shah, and P.~V. Kumar, ``{Optimal Exact-Regenerating Codes
  for Distributed Storage at the MSR and MBR Points via a Product-Matrix
  Construction},'' \emph{IEEE Trans. Inf. Theory}, vol.~57, no.~8, pp.
  5227--5239, Aug. 2011.

\bibitem{ShaRasKumRam_ia}
N.~B. Shah, K.~V. Rashmi, P.~V. Kumar, and K.~Ramchandran, ``{Interference
  Alignment in Regenerating Codes for Distributed Storage: Necessity and Code
  Constructions},'' \emph{IEEE Trans. Inf. Theory}, vol.~58, no.~4, pp.
  2134--2158, Apr. 2012.

\bibitem{SuhRam}
C.~Suh and K.~Ramchandran, ``{Exact-repair MDS code construction using
  interference alignment},'' \emph{IEEE Trans. Inf. Theory}, vol.~57, no.~3,
  pp. 1425--1442, Mar. 2011.

\bibitem{PapDimCad}
\BIBentryALTinterwordspacing
D.~S. Papailiopoulos, A.~G. Dimakis, and V.~R. Cadambe, ``{Repair Optimal
  Erasure Codes through Hadamard Designs},'' {2011}. [Online]. Available:
  \url{arXiv:1106.1634}
\BIBentrySTDinterwordspacing

\bibitem{TamWanBru}
\BIBentryALTinterwordspacing
I.~Tamo, Z.~Wang, and J.~Bruck, ``{Zigzag Codes: MDS Array Codes with Optimal
  Rebuilding},'' 2011. [Online]. Available: \url{arXiv:1112.0371}
\BIBentrySTDinterwordspacing

\bibitem{CadJafMalRamSuh}
\BIBentryALTinterwordspacing
V.~R. Cadambe, S.~A. Jafar, H.~Maleki, K.~Ramchandran, and C.~Suh,
  ``{Asymptotic Interference Alignment for Optimal Repair of MDS codes in
  Distributed Data Storage},'' 2012. [Online]. Available:
  \url{http://www.mit.edu/~viveck/resources/Research/asymptotic\_storage.pdf}
\BIBentrySTDinterwordspacing

\bibitem{ElrRam}
S.~El~Rouayheb and K.~Ramchandran, ``Fractional repetition codes for repair in
  distributed storage systems,'' in \emph{Proc. 48th Annual Allerton Conf. on
  Communication, Control, and Computing (Allerton)}, Urbana-Champaign, IL, Sep.
  2010, pp. 1510 --1517.

\bibitem{ShuHu}
\BIBentryALTinterwordspacing
K.~Shum and Y.~Hu, ``{Cooperative Regenerating Codes},'' 2012. [Online].
  Available: \url{arXiv:1207.6762}
\BIBentrySTDinterwordspacing

\bibitem{HuYuLiLeeLui}
Y.~Hu, C.~Yu, Y.~Li, P.~Lee, and J.~Lui, ``{NCFS: On the practicality and
  extensibility of a network-coding-based distributed file system},'' in
  \emph{{Proc. IEEE Int. Symp. Netw. Coding (NetCod)}}, 2011, pp. 1--6.

\bibitem{HuCheLeeTan}
Y.~Hu, H.~Chen, P.~Lee, and Y.~Tang, ``{NCCloud: applying network coding for
  the storage repair in a cloud-of-clouds},'' in \emph{{Proc.10th USENIX Conf.
  File and Storage Technologies (FAST)}}, 2012, pp. 265--272.

\bibitem{DumBie}
A.~Duminuco and E.~Biersack, ``{A practical study of regenerating codes for
  peer-to-peer backup systems},'' in \emph{{Proc. 29th IEEE Int. Conf.
  Distributed Computing Systems (ICDCS)}}, 2009, pp. 376--384.

\bibitem{HuaCheLi}
C.~Huang, M.~Chen, and J.~Li, ``{Pyramid codes: Flexible schemes to trade space
  for access efficiency in reliable data storage systems},'' in \emph{{Proc.
  6th IEEE Int. Symposium on Network Computing and Applications (NCA)}}, 2007,
  pp. 79--86.

\bibitem{OggDat}
F.~Oggier and A.~Datta, ``{Self-repairing homomorphic codes for distributed
  storage systems},'' in \emph{{Proc. IEEE Int. Conf. Comput. Communications
  (INFOCOM)}}, Shanghai, China, Apr. 2011, pp. 1215--1223.

\bibitem{SilRawVis}
\BIBentryALTinterwordspacing
N.~Silberstein, A.~S. Rawat, and S.~Vishwanath, ``{Error Resilience in
  Distributed Storage via Rank-Metric Codes},'' 2012. [Online]. Available:
  \url{arXiv:1202.0800}
\BIBentrySTDinterwordspacing

\bibitem{PapDim}
D.~S. Papailiopoulos and A.~G. Dimakis, ``{Locally repairable codes},'' in
  \emph{{Proc. IEEE Int. Symp. Inf. Theory (ISIT)}}, Cambridge, MA, Jul. 2012,
  pp. 2771--2775.

\bibitem{PapLuoDimHuaLi}
D.~Papailiopoulos, J.~Luo, A.~Dimakis, C.~Huang, and J.~Li, ``{Simple
  regenerating codes: Network coding for cloud storage},'' in \emph{{Proc. IEEE
  Int. Conf. Comput. Communications (INFOCOM)}}, Mar. 2012, pp. 2801--2805.

\bibitem{HanLas}
J.~Han and L.~A. Lastras-Montano, ``{Reliable Memories with Subline
  Accesses},'' in \emph{{Proc. IEEE Int. Symp. Inf. Theory (ISIT)}}, Jun. 2007,
  pp. 2531--2535.

\bibitem{BlaHafHet}
\BIBentryALTinterwordspacing
M.~Blaum, J.~L. Hafner, and S.~Hetzler, ``{Partial-MDS Codes and their
  Application to RAID Type of Architectures},'' 2012. [Online]. Available:
  \url{arXiv:1205.0997}
\BIBentrySTDinterwordspacing

\bibitem{RawKoySilVis}
\BIBentryALTinterwordspacing
A.~S. Rawat, O.~O. Koyluoglu, N.~Silberstein, and S.~Vishwanath, ``{Optimal
  Locally Repairable and Secure Codes for Distributed Storage Systems},'' 2012.
  [Online]. Available: \url{arXiv:1210.6954}
\BIBentrySTDinterwordspacing

\bibitem{sathiamoorthy}
\BIBentryALTinterwordspacing
M.~Sathiamoorthy, M.~Asteris, D.~Papailiopoulos, A.~G. Dimakis, R.~Vadali,
  S.~Chen, and D.~Borthakur, ``Xoring elephants: Novel erasure codes for big
  data,'' 2013. [Online]. Available: \url{arXiv:1301.3791}
\BIBentrySTDinterwordspacing

\bibitem{BlaFarTil}
M.~Blaum, P.~G. Farrell, and H.~C. van Tilborg, ``{Array Codes},''
  \emph{Handbook of Coding Theory}, vol.~2, pp. 1855--1909, 1998.

\bibitem{BlaBraBruMen}
M.~Blaum, J.~Brady, J.~Bruck, and J.~Menon, ``{{EVENODD:} An efficient scheme
  for tolerating double disk failures in {RAID} architectures},'' \emph{IEEE
  Trans. Comput.}, vol.~44, no.~2, pp. 192--202, 1995.

\bibitem{BlaBraBruMenVar}
M.~Blaum, J.~Brady, J.~Bruck, J.~Menon, and A.~Vardy, ``{The EVENODD code and
  its generalization},'' \emph{High Performance Mass Storage and Parallel I/O},
  pp. 187--208, 2001.

\bibitem{CorEngGoeGrcKleLeoSan}
P.~Corbett, B.~English, A.~Goel, T.~Grcanac, S.~Kleiman, J.~Leong, and
  S.~Sankar, ``{Row-diagonal parity for double disk failure correction},'' in
  \emph{{Proc. 3rd USENIX Conf. File and Storage Technologies (FAST)}}, 2004,
  pp. 1--14.

\bibitem{For}
G.~{Forney Jr}, ``{Concatenated Codes},'' 1966, {MIT} Press, Cambridge.

\bibitem{Dum}
I.~Dumer, ``{Concatenated codes and their multilevel generalizations},''
  \emph{Handbook of Coding Theory}, vol.~2, pp. 1911--1988, 1998.

\bibitem{BenDivMonPol}
S.~Benedetto, D.~Divsalar, G.~Montorsi, and F.~Pollara, ``{Serial concatenation
  of interleaved codes: Performance analysis, design, and iterative
  decoding},'' \emph{Information Theory, IEEE Transactions on}, vol.~44, no.~3,
  pp. 909--926, 1998.

\bibitem{BlaRot}
M.~Blaum and R.~Roth, ``On lowest density {MDS} codes,'' \emph{IEEE Trans. Inf.
  Theory}, vol.~45, no.~1, pp. 46--59, 1999.

\bibitem{BlaBruVar}
M.~Blaum, J.~Bruck, and A.~Vardy, ``{MDS} array codes with independent parity
  symbols,'' \emph{IEEE Trans. Inf. Theory}, vol.~42, no.~2, pp. 529--542,
  1996.

\bibitem{Alo}
N.~Alon, ``{Combinatorial Nullstellensatz},'' \emph{Combinatorics, Probability
  and Computing}, 1999.

\end{thebibliography}

\appendices

\section{Proof of Theorem \ref{thm:scalar_info_locality}} \label{app:scalar_info_locality_proof}

We will make use the following lemma (see \cite{GopHuaSimYek}) in the proof:

\vspace{0.1in}

\begin{lem} \label{lem:scalar_fact_dmin}
Given any set $T\subseteq [n]$ such that $\text{rank}\left(G|_T\right) \leq k-1$, we have
\begin{eqnarray} \label{eq:fact_dmin}
d_{\text{min}} \leq n - |T|,
\end{eqnarray}
with equality iff $T \subseteq [n]$ is of largest size such that $\text{Rank}(G|_T) = k-1$.
\end{lem}

\vspace{0.1in}

\begin{proof}[Proof of Theorem \ref{thm:scalar_info_locality}]

Assume that we are given an $[n,k,d_{\min}]$ scalar code $\mathcal{C}$ which  has $(r, \delta)$ information
locality. As in the proof of Theorem $5$ in \cite{GopHuaSimYek}, we will construct using Algorithm 1 below, a set $T
\subseteq [n]$ such that $\text{rank}\left(G|_T\right) \leq k-1$ and then apply Lemma \ref{lem:scalar_fact_dmin} to get the
required result. We define $\forall i \in \mathcal{L}, {V_i} = \text{Col}(G|_{S_i})$, the column space of the matrix
$G|_{S_i}$.

\begin{algorithm}
\caption{Used in the Proof of Theorem \ref{thm:scalar_info_locality}} \label{alg:scalar_info_locality}
\begin{algorithmic}[1]
\State Let  $T_0 = \{\ \}, \ \ j = 0$
\While {$1$}
    \State Pick $i \in \mathcal{L}$ such that $V_i \nsubseteq \text{Col}(G|_{T_j})$
    \If {$\text{Rank}\left(G|_{T_j \cup S_i} \right) \leq k-1 $}
      \State $j = j + 1$
      \State $T_{j} = T_{j-1} \cup S_i $
    \ElsIf {$\text{Rank}\left(G|_{T_j \cup S_i} \right) = k $}
      \State Pick any maximal subset $S''$ of $S_i$ such that  $\text{rank}\left(G|_{T_j \cup S''} \right) = k-1$
      \State $S_{\text{end}} =S_{i} $
      \State $ j = j + 1$
      \State $T_{j} = T_{j-1} \cup S'' $
      \State Exit
    \EndIf
\EndWhile
\end{algorithmic}
\end{algorithm}

With respect  to the $j^{\text{th}}$ iteration of Algorithm 1, note that as long as
$\text{Rank}\left(G|_{T_j} \right) \leq k-1 $, one can always pick an $i \in \mathcal{L}$ such that $V_i \nsubseteq
\text{Col}(G|_{T_j})$. Let the algorithm exit after $J$ iterations, i.e., $j = J$ when the algorithm exits. Note that
necessarily, $\text{Rank}\left(G|_{T_{J-1} \cup S_{\text{end}}} \right) = k$. Clearly, as each local code has dimension at
most $r$, it must then be true that
\begin{eqnarray} \label{eq:scalar_info_locality_proof1}
 J & \geq & \left\lceil \frac{k}{r} \right\rceil.
\end{eqnarray}
Next, for $j \in [J]$, let
\begin{eqnarray}
s_j & = & |T_j| - |T_{j-1}| \ , \nonumber \\
\nu_j & = & \text{dim}(\mathcal{C}|_{T_j}) - \text{dim}(\mathcal{C}|_{T_{j-1}}) \ = \ \text{Rank}\left(G|_{T_j}\right)- \text{Rank}\left(G|_{T_{j-1}}\right). \label{eq:localitythmproof_temp6}
\end{eqnarray}
We claim that for $j \in [J-1]$,
\begin{eqnarray} \label{eq:locality_proof_condition1}
s_j & \geq & \nu_j + (\delta - 1).
\end{eqnarray}
This follows because the local codes have minimum distance at least equal to $\delta$ and hence, the last $\delta-1$
code symbols of a local code are known given the rest of the code symbols. We also have that,
\begin{eqnarray}
 s_J & \geq & \nu_J. \label{eq:locality_proof_condition2}
\end{eqnarray}
Summing up, we obtain that
\begin{eqnarray}
 |T_J| \ = \  \sum_{1 = 1}^{J}s_j & \geq &  \sum_{1 = 1}^{J}\nu_j \ + \ (J-1)(\delta - 1) \label{eq:scalar_locality_proof_temp2}\\
 & \geq &  \sum_{1 = 1}^{J}\nu_j \ + \ \left(\left\lceil \frac{k}{r} \right\rceil-1\right)(\delta - 1)
\label{eq:scalar_locality_proof_temp3}\\
& = & k-1 \ + \ \left(\left\lceil \frac{k}{r} \right\rceil-1\right)(\delta - 1),
\label{eq:scalar_locality_proof_temp4}
\end{eqnarray}
The result then follows from an application of Lemma \ref{lem:scalar_fact_dmin}.

\end{proof}

\section{Proof of Theorem \ref{thm:scalar_info_locality_equality}} \label{app:scalar_info_locality_equality_proof}

(a) For $d_{\min}$ to achieve \eqref{eq:bound_info_locality}, we need that \eqref{eq:scalar_locality_proof_temp2} and \eqref{eq:scalar_locality_proof_temp3} must be satisfied with equality.  We reproduce the chain of inequalities for the case when $r \mid k$ here for the sake of convenience:
\begin{eqnarray}
 |T_J| \ = \  \sum_{j = 1}^{J}s_j & \stackrel{(i)}{\geq} &  \sum_{j = 1}^{J}\nu_j \ + \ (J-1)(\delta - 1) , \\
 & \stackrel{(ii)}{\geq} &  \sum_{j = 1}^{J}\nu_j \ + \ \left( \frac{k}{r} -1\right)(\delta - 1) , \\
& = & k-1 \ + \ \left( \frac{k}{r} -1\right)(\delta - 1).
\end{eqnarray}
Optimality of the code with respect to the bound on $d_{\min}$ in \eqref{eq:bound_info_locality} implies that equality holds
in both inequalities, $(i)$ and $(ii)$, above. Equality in (i), coupled with the fact that $s_j \geq \nu_j + \delta-1,  \ 1
\leq j \leq J-1$ and $s_J \geq \nu_J$ imply that
\bean
s_j & = & \begin{cases} \nu_j +\delta-1, & \mbox{if } 1 \leq j \leq J-1, \\ \nu_J, & \mbox{if } j=J.\end{cases}
\eean
Equality in (ii), coupled with the fact that $J \geq \frac{k}{r} $ gives us that
\bean
J &=& \frac{k}{r},\\
\implies k &=& Jr.
\eean
Thus,
\bean
k -1 & = & \sum_{1 = 1}^{J-1}\nu_j  + \nu_J \ = \ (J-1)r + r-1.
\eean
Coupled  with the fact that $ \nu_j \leq r,  \ 1 \leq j \leq J-1$ and $\nu_J \leq r-1$, we obtain
\bean
\nu_j & = & \begin{cases} r, & \mbox{if } 1 \leq i \leq J-1, \\ r-1, & \mbox{if } i=J.\end{cases}
\eean
This leads to
\bea \label{eq:proof_struc_thm_scalar_temp1}
s_j & = & \begin{cases} r +\delta-1, & \mbox{if } 1 \leq i \leq J-1, \\ r-1, & \mbox{if } i=J.\end{cases}
\eea
Also note that since Algorithm \ref{alg:scalar_info_locality}, can start with any local code, all local codes have length $n_L=r+\delta-1$ and dimension $r$ and hence are $[r+\delta-1,r,\delta]$ MDS codes thus proving (a).

(b) First of all, note from \eqref{eq:proof_struc_thm_scalar_temp1} that with the possible exception of the last code
picked by the algorithm, all the remaining codes must have pairwise, disjoint support.  The equality $\nu_J=r-1$ implies that
the increase in dimension resulting from replacing $S''$ by $S_{\text{end}}$ would $r$ as opposed to $(r-1)$ with
$S''$. Now if the last code overlapped with the union of the rest, since the last $\delta -1$ code symbols of any local
code are dependent on its first $r$ symbols, the increase in dimension due to $S_{\text{end}}$ cannot be $r$. It follows that
the last code $\mathcal{S}_J$ picked by the algorithm must also have support that is disjoint from the support of any of the
prior local codes $\mathcal{C}_i$, $1 \leq i \leq J-1$ picked by the algorithm. Thus in summary, all the codes picked by
Algorithm \ref{alg:scalar_info_locality} must have pairwise, disjoint support.

It remains only to show that the collection of local codes in ${\cal L}$ are support disjoint, even when we include local codes belonging to ${\cal L}$, but not encountered by Algorithm~\ref{alg:scalar_info_locality}.   We note first that for $i_1, i_2 \in \mathcal{L}$, $i_1 \neq i_2$, if $V_{i_1}\neq V_{i_2}$, we could pick $\mathcal{C}_{i_1}$ and $\mathcal{C}_{i_2}$ (not necessarily in that order) as the first two codes used in Algorithm \ref{alg:scalar_info_locality} and hence obtain that
$$S_{i_1} \cap S_{i_2} = \phi.$$

Thus we only have consider the case when two distinct local codes, say ${\cal C}_{i_1}, {\cal C}_{i_2}$ are such that
$V_{i_1}=V_{i_2}$.   If we were to run the algorithm, beginning with the code ${\cal C}_{i_1}$, at the conclusion of the
algorithm, we would obtain a set $T_J$ such that $\text{Rank}(G|_{T_J}) = k-1$ and $|T_J|=k-1 + (\frac{k}{r}-1)(\delta -
1)$.  This follows, because as the code is optimal, every instance of the algorithm must yield a set $T_J$ satisfying
\eqref{eq:scalar_locality_proof_temp2} and \eqref{eq:scalar_locality_proof_temp3} with equality.  Let $\mathcal{T}$ be the
set of all local codes encountered by the algorithm in this case. We could then replace $T_J$ by
 \bean
 T_J'=T_J \cup S_{i_2},
 \eean
and since $V_{i_1}=V_{i_2}$, we would obtain that even with this augmentation of support, $\text{Rank}(G|_{T_J'}) = k-1$. If
$S_{i_2} \nsubseteq \cup_{j\in \mathcal{T}}S_j $, we would have have that $ |T_J'| > k-1 + (\frac{k}{r}-1)(\delta - 1)$.  But
this would imply a tighter bound on $d_{\min}$, which would contradict our assumption that the code under consideration
satisfies the earlier bound on $d_{\min}$. On the other hand, if $S_{i_2} \subseteq \cup_{j\in \mathcal{T}}S_j $, then we can
start the algorithm with  ${\cal C}_{i_2}$, and pick the same sequence of codes we picked when we started the
algorithm beginning with $\mathcal{C}_{i_1}$. This can be done as $V_{i_1}=V_{i_2}$. But this
would lead us to conclude that the support $S_{i_2}$ of the code $\mathcal{C}_{i_2}$ and
 \bean
\cup_{j\in \mathcal{T}}S_j  \setminus S_{i_1},
 \eean
were disjoint, which by $S_{i_2} \subseteq \cup_{j\in \mathcal{L}}S_j $, would then force
\bean
S_{i_1} & = & S_{i_2}.
\eean
But this would imply that $\mathcal{C}_{i_1}=\mathcal{C}_{i_2}$ contradicting our earlier assumption that the codes were distinct.  It follows that all the local codes $\mathcal{C}_i, i \in \mathcal{L}$ have disjoint supports.

(c) To prove the third assertion in the theorem, first of all note that if $i_1, i_2, \cdots, i_{t} \in
\mathcal{L}$ are such that $\mathcal{C}_{i_j}$ is picked by the Algorithm \ref{alg:scalar_info_locality} in the
$j^{\text{th}}$ step, then it must be that
\bea \label{eq:rank_disjoint_scalar}
V_{i_{t}} \cap \left( \displaystyle\sum_{j=1 }^{t-1}V_{i_j} \right) =\mathbf{0}
\eea
simply because
\bean
\text{dim}\left( \displaystyle\sum_{j=1 }^{t}V_{i_j} \right) &=& k \ = \ rt.
\eean
and $ \text{dim}\left(V_{i} \right) \leq r, \forall \  i \in \mathcal{L}$. It remains to be proved that the assertion is
true for any set of ordered set of indices  $i_1,i_2, \cdots, i_{t}$ belonging to $\mathcal{L}$. If we can show that
it is possible for the algorithm to proceed in such a way that $\mathcal{C}_{i_j}$ is the local code picked in the $j$th
step, then from \eqref{eq:rank_disjoint_scalar}, the assertion would be proved.

Assume that we are given an ordered set of indices  $i_1,i_2, \cdots, i_{t} \in \mathcal{L}$ and that it is not possible for the algorithm to proceed in such a way that $\mathcal{C}_{i_j}$ is the local code picked in the $j$th step.  Let $m$ the first index at which the algorithm runs into trouble, i.e., the algorithm is unable to pick the code $\mathcal{C}_{i_m}$ during the $m$th iteration.  This can happen only if
\bean
V_{i_m} \subseteq  \sum_{j=1}^{m-1}V_{i_j} .
\eean
Suppose next, that the algorithm were allowed to proceed beyond the $(m-1)$th step by picking indices without restriction (as
opposed to picking them from the given sequence $\{i_{m+1},i_{m+2},..i_t\}$) until eventually a set $T_J$ was obtained such
that $\text{Rank}(G|_{T_J}) =  k-1$ and $|T_J|=k-1 + (t-1)(\delta - 1)$. We could then replace $T_J$ by
 \bean
 T_J'=T_J \cup S_{i_m},
 \eean
 while maintaining the property that $\text{Rank}(G|_{T_J'}) = k-1$.  From part (a) of the theorem  we have that all local
codes are support disjoint and it follows therefore that
 $$|T_J'|=|T_J|+r+\delta -1 > k-1 + (t-1)(\delta - 1).$$
This would imply a tighter bound on $d_{\min}$ and we have once again arrived at a contradiction.   It follows that given any set of indices  $i_1,i_2, \cdots, i_{t} \in \mathcal{L}$, it is possible for the algorithm to proceed in such a way that $\mathcal{C}_{i_j}$ is the local code picked in the $j^{\text{th}}$ step.  This concludes the proof.

\section{Proof of Theorem \ref{thm:scalar_allsymbol_existence}}  \label{app:scalar_allsymbol_existence_proof}

The proof here is similar to the proof of Theorem $17$ of \cite{GopHuaSimYek}.  We will state a couple of definitions and a lemma from \cite{GopHuaSimYek}, which will be useful in proving the existence of optimal codes with all-symbol locality.

\vspace{0.1in}

\begin{defn}[$k$-core \cite{GopHuaSimYek}]
Let $L$ be a subspace of $\mathbb{F}_q^n$ and $S \subseteq [n]$ be a set of size $k$. $S$ is said to be a $k$-core
for $L$ if for all non-zero vectors $v \in L$, $\text{Supp}(\mathbf{v}) \nsubseteq S$.
\end{defn}

\vspace{0.1in}

In our application to codes, $L$ will frequently denote the dual $\mathcal{C}^{\perp}$ of a linear code $\mathcal{C}$ of length $n$.  In this setting, we note that saying that $S$ is a $k$-core for $\mathcal{C}^{\perp}$ is equivalent to saying that the $k$ columns of the generator matrix of the code $\mathcal{C}$ corresponding to $S$ are linearly independent.

\vspace{0.1in}

\begin{defn}[Vectors in General Position Subject to $L$ \cite{GopHuaSimYek}]
 Let $L$ be a subspace of $\mathbb{F}_q^n$. Let $G = [\mathbf{g}_1, \cdots , \mathbf{g}_n ]$ be a $(k \times n)$ matrix over $\mathbb{F}_q$. The columns of $G$, $\{\mathbf{g}_i \}_{i=1}^n$ are said to be in general position with respect to $L$ if:
\begin{itemize}
 \item Row space of $G$, denoted by Row$(G) \subseteq L^\perp$.
 \item For all $k$-cores $S$ of $L$, we have $\text{Rank}(G|_S ) = k$.
\end{itemize}
\end{defn}

\vspace{0.1in}

\begin{lem}[Lemma 14 of \cite{GopHuaSimYek}] \label{exist}
 Let $n,k,q$ be such that $q > kn^k$. Let $L$ be a subspace of $\mathbb{F}_q^n$ and $ 0 < k \leq n - \text{dim}(L)$.
 Then $\exists$ a set of vectors $\{\mathbf{g}_i\}_{i=1}^n$ in $\mathbb{F}_q^k$
that are in general position with respect to $L$.
\end{lem}

\vspace{0.1in}

Using the above lemma, we will now prove the existence of optimal $(r,\delta)$ codes when $(r+\delta-1) \mid n$.

\vspace{0.1in}

Let $n = (r+\delta-1)t$. Let $\{P_1, \cdots, P_t\}$ be a partition of $[n]$, where $|P_i| = r+\delta-1, 1\leq i \leq t$. Let $Q_i$ be the parity check matrix of an $[r+\delta-1,r,\delta]$ MDS code with support $P_i$. Consider the block-diagonal matrix
\begin{equation}
H'_{t(\delta -1) \times n}
	      = \begin{bmatrix}
               Q_1 &    &  &  \\
	          & Q_2 &  &  \\
	        &   & \ddots &  \\
	          &    &  &Q_t  \\
             \end{bmatrix}.
\end{equation}

Let $L=\text{Rowspace}(H')$. For any $[n,k,d_{\min}]$ code with $(r,\delta)$ all-symbol locality,  \eqref{eq:bound_info_locality} along with the fact that for any $(r,\delta)$ all-symbol locality code, one has that
\bean
d_{\min} - \delta & \geq & 0,
\eean
gives us that
\begin{equation}
 n-k  \geq  \left\lceil \frac{k}{r} \right\rceil(\delta -1) \geq \frac{k}{r}(\delta -1). \nonumber \\
\end{equation}
Rearranging the above equation, we get
\begin{equation}
k \leq n - t(\delta-1),
\end{equation}
which by Lemma \ref{exist}, tells us that the first requirement for the existence of a $k$-core has been met. Thus, from Lemma \ref{exist}, for large enough field size,  $\exists \ \{\mathbf{g}_i\}_{i=1}^{n}$, $\mathbf{g}_i \in \mathbb{F}_q^k$ which are in general position with respect to $L$. Now consider the code $\mathcal{C}$ whose generator matrix $G_{k\times n}=[\mathbf{g}_1 \cdots \mathbf{g}_n]$. Clearly, $\mathcal{C}$ is an $[n,k]$ code.
Also  $\mathcal{C}$ has $(r,\delta)$ all-symbol locality,  as each co-ordinate of $\mathcal{C}$ has  an $[r+\delta-1,\leq r,\geq \delta]$  punctured code checking on it, whose parity check matrix contains one of $Q_i \ ; \ 1 \leq i \leq t$.

It remains to prove that $d_{\min}(\mathcal{C}) = d_{\min}$ is given by the equality condition in \eqref{eq:bound_info_locality}. Towards this, we will show that for any set $S \subseteq [n]$ such that $\text{Rank}(G|_S) \leq k-1$, it must be true that
\begin{equation} \label{eq:S_cardinality_allsymbloc}
 |S| \leq k-1+(\delta-1)\left( \displaystyle \lceil \frac{k}{r} \rceil -1 \right).
\end{equation}
Assuming this to be the case, it follows that the minimum distance of the code $\mathcal{C}$ satisfies
\begin{equation}
 d_{\min}= n- \max_{\substack{S \subseteq [n] \\ \text{Rank}(G|_S) \leq k - 1}} |S| \geq
n-k+1 - \left( \displaystyle \lceil \frac{k}{r} \rceil -1 \right)(\delta-1).
\end{equation}
Combining the above equation and \eqref{eq:bound_info_locality}, it follows that the code ${\cal C}$ has the distance given in the theorem statement.

It remains to prove \eqref{eq:S_cardinality_allsymbloc}.  Towards this, let $S \subseteq [n]$ be such that $\text{Rank}(G|_S)
\leq k-1$. Clearly, $S$ does not contain a $k$-core, $k$ being the dimension of $\mathcal{C}$.   Note that if $S$ were such
that \bean
|S\cap P_i| \leq r \ \ \forall \ i  \in [t],
\eean
it would then follow that $S$ contains a $k$-core.  Thus there exists $\text{some } i \in [t]$ such that $|P_i \cap S| \geq r+1$.

Define
\begin{equation*}
 b_{\ell} := \left|\left\{ i \in [t] \arrowvert \  |P_i \cap S| = r+\ell \right\}\right| \ \ \ 1 \leq \ell \leq \delta -1.
\end{equation*}
For $1 \leq \ell \leq \delta -1$, consider the set, $S_{\ell}$, obtained from $S$ by dropping $\ell$ elements of $S$ from
each of the $b_{\ell}$ sets $\{P_i | \ |P_i \cap S| = r + \ell\}$. Clearly,
the set $\displaystyle \cap_{1 \leq \ell \leq \delta -1}S_{\ell}$ is an $|S|-b_1-2b_2-\cdots-(\delta-1)b_{\delta-1}$ core contained in $S$ and thus as $S$ does not contain a $k$-core,
\begin{equation} \label{eq:nc1}
 |S|-(\delta-1)(\sum_{i=1}^{\delta-1} b_i)  \leq |S|-b_1-2b_2-\cdots-(\delta-1)b_{\delta-1} \leq k-1.
\end{equation}
Also if we pick $r$ co-ordinates from each $P_i$ which is such that $|P_i \cap S| \geq r + 1$, we get a $(r)(\sum_{i=1}^{\delta-1} b_i)$-core contained in $S$.
Thus as $S$ does not contain a $k$-core,
\begin{equation} \label{eq:nc2}
 \sum_{i=1}^{\delta-1} b_i \leq \left\lfloor \frac{k-1}{r} \right\rfloor = \left \lceil \frac{k}{r} \right \rceil -1.
\end{equation}
 Combining \eqref{eq:nc1} and \eqref{eq:nc2}, we have
\begin{equation}
 |S| \leq k - 1 + \left(  \left \lceil \frac{k}{r} \right \rceil -1 \right)(\delta-1).
\end{equation} 

\section{Proof of Theorem \ref{thm:msr_info_locality_existence}} \label{app:msr_infolocality_existence_proof}

 We begin with a useful lemma.

\begin{lem}[Combinatorial Nullstellensatz (Thm. 1.2 of \cite{Alo})] \label{lem:null}
Let $\mathbb{F}$ be a field, and let $f = f(x_1,\dots,x_n)$ be a polynomial in $\mathbb{F}[x_1,\dots,x_n]$.
Suppose the degree deg$(f)$ of f is expressible in the form $\displaystyle\sum_{i=1}^{n}t_i$, where each $t_i$ is a non-negative integer and suppose
that the coefficient of the monomial term $\displaystyle\prod_{i=1}^{n}x_i^{t_i}$ in $f$ is nonzero. Then, if $S_1, \dots ,S_n$ are subsets of
$\mathbb{F}$ with sizes $|S_i| $ satisfying $|S_i| > t_i$ , then there exist elements $s_1 \in S_1, s_2 \in S_2, \dots, s_n \in S_n$ such that
\begin{eqnarray*}
 f(s_1,s_2,\dots,s_n) & \neq & 0.
\end{eqnarray*}
\end{lem}

\vspace{0.15in}

Set
\bean
\nu & = & mr+(m-1)(\delta-1), \\
\Delta & = & n- mn_L.
\eean
Let $\mathcal{C}_L$ be an $((n_L,r,d), (\alpha,\beta))$ MSR code and let $G_L$ be an $(r\alpha \times n_L\alpha)$ generator matrix for $\mathcal{C}_L$. Let $m \geq 2$ be an integer.  Consider the vector code
$\mathcal{C}$, generated by the $(m r\alpha \times n \alpha)$ generator matrix
\begin{eqnarray} \label{eq:last_label_for_Q}
 G & = & \left[ \begin{array}{ccc|c} G_L & &  & \\ & \ddots & & Q \\ & & G_L & \end{array} \right],
\end{eqnarray}
in which $Q$ is an $(mr \alpha \times \Delta \alpha)$ matrix over $\mathbb{F}_q$ and the block matrix $G_L$ appears $m$ times along the diagonal. It is clear that the code $\mathcal{C}$ has length $n$ and also has $m$ support-disjoint local MSR codes, each generated by $G_L$.  Thus, $\mathcal{C}$ is an MSR-local code with $(r, \delta)$ information locality.   As a result,  the right hand side of \eqref{eq:msr_info_locality_existence_thm} with $K=mr\alpha$, is an upper bound on the minimum distance, $d_{\min}$ of $\mathcal{C}$.

We will now show that it is possible to pick the matrix $Q$ such that the minimum distance of $\mathcal{C}$ is indeed given
by \eqref{eq:msr_info_locality_existence_thm}. Towards this, we treat the entries of $Q$ as indeterminates, with the $(i,j)^{\text{th}}$ element of $Q$ denoted by $x_{ij}$.  We write $G(\mathbf{X})$ to indicate that the generator matrix is a function of the matrix
$\mathbf{X} \triangleq [x_{ij}]$ and $G(Q)$, its evaluation at $\mathbf{X} = Q$.

The expression in \eqref{eq:msr_info_locality_existence_thm}  for $d_{\min}$ can be rewritten in the form
\begin{eqnarray}
 d_{\min} & = & n - mr + 1 - \left(m - 1\right)(\delta-1),
\end{eqnarray}
so that $n-d_{\min}+1=mr+(m-1)(\delta-1)=\nu$.  It follows that it suffices to show that all the $(K \times \nu \alpha)$ sub-matrices of $G(Q)$ that are obtained by selecting a set of $\nu$ thick columns drawn from the generator matrix $G$, are of full rank.

Let $S_i$ denote the support of the $i^{\text{th}}$ local code, $i \in [m]$.   If the $\nu$ thick columns have indices chosen from $\cup_{i=1}^mS_i$, then it can be shown that the full rank condition is always satisfied simply because, one is forced to pick at least $r$ thick columns with indices from the support $S_i$ of each local code. It follows that it is enough to ensure that
\begin{eqnarray}
\text{rank}(G(Q)|_T)& = & K, \ \forall \ T \in \mathcal{T},  \\
\text{where } \mathcal{T} & =&\{T\subseteq [n]:  \ |T|=mr, \ \ |T\cap S_i | \leq r, \ \forall \ i \in [m]\}.
\end{eqnarray}
Note that $G(Q)|_T$ is square $(K \times K)$ matrix.    Next, consider the set of polynomials
$f_T(\mathbf{X})=\det(G(\mathbf{X})|_T), \ T \in \mathcal{T}$, where $\det(A)$ denotes the determinant of the square matrix
$A$. Also, let  $f(\mathbf{X}) = \prod_{T \in \mathcal{T}}f_T(\mathbf{X})$. The degree of any individual indeterminate
$x_{ij}$ in $f(\mathbf{X})$ is at most $|\mathcal{T}|$. Noting that
\bean
|\mathcal{T}| & \leq &  {n \choose mr},
\eean
by applying Lemma
\ref{lem:null}, we conclude that there exists a matrix $Q$ such that $f(Q) \neq 0$, whenever the base field $\mathbb{F}_q$
has size $q > {n \choose mr}$. 

\section{Proof of Theorem \ref{thm:msr_all_symbol}} \label{app:msr_allsymbol_proof}

The proof is similar to the proof of Theorem \ref{thm:scalar_allsymbol_existence}. Let the integer $\ell$ be defined from
\bean
K & = & \alpha \ell.
\eean
The idea is to first construct a partial parity-check matrix consisting of $m$ disjoint local parity matrices, which ensures that the locality constraints are satisfied. Then we show that one can always add extra rows to this partial parity-check matrix in order to guarantee the optimum minimum distance.

Let $H_L$  denote the parity check matrix of an $((n_L,r,d), (\alpha,\beta))$ MSR code $\mathcal{C}_L$. By Lemma \ref{lem:MSR_kappa}, we know that $\mathcal{C}_L$ is vector-MDS. Thus the dual code $\mathcal{C}_L^{\perp}$, generated by $H_L$, will also be vector-MDS.   Consider the $(m(\delta-1)\alpha \times n\alpha)$ matrix $H_0$ given by
\begin{eqnarray}
 H_0 & = & \begin{bmatrix}
		    H_L &  & \dots &  \\
		     & H_L & \dots &  \\
		    \vdots & & \ddots & \vdots \\
		     &  & \dots & H_L \\
		    \end{bmatrix},
\end{eqnarray}
in which the matrix $H_L$ appears $m$ times along the diagonal.  Also, let $\mathcal{C}_0$ denote the code whose parity check matrix is $H_0$. Next, consider the code $\mathcal{C}$ whose parity check matrix $H$ is obtained by augmenting $H_0$ with additional rows as shown below:
\begin{eqnarray}
H & = & \left[ \begin{array}{c} H_0  \\ H_1\end{array}\right],
\end{eqnarray}
where $ H_1$ is an $ ( (n-\ell - m(\delta-1))\alpha \times n\alpha  )$ matrix.  As a result, $H$ is an $((n-\ell)\alpha \times n\alpha)$ matrix. Note that under optimality, $n - \ell \geq m(\delta - 1)$ (since $d_{\min} \geq \delta$ for a code with $(r,\delta)$ information locality).

Let $\mathcal{C}_0^{\perp}$, $\mathcal{C}^{\perp}$ denote the dual codes of $\mathcal{C}_0$, $\mathcal{C}$ respectively. Thus   $\mathcal{C}_0^{\perp}, \mathcal{C}^{\perp}$ are the row spaces of $H_0,H$ respectively.  It is clear that $\mathcal{C}$ has $(r,\delta)$ locality with each local code being a sub-code of an MSR code.

Let $S \subseteq [n]$ such that $|S| = \nu$ be referred to as  a $\nu$-core of $\mathcal{C}_0^{\perp}$ if $\forall \ \mathbf{c}_0 \in \mathcal{C}_0^{\perp}, \ \text{supp}(\mathbf{c}_0) \nsubseteq S$. We will now show that if the matrix $H_1$ is selected in such a way that any $S$ which is an $\ell$-core of $\mathcal{C}_0^{\perp}$ is also an $\ell$-core of $\mathcal{C}^{\perp}$, then the minimum distance of $\mathcal{C}$ will be given by \eqref{eq:msr_allsymbol_thm}.   We will subsequently show by appealing to Lemma \ref{lem:null}, that it is always possible to pick $H_1$ such that the above condition is met.

 Let $S_i, i \in [m]$ denote the disjoint supports of the $m$ local codes of $\mathcal{C}_0$ (and hence of $\mathcal{C}$ as well). Clearly, $S$ is an $\ell$-core of $\mathcal{C}_0^{\perp}$ if and only if $|S \cap S_i| \leq r, \ \forall \ i \in [m]$. We note that any $\Gamma \subset S_i, \ |\Gamma| \leq r$ can be extended to an $\ell$-core of $\mathcal{C}_0^{\perp}$. It is also clear that the code $\mathcal{C}^{\perp}$ when shortened to $S_i$ has $H_L$ as a submatrix of its generator matrix.

 Let the matrix $H_1$ be such that any $S$ which is an $\ell$-core of $\mathcal{C}_0^{\perp}$ is also an $\ell$-core of $\mathcal{C}^{\perp}$.  This has the following implication:  the code $\mathcal{C}^{\perp}$ when shortened to $S_i$ has generator matrix $H_L$,   for otherwise, if it were to contain one or more additional rows, we would be able to find a code-word $\mathbf{c'}$ of $\mathcal{C}^{\perp}$ that does not belong to $\mathcal{C}_0^{\perp}$ and that is supported on  $\Gamma \subset S_i, \ |\Gamma| \leq r$.  But this would then contradict the assumption that any $S$ which is an $\ell$-core of $\mathcal{C}_0^{\perp}$ is also an $\ell$-core of $\mathcal{C}^{\perp}$, as $\ell\geq r$.  We thus conclude that $\mathcal{C}|_{S_i}$ is an MSR code whose parity check matrix is given by $H_L$. Thus the code $\mathcal{C}$ is an MSR-local code with $(r,\delta)$ all-symbol locality (we were able to assert earlier only that each local code is a sub-code of an MSR code).

Let $G$ denote the generator matrix of $\mathcal{C}$. Note that if $S$ is any $\ell$-core of $\mathcal{C}_0^{\perp}$ (and hence of $\mathcal{C}^{\perp}$ as well), it must be that $\text{Rank}\left(G|_S\right) = \ell\alpha = K$, because this means that there can not be any dependencies in the $\ell$ thick columns of $G|_S$.

We continue under the assumption as above, that any $S$ which is an $\ell$-core of $\mathcal{C}_0^{\perp}$ is also an $\ell$-core of $\mathcal{C}^{\perp}$.  Next, let $T \subseteq [n]$ of size $|T| \geq \ell$ be such that $\text{Rank}\left(G|_T\right) < K$ (Such a $T$ can be constructed by ensuring that $ |T \cap S_i| \geq (r+1)$ for some $i$). Clearly, $T$ does not contain any $\ell$-core of $\mathcal{C}^{\perp}$ (and hence does not contain any $\ell$-core of $\mathcal{C}_0^{\perp}$ as well), which implies that at least for some $i \in [m], |T \cap S_i| \geq (r+1)$. Let the integers $b_{j}, 1 \leq j \leq (\delta -1)$ be defined as follows:
\begin{eqnarray}
b_{j} & \triangleq & \left|\left\{ i \in [m] : \  |S_i \cap T| = r+j \right\}\right|.
\end{eqnarray}
Also, define the sets $Q_i, i \in [m]$ as follows:
\begin{eqnarray}
Q_i & = & \left\{\begin{array}{cc} T \cap S_i, & \text{if} \ |T \cap S_i| \leq r  \\
                                 (\text{any}) \ Q_i' \subset T \cap S_i, \text{ s.t.} \ |Q_i'| = r, &  \ \text{if} \ |T \cap S_i| > r . \end{array}  \right.
\end{eqnarray}
Also, let $Q = \cup_{i = 1}^{m}Q_i$. Note that $|Q| = |T|-\sum_{j = 1}^{\delta -1}j b_{j}$. Clearly, the set $Q$ is a $|Q|$-core of $\mathcal{C}^{\perp}$ and
\begin{eqnarray} \label{eq:msr_allsymbol_proof1}
 |T|-(\delta-1)\left(\sum_{j=1}^{\delta-1} b_j \right) \ \leq |T|-\sum_{j = 1}^{\delta -1}j b_{j} & = & |Q| \  \leq \ell - 1.
\end{eqnarray}
Next, let $M = \{ i \in [m] : |S_i \cap T| \geq r + 1 \}$and note that $|M|=(\sum_{j=1}^{\delta-1} b_j)$.  If we pick $r$ elements from each set $S_i, i \in M$, we will then obtain an $(|M|r)$-core. Thus we have that
\begin{eqnarray} \label{eq:msr_allsymbol_proof2}
 \sum_{j=1}^{\delta-1} b_j \leq \left\lfloor \frac{\ell-1}{r} \right\rfloor = \left \lceil \frac{\ell}{r} \right \rceil -1.
\end{eqnarray}
Combining \eqref{eq:msr_allsymbol_proof1} and \eqref{eq:msr_allsymbol_proof2}, we get that
\begin{eqnarray}
 |T| &\leq& \ell - 1 + (\delta-1)\left(\left\lceil\frac{\ell}{r}\right\rceil -1\right) \\
 &=& \frac{K}{\alpha} - 1 + (\delta-1)\left(\left\lceil\frac{K}{r\alpha}\right\rceil -1\right).
\end{eqnarray}
It follows from Lemma \ref{lem:fact_dmin} that
\bean
d_{\min} & \geq & n- \left( \frac{K}{\alpha} - 1 + (\delta-1)\left(\left\lceil\frac{K}{r\alpha}\right\rceil -1\right) \right),
\eean
and then from the K-bound \eqref{eq:URA_MSR_bound} that the code $\mathcal{C}$ has minimum distance equal to
\begin{eqnarray}
 d_{\min} & = & n - \frac{K}{\alpha} + 1 + (\delta-1)\left(\left\lceil\frac{K}{r\alpha}\right\rceil -1\right).
\end{eqnarray}

It  remains to be proved that one can pick a matrix $H_1$ such that any $S$ which is an $\ell$-core of $\mathcal{C}_0^{\perp}$ is also an $\ell$-core of $\mathcal{C}^{\perp}$. Towards this, consider a set $S$ such that $|S| = \ell$ and let $S^{c}$ denote the set $[n]\backslash S$. Note that $S$ is an $\ell$-core of $\mathcal{C}^{\perp}$ if and only if the square matrix $H|_{S^{c}}$ is full rank, i.e., $\det \left(H|_{S^{c}}\right) \neq 0$. Now, we need to pick $H_1$ such that for all $S\subset [n]$,  $S$  an $\ell$-core of $\mathcal{C}_0^{\perp}$,  $\det (H|_{S^C})\neq 0$.   This can be done using a similar technique as in Theorem \ref{thm:sum_msr_optimality} where we used Lemma \ref{lem:null} to pick the matrix $Q$ in \eqref{eq:last_label_for_Q}. One can show that there exists a matrix $H_1$ such that $H|_{S^{c}}$ is full rank for all $S$, $\ell$-core of $\mathcal{C}_0^{\perp}$. Here also we take take $H_1$ to be a matrix of indeterminates and each $\ell$-core of $\mathcal{C}_0^{\perp}$, gives us a determinant and hence a non-zero polynomial whose evaluation must be nonzero. Noting that $\mathcal{C}_0^{\perp}$ has at-most ${n \choose \ell}$ $\ell$- cores and  using the Combinatorial Nullstellensatz of Lemma \ref{lem:null}  as we did in  Theorem \ref{thm:sum_msr_optimality}, we conclude that such $H_1$ can be picked if the field size  $q > {n \choose \ell}$.

\section{Proof of Theorem \ref{thm:mbr_allsymbol_existence}} \label{app:mbr_allsymbol_existence_proof}

All claims in the theorem are clear with the exception of the claim concerning the minimum distance.
 Since $K_L \mid K$, an upper bound on $d_{\min}$ from \eqref{eq:URA_MBR_bound} is given by
\begin{eqnarray}
d_{\min} & \leq & n-\frac{K}{K_L}  r +1  - \left( \frac{K}{K_L} -1 \right) (\delta-1) \label{eq:mbr_allsymbol_proof1} \\
& = & m(r+\delta-1)-\ell r+1 -(\ell-1)(\delta-1), \\
& = & (m-\ell)n_L + \delta. \label{eq:MBR_all_symbol equality}
\end{eqnarray}
It suffices to show that any pattern of $\delta + (m-\ell)n_L - 1$ erasures can be corrected by the code. Towards this, we note that the scalar code $\mathcal{A}$ employed in Construction \ref{constr:mbr_allsymbol_existence} has minimum distance given by
\bean
D_{\min} & = &(m-\ell)N_L+\Delta_L  \ = \ (m-\ell){n_L \choose 2} + {\delta-1 \choose 2} +1 .
\eean

Now we argue that when any pattern of $(m-\ell)n_L + \delta-1$ vector code symbols are erased, this leads to the erasure of at most $D_{\min}-1$ scalar code symbols. This would imply that the code, $\mathcal{A}$, can recover from this many erasures and hence, so can the vector code $\mathcal{C}$.

As in the proof of Theorem \ref{thm:mbr_info_locality}, for a given pattern of  $(m-\ell)n_L + \delta-1$ vector code symbol erasure, let $\gamma_i, 1\leq i\leq m$ be the number of code-word symbols erased from the $i^{\text{th}}$ local code among these symbols. Note that $0 \leq \gamma_i \leq \delta -1 < n_L,\ 1 \leq i \leq m$ and $\sum_{i=1}^{m}\gamma_i=(m-\ell)n_L+\delta-1$. Thus the number of scalar code symbols lost by the code in this pattern of erasures, $L$ has to be that
\bean
L=\sum_{i=1}^{m} {\gamma_i \choose 2} \leq (m-\ell){n_L \choose 2} + {\delta -1 \choose 2} = D_{\min}-1,
\eean
where we have used the fact that ${a \choose 2} +{b\choose 2} \leq {a+b \choose 2}$. The result follows.

\section{Proof of Theorem \ref{thm:kappa_bound}} \label{app:kappa_bound_proof}

We will make use of the following two facts and Lemma \ref{lem:fact_dmin} to prove the theorem. Their proofs are straightforward and are hence omitted.

\vspace{0.1in}

\begin{lem} \label{lem:infosets_1}
 Consider two sets $S_1$ and $S_2$ such that $S_1 \subset S_2 \subseteq [n]$, and
\begin{eqnarray}
 \text{q-dim}\left(\mathcal{C}|_{S_2}\right) - \text{q-dim}\left(\mathcal{C}|_{S_1}\right) & \triangleq & \Delta \nu \ > \ 0.
\end{eqnarray}
Then, if $\mathcal{I}$ is any minimum cardinality information set for $\mathcal{C}|_{S_2}$, then it must true that $\left|\mathcal{I} \cap \left(S_2\backslash S_1
\right)\right| \geq \Delta \nu$.
\end{lem}

\vspace{0.1in}

\begin{lem} \label{lem:infosets_2}

Consider two sets $S_1$ and $S_2$, such that $S_1 \subseteq S_2 \subseteq [n]$, and $\text{rank}\left(G|_{S_1}\right) =
\text{rank}\left(G|_{S_2}\right) = K$. Then $\text{q-dim}\left(\mathcal{C}|_{S_1} \right) \geq
\text{q-dim}\left(\mathcal{C}|_{S_2} \right)$.
\end{lem}

\vspace{0.1in}

We will assume in the proof, that we are given an $[n,K,d_{\min},\alpha]$ code $\mathcal{C}$ which  has $(r, \delta)$ information locality. We will construct a set $T \subseteq [n]$ such that $\text{rank}\left(G|_T\right) < K$ using Algorithm \ref{alg:URA_bound}(the same Algorithm which is used in the proof of Theorem \ref{thm:URA_bound}), and then apply Lemma \ref{lem:fact_dmin} to get the required result.

Let the algorithm exit after $J$ iterations, i.e., $j = J$ when the algorithm exits.  Let $S'' = T_{J-1}
\cup S_i$, where $i$ is the index picked in the $J^{\text{th}}$ iteration.  Note that necessarily, $\text{Rank}\left(G|_{S''} \right) = K$.

Let $\mathcal{I}'$ denote a minimum cardinality information set for $\mathcal{C}|_{S''}$. Clearly, it must be true that
\begin{eqnarray} \label{eq:localitythmproof_temp1}
 J & \geq & \left\lceil \frac{|\mathcal{I}'|}{r} \right\rceil.
\end{eqnarray}
Next, for $j \in [J]$, let
\begin{eqnarray}
s_j & = & |T_j| - |T_{j-1}| \ , \nonumber \\
\nu_j & = & \text{q-dim}(\mathcal{C}|_{T_j}) - \text{q-dim}(\mathcal{C}|_{T_{j-1}}). \label{eq:localitythmproof_temp6}
\end{eqnarray}
We claim that for $j \in [J-1]$,
\begin{eqnarray} \label{eq:locality_proof_condition1}
s_j & \geq & \nu_j + (\delta - 1).
\end{eqnarray}
To see this, first note whenever we pick $i \in \mathcal{L}$ such that $V_i \nsubseteq \sum_{\ell \in T_j}W_{\ell}$, since
$d_{\text{min}}\left(\mathcal{C}|_{S_i}\right) \geq \delta - 1$, it must be true that $s_j  \geq  1 + (\delta - 1)=\delta$.
Also, whenever $\nu_j > 0$, Lemma \ref{lem:infosets_1} implies that \eqref{eq:locality_proof_condition1} must be true
and thus we see that \eqref{eq:locality_proof_condition1} is true always.  We also have that
\begin{eqnarray}
 s_J & \geq & \nu_J. \label{eq:locality_proof_condition2}
\end{eqnarray}
Summing up, we obtain
that
\begin{eqnarray}
 |T_J| \ = \  \sum_{1 = 1}^{J}s_j & \geq &  \sum_{1 = 1}^{J}\nu_j \ + \ (J-1)(\delta - 1) \\
 & \geq & \sum_{1 = 1}^{J}\nu_j \ + \ \left(\left\lceil \frac{|\mathcal{I}'|}{r} \right\rceil-1\right)(\delta - 1)
\label{eq:locality_proof_temp3}\\
& \geq & (|\mathcal{I}'| - 1) \ + \ \left(\left\lceil \frac{|\mathcal{I}'|}{r}\right\rceil-1\right)(\delta - 1),
\label{eq:locality_proof_temp4}
\end{eqnarray}
where \eqref{eq:locality_proof_temp3} follows from \eqref{eq:localitythmproof_temp1} and \eqref{eq:locality_proof_temp4}
follows by noting that
\begin{eqnarray} \label{eq:localitythmproof_temp5}
\sum_{1 = 1}^{J}\nu_j & = & \text{q-dim}(\mathcal{C}|_{T_J}) \geq |\mathcal{I}'| - 1,
\end{eqnarray}
which is because of the  maximality of $S_{\text{end}}$ in $S_i$ (i.e., even adding one more element of $S_i$ to $S_{\text{end}}$
in step 8 of Algorithm 1 would result in an accumulated rank of $K$ and thus $ \text{q-dim}(\mathcal{C}|_{T_J}) \ \geq \ |\mathcal{I}'| - 1$.)  Now,
since $\text{rank}(G|_{T_J}) < K$, Lemma \ref{lem:fact_dmin} can be applied to give that
\begin{eqnarray}
d_{\text{min}} & \leq & n - |\mathcal{I}'| + 1 - \left(\left \lceil \frac{|\mathcal{I}'|}{r}\right \rceil - 1\right)(\delta
- 1), \label{eq:locality_proof_I''} \\
 & \leq & n - |\mathcal{I}_0| + 1 - \left(\left \lceil \frac{|\mathcal{I}_0|}{r}\right \rceil - 1\right)(\delta
- 1). \label{eq:locality_proof_temp_I_0}
\end{eqnarray}
where, as $S'' \subseteq \cup_{i \in \mathcal{L}}S_i$ and thus Lemma \ref{lem:infosets_2}, we have that $|\mathcal{I}'| \geq |\mathcal{I}_0|$ leading to \ref{eq:locality_proof_temp_I_0}, where $\mathcal{I}_0$ is as defined in Theorem \ref{thm:kappa_bound}.  The bound in \eqref{eq:locality_looser_kappa}, then follows from Lemma~\ref{lem:infosets_2}.    Further, since $\kappa \geq \left\lceil \frac{K}{\alpha}\right\rceil$, \eqref{eq:locality_looser_kappa} can be upper bounded as
follows:
\bean
 d_{\text{min}} & \leq & n - \kappa + 1 - \left(\left \lceil \frac{\kappa}{r}\right \rceil - 1\right)(\delta - 1) \\
& \leq & n - \left\lceil \frac{K}{\alpha}\right\rceil + 1 - \left(\left \lceil \frac{1}{r} \left\lceil
\frac{K}{\alpha}\right\rceil\right \rceil - 1\right)(\delta - 1) \\
& = & n - \left\lceil \frac{K}{\alpha}\right\rceil + 1 - \left(\left \lceil \frac{K}{r\alpha}\right\rceil- 1\right)(\delta
- 1),
\eean
where the last equation follows since $\left \lceil \frac{1}{r} \left\lceil
\frac{K}{\alpha}\right\rceil\right \rceil = \left \lceil \frac{K}{r\alpha}\right\rceil$.
This concludes the proof of the theorem.

\end{document}